\documentclass{article}

\usepackage[utf8]{inputenc}
\usepackage[T1]{fontenc}

\usepackage[english]{babel}

\usepackage{amsmath}
\usepackage{amssymb}
\usepackage{amsthm}
\usepackage{enumerate}
\usepackage{appendix}
\usepackage{geometry}
\usepackage{eurosym}
\usepackage{hyperref}
\hypersetup{
    colorlinks,
    citecolor=blue,
    filecolor=blue,
    linkcolor=blue,
    urlcolor=blue
}
\usepackage{todonotes}
\usepackage{mathrsfs}
 \usepackage{dsfont}

 \usepackage{rotating}
\usepackage[round]{natbib}
\setcitestyle{round}

\usepackage{etoolbox}

\newcounter{bibvis}

\makeatletter
\pretocmd{\thebibliography}{\setcounter{bibvis}{0}}{}{}

\let\bibvis@orig@lbibitem\@lbibitem
\let\bibvis@orig@bibitem\@bibitem

\renewcommand\@lbibitem[2][]{%
  \bibvis@orig@lbibitem[#1]{#2}%
  \stepcounter{bibvis}\text{[\arabic{bibvis}]}\hspace{0.6em}\ignorespaces
}
\renewcommand\@bibitem[1]{%
  \bibvis@orig@bibitem{#1}%
  \stepcounter{bibvis}\text{[\arabic{bibvis}]}\hspace{0.6em}\ignorespaces
}
\makeatother

\usepackage{empheq}

\def\be{\begin{eqnarray}}
\def\ee{\end{eqnarray}}

\def\b*{\begin{eqnarray*}}
\def\e*{\end{eqnarray*}}

\newtheorem{Definition}{Definition}[part]
\newtheorem{Proposition}{Proposition}[part]

\newcommand{\ba}{\begin{array}}
\newcommand{\ea}{\end{array}}
\newcommand{\ben}{\begin{equation*}} 
\newcommand{\een}{\end{equation*}}
\newcommand{\bea}{\begin{eqnarray}}
 \newcommand{\eea}{\end{eqnarray}}
\newcommand{\bean}{\begin{eqnarray*}} 
\newcommand{\eean}{\end{eqnarray*}}
\newcommand{\bel}{\begin{align}} 
\newcommand{\eel}{\end{align}}
\newcommand{\beln}{\begin{align*}} 
\newcommand{\eeln}{\end{align*}}
\newcommand{\bit}{\begin{itemize}}
\newcommand{\eit}{\end{itemize}}

\makeatletter \@addtoreset{equation}{section}

\@addtoreset{Definition}{section}

\@addtoreset{Theorem}{section}

\@addtoreset{Proposition}{section}

\@addtoreset{Property}{section}

\@addtoreset{Assumption}{section}

\@addtoreset{Corollary}{section}

\@addtoreset{Lemma}{section}

\@addtoreset{Remark}{section}

\@addtoreset{Example}{section}

\@addtoreset{Condition}{section}

\def \E{\mathbb{E}}

\def \H{\mathbb{H}}
\def \L{\mathbb{L}}
\def \M{\mathbb{M}}
\def \N{\mathbb{N}}
\def \P{\mathbb{P}}
\def \Q{\mathbb{Q}}
\def \R{\mathbb{R}}

\def \Z{\mathbb{Z}}

\def \G{\mathbb{G}}

\def\Cov{\mathbb{C}\text{ov}}
\def\Var{\mathbb{V}\text{\rm ar}}

\def\Ec{{\cal E}}

\def\Gc{{\cal G}}

\def\Kc{{\cal K}}
\def\Lc{{\cal L}}

\def\Tr#1{{\rm Tr}\left[#1\right]}

\def\={\;=\;}
\def\.{\;.}

\def\eps{\varepsilon}

\def\1{{\bf 1}}
\def \ep{\hbox{ }\hfill$\Box$}

\def \proof{{\noindent \bf Proof. }}
\def\eps{\epsilon}

 \def\normeL2#1{\left\|{#1}\right\|_{L^2}}
 
\newcommand{\alias}[2]{
\providecommand{#1}{}
\renewcommand{#1}{#2}
}

\alias{\P}{\mathbb{P}}
\alias{\N}{\mathcal{N}}
\alias{\L}{\mathcal{L}}
\alias{\Z}{\mathbb{Z}}
\alias{\Q}{\mathbb{Q}}
\alias{\R}{\mathbb{R}}
\alias{\C}{\mathcal{C}}
\alias{\T}{\mathbb{T}}
\alias{\E}{\mathbb{E}}
\alias{\H}{\mathcal{H}}
\alias{\1}{\mathbf{1}}
\alias{\B}{\mathcal{B}}
\alias{\M}{\mathcal{M}}
\alias{\G}{\mathcal{G}}
\alias{\Y}{Y_{\bullet}}

\newcommand{\nc}{\newcommand}
\nc{\cA}{{\mathcal A}} \nc{\cB}{{\mathcal B}} \nc{\cC}{{\mathcal
C}} \nc{\cD}{{\mathcal D}} \nc{\bbD}{\mathbb{D}}
\nc{\cG}{{\mathcal G}} \nc{\cF}{{\mathcal F}} \nc{\cS}{{\mathcal
S}} \nc{\cU}{{\mathcal U}} \nc{\cH}{{\mathcal H}}
\nc{\cK}{{\mathcal K}} \nc{\cM}{{\mathcal M}} \nc{\cO}{{\mathcal
O}} \nc{\cP}{{\mathcal P}} \nc{\bbE}{\mathbb{E}}
\nc{\bbEP}{\mathbb{E}_{\mathbb{P}}}\nc{\bbL}{\mathbb{L}}
\nc{\bbP}{\mathbb{P}} \nc{\bbQ}{\mathbb{Q}} \nc{\del}{\partial}
\nc{\Om}{\Omega} \nc{\om}{\omega} \nc{\bbR}{\mathbb{R}}
\nc{\bbC}{\mathbb{C}} \nc{\bfr}{\begin{flushright}}
\nc{\efr}{\end{flushright}} \nc{\dXt}{\delta q_{t}}
\nc{\dXs}{\delta q_{s}} \nc{\bs}{\blacksquare} \nc{\dX}{\delta q}
\nc{\dY}{\Delta Y}
\nc{\dnkx}{\left(X(T^{n}_{k})-X(T^{n}_{k-1})\right)}
\nc{\esssup}{\mathrm{ess}\mbox{ }\mathrm{sup}}

\nc{\essinf}{\mathrm{ess}\mbox{ } \mathrm{inf}}
\nc{\dhats}{\widehat{\delta_s}}

\nc{\chf}{\mbox{$\mathbf1$}}
\nc{\ind}{\mathds{1}}

\begingroup\expandafter\expandafter\expandafter\endgroup
\expandafter\ifx\csname pdfsuppresswarningpagegroup\endcsname\relax
\else
\fi

\nc{\zfb}{ z_{\mbox{\footnotesize\rm\sc fb}} }
\nc{\zsb}{ z_{\mbox{\footnotesize\rm\sc sb}} }
\nc{\thetah}{ \theta^{\mbox{\footnotesize\rm\sc h}} }
\nc{\zsbh}{ z_{\mbox{\footnotesize\rm\sc sb}}^{\mbox{\footnotesize\rm\sc h}} }
\nc{\psbh}{ p_{\mbox{\footnotesize\rm\sc sb}}^{\mbox{\footnotesize\rm\sc h}} }
\nc{\psb}{ p_{\mbox{\footnotesize\rm\sc sb}} }

\nc{\thetasbh}{ \theta_{\mbox{\footnotesize\rm\sc sb}}^{\mbox{\footnotesize\rm\sc h}} }
\nc{\thetafbh}{ \theta_{\mbox{\footnotesize\rm\sc fb}}^{\mbox{\footnotesize\rm\sc h}} }

\nc{\afb}{ a_{\mbox{\footnotesize\rm\sc fb}} }
\nc{\asb}{ a_{\mbox{\footnotesize\rm\sc sb}} }
\nc{\asbh}{ a_{\mbox{\footnotesize\rm\sc sb}}^{\mbox{\footnotesize\rm\sc h}} }
\nc{\afbh}{ a_{\mbox{\footnotesize\rm\sc fb}}^{\mbox{\footnotesize\rm\sc h}} }
 
\nc{\xisb}{ \xi_{\mbox{\footnotesize\rm\sc sb}} }
\nc{\xifb}{ \xi_{\mbox{\footnotesize\rm\sc fb}} }
\nc{\xifbh}{ \xi_{\mbox{\footnotesize\rm\sc fb}}^{\mbox{\footnotesize\rm\sc h}} }
\nc{\xisbm}{ \xi_{\mbox{\footnotesize\rm\sc sb$_m$}} }
\nc{\xisbf}{ \xi_{\mbox{\footnotesize\rm\sc sb}}^{\mbox{\footnotesize\rm\sc f}} }
\nc{\xifbf}{ \xi_{\mbox{\footnotesize\rm\sc fb}}^{\mbox{\footnotesize\rm\sc f}} }
\nc{\xisbmf}{ \xi_{\mbox{\footnotesize\rm\sc sb$_m$}}^{\mbox{\footnotesize\rm\sc f}} }
\nc{\xisbv}{ \xi_{\mbox{\footnotesize\rm\sc sb}}^{\mbox{\footnotesize\rm\sc v}} }
\nc{\xifbv}{ \xi_{\mbox{\footnotesize\rm\sc fb}}^{\mbox{\footnotesize\rm\sc v}} }
\nc{\xisbmv}{ \xi_{\mbox{\footnotesize\rm\sc sb$_m$}}^{\mbox{\footnotesize\rm\sc v}} }

\nc{\Vo}{ V^{\mbox{\footnotesize\rm\sc 0}}  }
\nc{\Vh}{ V^{\mbox{\footnotesize\rm\sc h}}  }
\nc{\Vsb}{ V_{\mbox{\footnotesize\rm\sc sb}} }
\nc{\Vsbh}{ V_{\mbox{\footnotesize\rm\sc sb}}^{\mbox{\footnotesize\rm\sc h}} }
\nc{\Vfbh}{ V^{\mbox{\footnotesize\rm\sc h}}_{\mbox{\footnotesize\rm\sc fb}} }
\nc{\Vfb}{ V_{\mbox{\footnotesize\rm\sc fb}} }
\nc{\Ufb}{ U_{\mbox{\footnotesize\rm\sc fb}} }

\nc{\Go}{ G^{\mbox{\footnotesize\rm\sc 0}} }
\nc{\Gh}{ G^{\mbox{\footnotesize\rm\sc h}} }
\nc{\Gsb}{ G_{\mbox{\footnotesize\rm\sc sb}} }
\nc{\Gsbh}{ G^{\mbox{\footnotesize\rm\sc h}}_{\mbox{\footnotesize\rm\sc sb}} }
\nc{\Gfb}{ G_{\mbox{\footnotesize\rm\sc fb}} }
\nc{\Gfbh}{ G^{\mbox{\footnotesize\rm\sc h}}_{\mbox{\footnotesize\rm\sc fb}} }

\nc{\mo}{ m^{\mbox{\footnotesize\rm\sc 0}} }
\nc{\mh}{ m^{\mbox{\footnotesize\rm\sc h}} }
\nc{\msb}{ m_{\mbox{\footnotesize\rm\sc sb}} }
\nc{\msbh}{ m^{\mbox{\footnotesize\rm\sc h}}_{\mbox{\footnotesize\rm\sc sb}} }
\nc{\mfb}{ m_{\mbox{\footnotesize\rm\sc fb}} }
\nc{\mfbh}{ m^{\mbox{\footnotesize\rm\sc h}}_{\mbox{\footnotesize\rm\sc fb}} }

\nc{\epsilonsb}{ \epsilon^{\mbox{\footnotesize\rm\sc sb}} }
\nc{\epsilonfb}{ \epsilon^{\mbox{\footnotesize\rm\sc fb}} }
\nc{\Hv}{ H_{\rm v} }
\nc{\Hm}{ H_{\rm m} }

\nc{\hs}{\vspace{5mm}}

\nc{\bga}{\bar \gamma_a}
\newcommand{\tp}{\intercal}

\nc{\mcdot}{\!\cdot\!}
\nc{\mddot}{\!:\!}

\graphicspath{{./figs/}}

\begin{document}

\title{Forward Hedging Reshapes Incentive Provision}
\author{René Aïd\thanks{Université Paris-Dauphine -- PSL Research University, Department of Economics, UMR CNRS 8007. René Aïd thanks the financial support of the {\em Finance and Sustainable Development EDF-CA CIB Chair}, the {\em Finance for Energy Market} Research Initiative, the French ANR PEPR Math-Vives project MIRTE ANR-23-EXMA-0011 and the UMP6P-Africa Business School.}  \quad Nizar Touzi\thanks{New York University, Tandon School of Engineering, Finance and Risk Engineering Department. Nizar Touzi is partially supported by the {\em Finance and Sustainable Development EDF-CA CIB Chair}, the {\em Finance for Energy Market} Research Initiative and and the National Science Foundation under grant No. DMS-2508581. } \quad Stéphane Villeneuve\thanks{Toulouse School of Economics. Villeneuve acknowledges funding from FDR-SCOR {\it Chaire Marché des risques et création de valeur}, AFOSR-22IOE016 and ANR under grant ANR-17-EUR-0010 (Investissements d'Avenir program). Many thanks for insightful feedback to participants at seminars  at Collegio Carlo Alberto, University of Lugano and Toulouse School of Economics.}}

\maketitle

\begin{abstract}

We study how forward hedging reshapes incentive provision inside the firm. We consider a risk-averse producer facing demand and production risk that can either operate in-house or delegate production to a risk-averse agent under moral hazard, while hedging output in a competitive forward market with a rational market maker. Within a tractable continuous-time CARA framework, we jointly characterize optimal production, compensation, and static hedging in equilibrium.

Delegation and external hedging are partial substitutes because both create value through risk sharing. Delegation can increase firm value even when the agent uses the same technology and is more risk averse than the principal, while access to forward hedging reduces the need to provide incentives through risk exposure. This mechanism delivers two main results. First, the principal hedges less under delegation than under in-house production. Second, this lower hedging demand under delegation raises the equilibrium forward price relative to the integrated benchmark.

In the constant-demand case, we show that access to hedging lowers the agent’s expected
compensation under delegation. Numerical results indicate that this mechanism remains robust
in the presence of demand uncertainty. More broadly, our results show that external risk transfer
through financial markets feeds back into internal organizational design.

\end{abstract}

\clearpage

\section{Introduction}

Firms use financial markets to transfer risk outside the firm, but doing so should also change how much risk must be allocated inside the firm. This interaction is especially important when production is delegated to a manager or specialist whose effort is unobservable. In that case, compensation must both provide incentives and share risk. A standard implication of principal-agent theory is that stronger pay-performance sensitivity improves incentives but also exposes the agent to more firm-specific risk \cite{Holmstrom,Salanie,LaffontMartimort}. At the same time, the corporate hedging literature emphasizes that firms use forward and futures markets to reduce the volatility of operating cash flows \cite{SmithStulz,Froot93}. This paper studies the link between these two margins. Our central question is the following: when the principal can hedge output price risk in a forward market, how does this reshape optimal incentive provision inside the firm, and what are the resulting implications for equilibrium forward prices?

This question is distinct from the managerial hedging problem emphasized in the executive-compensation literature. That literature studies how managers undo or attenuate incentive schemes by trading financial claims on their own account, with important implications for contract design and pay-performance sensitivity \cite{Ofek,Gao2010,Bebchuk02,GarveyMilbourn03}. We instead focus on \emph{principal hedging}. The issue is not whether the agent privately neutralizes the contract, but whether the principal's access to financial markets changes the contract that is optimal \emph{ex ante}. Once external hedging is available, the principal can offload part of firm's risks to the market, which reduces the need to load that same risk onto the agent through compensation. External hedging and internal incentive provision are therefore partial substitutes.

More specifically, we consider a risk-averse commodity producer with market power that faces both demand risk and production risk and can hedge output in a competitive futures market where liquidity is supplied by a market maker. The producer can either retain production in-house or delegate it to a risk-averse agent whose effort is unobservable. Delegation therefore creates the usual incentive-insurance trade-off: making compensation more sensitive to performance improves incentives, but also exposes the agent to more risk. In parallel, access to the futures market gives the principal an additional instrument to transfer part of output price risk outside the firm. Our objective is to understand how these two margins interact, and how the availability of external hedging reshapes optimal production, compensation, and equilibrium futures prices.

From a methodological viewpoint, our framework builds on continuous-time agency models in the spirit of \cite{HM}; see \cite{Cvitanic12} for a broad overview of this literature. Relative to the classical setting, however, our environment introduces an endogenous production decision and separates demand shocks from production shocks, so that only part of revenue risk can be hedged in financial markets. This distinction is central for our purposes, because it generates a non-trivial interaction between real decisions, contract design, and hedging demand. The principal chooses the organizational form, the production policy, the compensation scheme, and the futures position, subject to the agent's incentive compatibility and participation constraints. This tractable CARA framework allows us to compare in-house production and delegated production within a unified equilibrium setting, and to isolate how external risk transfer feeds back into internal incentives.

This agency problem has direct implications for futures-market equilibrium. Because the producer has market power in the spot market, her futures position affects not only risk transfer but also real production incentives and, ultimately, the terminal spot price. Our analysis therefore also connects to the literature on forward trading by firms with market power. As reviewed by \cite{Anderson84}, a monopolist’s ability to hedge may induce a departure from the standard monopoly allocation, since futures positions create incentives to adjust output in ways that affect the profitability of the hedge. In \cite{Newberry84}, for instance, the presence of a futures market leads a risk-averse monopolist to expand equilibrium supply, and hedging demand increases with risk aversion. The seminal contributions of \cite{Allaz92} and \cite{Allaz93} further show that forward trading can soften market power by committing firms to higher output. 

Our contribution is to bring this market-power channel together with delegated production and optimal risk sharing inside the firm. In our setting, the producer’s futures demand depends not only on market power and risk aversion, but also on the way risk is allocated between principal and agent through the compensation scheme. This additional organizational margin is central to the paper. Delegation already provides an internal channel for sharing risk, which reduces the producer’s incentive to hedge externally relative to the in-house benchmark. As a result, the agency problem affects not only contract design but also equilibrium hedging pressure and, through market clearing, the equilibrium forward price. This is the sense in which our analysis differs from prior work focused primarily on the interaction between futures trading and market power alone.

\hs

\noindent{\em Contribution.}
The primary contribution of this paper is to analyze how a principal's access to a forward market reshapes internal incentive provision and equilibrium pricing in a monopolistic production problem under demand uncertainty. We develop a tractable continuous-time CARA framework that combines delegated production under moral hazard with forward-market participation, and that allows us to characterize jointly optimal production, compensation, hedging demand, and equilibrium forward prices under both in-house production and second-best delegation. Our analysis delivers three main economic insights.

\begin{enumerate}[\upshape (i)]
\item First, delegation may remain optimal even when the agent is more risk-averse than the principal, provided that the participation constraint is not too tight. The mechanism is a trade-off between a positive risk-sharing benefit and a negative information-rent component. As the agent's risk aversion rises, the benefit from internal risk sharing declines while the information-rent cost increases, thereby generating an endogenous boundary for profitable delegation.

\item Second, access to the forward market changes the internal allocation of risk within the firm. External hedging and internal incentive provision are partial substitutes: when the principal can transfer price risk to the forward market, she needs to expose the agent less to that same risk through compensation. In the constant-demand case, we show that this weaker insurance role of the contract lowers the agent’s expected compensation under delegation. Our numerical results further indicate that this mechanism remains robust in the presence of demand uncertainty.

\item Third, the firm's production and trading organisation feeds back into hedging demand and equilibrium prices. Because delegation already provides an internal channel for sharing risk, the principal hedges less under delegation than under in-house production. In equilibrium, this lower hedging demand under delegation translates into higher forward prices relative to the integrated benchmark. More generally, the paper identifies internal organizational design as a determinant of hedging pressure and equilibrium forward pricing.
\end{enumerate}

\noindent{\it Related literature.} Our paper relates to four strands of research: corporate risk management and delegated hedging; moral hazard, financial trading, and continuous-time contracting; equilibrium pricing in commodity futures markets; and industrial-organization perspectives on vertical arrangements and forward trading under market power.

First, our paper contributes to the corporate-finance literature on risk management. A large literature studies why firms hedge and how hedging interacts with investment, financing, and governance frictions; classic contributions include \cite{Froot93} and the synthesis in \cite{Stulz96}. Empirical evidence on delegated hedging in commodity-producing firms is provided by \cite{Tufano96}, while \cite{Fehle05} emphasize dynamic implications of corporate risk management. Relative to this literature, our contribution is to embed hedging in a delegated production problem with moral hazard, and to show that a firm's access to futures markets reshapes internal incentive provision. In our framework, external hedging and internal risk sharing are partial substitutes: the ability to transfer price risk outside the firm changes both the structure of compensation and the firm's equilibrium hedging demand.

Second, our paper relates to the literature on moral hazard in environments where agents or principals interact with financial markets. \cite{Kahn95} show that when effort is unobservable, allowing additional trades may weaken incentives and generate market failures. Our analysis complements this insight by focusing on \emph{principal hedging}: rather than studying how the agent offsets compensation risk through side trading, we characterize how the principal's own access to commodity futures markets changes the set of implementable incentive schemes in a delegated production problem. More generally, our framework is connected to the literature on asset pricing under moral hazard and to continuous-time contracting models. From a general-equilibrium perspective, \cite{Ramakrishnan84} and \cite{Ou-Yang05} study how moral hazard affects asset prices and risk premia. We instead work in partial equilibrium and focus on the joint determination of production, compensation, hedging, and equilibrium futures prices when the principal can both trade and influence prices through spot market power. Our tractable CARA specification also places the paper in the broader tradition of continuous-time agency models initiated by Holmström and Milgrom and extended in, among others, \cite{Sannikov08,DeMarzo06,Biais07}.

Third, the paper contributes to the literature on equilibrium pricing in commodity futures markets when hedging demand is endogenous. The intuition that hedging pressure affects futures premia goes back to \cite{Keynes30} and has been formalized and tested in modern settings such as \cite{Hirshleifer88,Bessembinder92}. The relation between futures premia, term structure, and storage-based mechanisms is documented in \cite{Fama87} and rationalized in competitive storage models such as \cite{Deaton92,Deaton96}. Our mechanism is different. We show that hedging pressure depends not only on inventories, storage, or exogenous exposures, but also on the firm's internal organizational design. Because delegation itself provides an internal channel for sharing risk, it reduces the producer's equilibrium demand for futures relative to in-house production, which in turn raises equilibrium forward prices. The paper therefore identifies internal incentives and delegated production as determinants of equilibrium futures pricing.

Fourth, our analysis also connects to work on commitment, firm organization, and forward trading under market power. \cite{Song12} shows that futures markets can serve as commitment devices by transferring resources across states. In our setting, futures markets do not substitute for the principal's role; instead, they interact with delegation and incentive constraints to shape both risk management and equilibrium pricing. Our analysis is also related to dynamic agency models in which informational feedback and strategic manipulation matter, as in \cite{Taub97}, and to recent quantitative work on dynamic moral hazard and firm dynamics such as \cite{Ai23}. It also speaks to recent work on managerial hedging in continuous time, where private hedging and termination considerations distort optimal performance evaluation \cite{Huang23}; by contrast, we study delegated real production under commodity-price risk and emphasize how the principal's access to financial markets feeds back into compensation design. Finally, from an industrial-organization perspective, our paper relates to the economic analysis of commodity trading firms and their vertical arrangements \cite{Pirrong14,Joskow03}, complements \cite{Aid11} on vertical integration as a substitute for hedging, and connects to the literature on forward trading and market power, where forward positions affect incentives and competition \cite{Mahenc04}.

\vspace{4mm}

The paper proceeds as follows: Section 2 lays out the model. Section 3 addresses the agency problem under the condition that the principal lacks access to the futures market. Section 4 then introduces the principal's access to the futures market and analyzes the resulting equilibrium price. Finally, Section 5 leverages the explicit nature of our findings to provide numerical illustrations.

\section{The model}

A risk-averse entrepreneur, hereafter the Principal, produces a perishable good and is subject to three types of risks: an idiosyncratic risk affecting the production, a market risk of falling output prices, to which is added a principal-agent problem of delegation when production requires the expertise of a risk-averse agent. We assume that both the principal and the agent have CARA preferences,
$$
U_P(x) = -  e^{- \eta_p x} \hbox{ and } U_A(x)= -  e^{-\eta_a x}.
$$
In order to state the optimization problem faced by the Principal, we first describe the dynamics of demand and production for the perishable good.\\

{\it Production technology:} The firm in our model produces a good whose sale is exclusively restricted to a prespecified future time, T. A salient illustration of this setup is that of an agricultural producer, such as a farmer, whose output necessarily awaits maturity. Production is inherently subject to uncertainty, which we conceptualize as a composite of two distinct components: a global, systematic risk—exemplified by phenomena like the pervasive effects of climate change—and an idiosyncratic risk, encompassing localized factors such as soil quality or specific micro-climatic conditions.

We make the assumption that the producer has a model for the estimation of his final production $Q_T$ and with a slight abuse of notation, we denote by $(Q_t)_{t\le T}$ the stochastic process that describes the evolution of the final production estimator at time $t$. Production can be managed by the owner (in-house producer) or delegated to an expert, hereafter the agent who shares the same belief regarding the final production estimator. In either case, the average production can be improved if the person in charge exerts an effort, which will be unobservable when the principal delegates to the agent. Without effort, the expected production level evolves as
 $$
 Q_t  =q + \int_0^t\sigma(\rho dW_s^\perp+\sqrt{1-\rho^2}dW_s^0),
 $$
where $(W_t^0)_{t\le T}$ is a standard Brownian motion modeling the idiosyncratic risk and $(W_t^\perp)_t$ is a standard Brownian motion modeling the global risk. We assume that $(W_t^0)_t$ and $(W_t^\perp)_t$ are independent Brownian motions on some probability space $(\Omega,{\mathcal F}_T,\mathbb{P}_0)$, while the correlation coefficient between the production and the global risk is $\rho \in (-1,1)$. The expected production $Q$ is assumed to be observable\footnote{  In the context of agricultural production, for example, apple farming, this estimator can be interpreted as being related to the observable number of apple trees before the blooming season. Historical production data provides relevant information about the average number of apples produced by such an orchard.}.\\
In the two situations (in-house or delegated), the person in charge impacts the average level of production by exerting effort $a$, but also supports a pecuniary quadratic cost
$\frac12 \frac{a^2}{\gamma},$ where $\gamma$ is a positive constant. The parameter $\gamma$ is a measure of the efficiency of the production technology: the higher, the better. Note that we make the assumption that the firm and its delegate share the same technology efficiency.  The uncertainty of production is thus influenced by the agent's decisions through the implementation of actions\footnote{For instance, in agricultural production, the agent needs to protect crops from the weather or increase prevention of sanitary risks.} to increase or decrease the level of production. 
Thus, the production level, when effort $(\alpha_t)_{t\le T}$ is exerted, evolves as
\[
dQ_t  = \alpha_t dt + \sigma(\rho dW_t^\perp+\sqrt{1-\rho^2}dW_t^\alpha) ,
\]
where $\left(W_t^\alpha=W_t^0-\int_0^t \frac{\alpha_s}{\sigma\sqrt{1-\rho^2}}\,ds \right)_t$ is a Brownian motion independent of $(W_t^\perp)_t$ under the equivalent probability measure $\P^\alpha$ defined  by
$$
\frac{d\P^\alpha}{d\P^0}\vert \mathcal{F}_T =\exp\left(-\int_0^T \frac{\alpha_s}{\sigma\sqrt{1-\rho^2}}\,dW_s^0-\frac12\int_0^T \frac{\alpha_s^2}{\sigma^2(1-\rho^2)}\,ds  \right).
$$

{\it Demand Uncertainty:} Because the firm is a monopoly, we assume that the market price will change depending on the quantity of good the monopoly firm supplies to the market at time $T$. The inverse demand (willingness to pay) function $P_T$ is assumed to be of the form
$$
P_T=P(S_T,Q_T)=S_T-\eps Q_T,
$$
where $Q_T$ is the final output level and $S_T$ is the value at time $T$ of a stochastic shift factor. In the following, we call $P_T$ the spot market price at time $T$.
 The parameter $\eps \ge 0$ is the slope of the inverse demand function. The lower the value of $\eps$, the more the firm has a dominant position in its market. We will interpret hereafter $\frac{1}{\eps}$ as a measure of the firm's market power. Several factors can lead to a demand shift, for example: changes in income levels or changes in the size of the market, but we will remain silent in this paper on the reasons for these changes in demand. It is rather assumed that the random fluctuations of $(S_t)_t$ are given by
$$
dS_t = \nu dW_t^\perp,
$$
where $\nu>0$. Unlike production, the demand shift $(S_t)$ is not affected by the effort process.
Finally, we denote by $\pi_T=S_TQ_T-\eps Q_T^2$ the firm revenues.\\

{\it Contractual agreement:} Under delegation, the agent's employment starts at time $0$ and ends when production is fully sold at time $T$. Assume that the principal  can commit to an employment contract which specifies a compensation $\xi$ which is assumed to be restricted to depend only on the observables $(Q, S)$. For a given compensation $\xi$, the objective of the agent is given by:
\begin{equation}\label{agentutility}
V^A(\xi) := \sup_{\alpha} J_{\rm A}(\alpha;\xi),
\end{equation}
where
\begin{align}
J_{\rm A}(\alpha;\xi) := \E^{\P^{\alpha}} \Big[  U_A\Big( \xi - \int_0^T  \frac{\alpha_t^2}{2\gamma}  dt \Big) \Big].
\end{align}
To alleviate the notation, we will define the effort cost function by $c(a) =  \frac{a^2}{2\gamma}$.

We will say that a contract $\xi$ is {\it incentive compatible}, if there exists an effort policy $\alpha^\star$ that maximises \eqref{agentutility}. Hence, for any incentive-compatible contract $\xi$, we define the set of best replies for the agent as
\[
\mathcal A^\star(\xi):=\bigg\{\alpha^\star :V^A(\xi) := J_{\rm A}(\alpha^\star;\xi)]\bigg\}.
\]

{\it Hedging:} To mitigate the demand and market fluctuations, we finally assume that the firm owner has the possibility to sell her production on a forward market at time $0$ at a forward price $F$. The financial profit and loss of a forward position is $H_T=\theta(P_T-F)$, where $\theta$ is the quantity of goods that the producer buys (resp. sells) forward if positive (resp. if negative).\\

We now turn our attention to the individuals who provide the counterparty in the forward market. Following the literature on forward commodity market equilibrium without market power (see \cite{Ekeland19} and the references therein) or in the presence of market power  (see \cite{Newberry84}, \cite{Hughes97}), we assume the existence of a rational CARA market-maker who provides liquidity in the forward market. One considers that the market maker maximizes the expected value of his CARA utility function (with risk aversion parameter $m>0$). Given a forward price $F$, the rational market maker anticipates the producer optimal demand $\theta_{\rm m}^*(F)$ which, in turn, determines the probability distribution of the pair $(S_T,Q_T)$ through the agent's best reply $\alpha^\ast(F)$. We will use notation $\P^F$ for $\P^{\alpha^\ast(F)}$ to highlight the influence of the forward price $F$ of the probability distribution induced by the agent's optimal effort $\alpha^\ast(F)$. Hence, the market maker faces the following optimization problem

\begin{align*}
&\sup_{\theta \in \R} \E^{F}\Big[ U_m\Big( \theta(P(S_T,Q_T) - F) \Big)\Big], 
\end{align*}
to determine the optimal demand $\theta_{\rm m}^*(F)$ for a given forward price $F$.

\hs

{\it Principal Problem:} The firm owner's income at time $T$ is the sum of the production income $\pi_T$ and financial income $H_T$. The risk management instruments at the disposal of the principal are the level of compensation $\xi$ to mitigate the agency problem and the level of production $\theta$ sold on the futures market to mitigate the market risk. The agent has an outside opportunity that he values with a certainty equivalent utility of $R_0$. We suppose
that both agents agree on a contract that delivers at least  expected utility of $R_0$ to the agent. Therefore, the principal's objective is 
\begin{equation}\label{Principal:PB}
\Vsbh := \sup_{\theta,\xi} J_{\rm P}(\theta,\xi),
\end{equation}
where
\begin{align*}
J_{\rm P}(\theta,\xi)&:= \E^{\P^{a^\star}} \Big[U_P( \pi_T + H_T - \xi )  \Big]
 =\E^{\P^{a^\star}} \Big[U_P( \pi_T +\theta(P_T-F)- \xi  \Big]
\end{align*}
together with the participation constraint of the agent $V^A(\xi) \geq R_0 = U_A(r_0)$. Assume, for a while, there is an optimal demand for hedging $\theta_{\rm p}(F)$ for the Principal to formally state the equilibrium definition we study in this paper.

\begin{Definition}
    ({\it Forward market equilibrium definition})\\
    A forward price $F^\ast$ is an equilibrium price if it satisfies the market-clearing condition $\theta_{\rm p}(F^\ast) + \theta_{\rm m}(F^\ast) = 0$.
\end{Definition}

\section{The unhedgeable risk setting}
\label{sec:no-hedge}

We first consider the case in which the producer cannot access the futures market. This initial assumption allows us to primarily investigate the ramifications of delegation concerning risk sharing. We first analyze the in-house production scenario, in which the firm maintains full control for producing. Subsequently, we turn our attention to the case of delegation, where production is outsourced to an external agent.

\subsection{In-house producer}\label{sec:no-hedge-0}

The in-house producer problem is given by,
\begin{align}\label{IntegratedValue}
V^{\rm i}(0,s,q) & := \sup_{\alpha} \E\Big[ U_P\Big( (S_T-\eps Q_T) Q_T -  \int_0^T c(\alpha_t) dt \Big)\Big], \\
dQ_t  &= \alpha_t dt + \sigma \big(\rho dW_t^\perp + \sqrt{1-\rho^2} dW_t\big), \quad dS_t =  \nu dW_t^\perp,\nonumber
\end{align}
where $(W_t^\perp)_t$ and $(W_t)_t$ are two independent Brownian motions and for initial conditions $S_0=s$ and $Q_0=q$. This problem is a standard stochastic control problem that yields the following Hamilton-Jacobi-Bellman (HJB) equation:
\begin{align}\label{HJB}
-\partial_tV &= \frac12 \sigma^2 V_{qq}  + \frac12 \nu^2 V_{ss} + \mu V_s 
+  \rho \sigma\nu V_{qs} + \sup_{a} \big\{  a V_q  + \eta_p c(a) V \big\}, \quad V(T,s,q) =  U_P(s-\eps q)q).
\end{align}

According to the standard verification theorem in stochastic control (see, e.g., \cite{YongZhou99}, Chapter 4), a smooth solution to the aforementioned Hamilton-Jacobi-Bellman (HJB) equation corresponds to the value function $V^{\rm i}$ and yields the optimal control $\alpha^*_i$. We will explicitly construct a solution to Equation \eqref{HJB} in the Appendix, leading to the subsequent Proposition \ref{OptiProd-no-hedge}. This proposition demonstrates that in an environment where the function $\Delta_{\rm i}$, as defined in Proposition \ref{OptiProd-no-hedge}, is strictly positive, firms possess a profitable incentive to undertake in-house production. Conditions ensuring the positivity of $\Delta_{\rm i}$ are readily identifiable. Two illustrative conditions include: the cost of effort being non-excessive relative to the inherent risk ($\gamma > \eta_p\sigma^2$), or the contract possessing a sufficiently short maturity. More precisely, let us define the functions,
\begin{align*}
\Delta_{\rm i}(t) & := 
1 
+ 2\big[\eps\gamma - \eta_p(\eps\sigma^2 -\rho\sigma\nu) \big](T-t)
+ \eta_p\nu^2\big[\gamma - \eta_p(1-\rho^2)\sigma^2\big](T-t)^2, \\
k_{\rm i}(t) & :=\frac{\gamma(2\eps + \eta_p\nu^2(T-t))}{\Delta_{\rm i}(t)}, \quad
\ell_i(t) := \frac{1 +  \eta_p\rho\sigma\nu(T-t)}{1+\frac{\eta_p\nu^2}{2\eps} (T-t)}, \quad
q_m(s) := \frac{s}{2\eps}.
\end{align*}

The following proposition provides the dynamics of the optimal production
\begin{Proposition}\label{OptiProd-no-hedge}
Assume $ \Delta_{\rm i}(t) >0$.
The dynamics of the optimal production is
\begin{align*}
dQ^{*}_t & = k_{\rm i}(t)
\big[  \ell_i(t) q_m(S_t)  - Q_t^*  \big] dt  + \sigma \big(\rho dW_t^\perp + \sqrt{1-\rho^2} dW_t\big).
\end{align*}
The optimal level of effort for the in-house producer is
\begin{equation}\label{effort_integrated}
    \alpha_{\rm i}^*(t,s,q)=k_{\rm i}(t)(\ell_i(t) q_m(s) - q).
\end{equation}
\end{Proposition}

\proof See the Appendix \ref{proofoptiprodnohedge}.\ep\\

The optimal production quantity follows an Ornstein-Ulhenbeck process, characterized by a stochastic mean-reversion level. This level is, up to a deterministic scaling factor, equivalent to the production level $q_m(S_t)$ that maximizes instantaneous revenues at time $t$. The stochastic nature of $Q^*_t$ incorporates the volatility of the demand shift factor $S_t$. Notably, this holds true even when the production risk and demand risk are uncorrelated.
Consequently, the dynamics of the optimal production estimator constitute a mean-reverting process. However, the production can significantly deviate from the monopolist's production level $q_m(S_t)$, with the extent of this deviation contingent upon the producer's risk-aversion $\eta_p$, the demand volatility $\nu$ and the market power $\eps$. These findings align with the linear specification presented by \cite{RamanChatterjee95}, who characterize optimal monopoly pricing under demand uncertainty. In the absence of demand uncertainty (i.e., $\nu=0$), the optimal production precisely tracks the optimal monopoly production $q_m(s)$.

\subsection{Delegation}\label{sec:no-hedge-2}

Our attention now shifts to the situation where the principal offers a take-it-or-leave-it contract to an agent. This agent undertakes an unobservable effort to produce the good, receiving $\xi$ as compensation. Once the contract is accepted, the agent's problem is
\begin{align*}
& V^{\rm A}(\xi) := \sup_{\alpha} \E \Big[  U_A\Big( \xi - \int_0^T  \frac{\alpha_t^2}{2\gamma}  dt \Big) \Big], \\
& dQ_t = \alpha_t dt + \sigma \big(\rho dW_t^\perp + \sqrt{1-\rho^2} dW^\alpha_t\big),\\
& dS_t = \nu dW_t^\perp.
\end{align*}
The now well-established literature on dynamic contracting in the Holmstrom-Milgrom setting (see \cite{SchattlerSung} and \cite{Cvitanic18}) gives the form of the compensation $\xi$ as the value at time $T$ of the so-called expected wage process $Y_T^Z$, where, following \cite{SchattlerSung} and \cite{Cvitanic18}, $Z$ is a bi-dimensional control $Z=(Z_1,Z_2)$ where $Z_1$ and $Z_2$ are the contract sensitivities with respect to the market risk and production, respectively. Therefore, when the principal offers a contract $Y_T^Z$, the agent optimally exerts the effort $\alpha$ which is uniquely determined as the linear function $\alpha=\gamma Z_2$. It is important to recognize that market risk sensitivity $Z_1$ operates solely as a risk-sharing mechanism, not an incentive mechanism. Its purpose is confined to the reallocation of market risk between the principal and agent, serving to ameliorate the principal's exposure without providing direct motivation for specific effort levels. Specifically, the choice of a control $Z$ influences the probability measure used to evaluate final production solely through its second component $Z_2$.\\
In the following, we will precisely describe the compensation mechanism. To do so, we first introduce the relevant notation:
\begin{align}\label{eq:Ci}
\Sigma:= \left(\begin{array}{cc}\nu&0\\ \sigma\rho&\sigma\sqrt{1-\rho^2}\end{array} \right),  \quad A=\Sigma\Sigma^\tp, \quad E_{22}:=\left(\begin{array}{cc}0&0\\ 0&1\end{array} \right) \quad \text{and} \quad
C_j := \gamma E_{22}+\eta_j A, \quad  j\in \{ a,p\}, 
\end{align}
the contract is given by  $\xi=Y_T^Z$ where the dynamics of $Y_t^Z$ is under $\P^{Z_2}$,
\begin{align*} 
dY_t & = \frac12\big(\gamma |E_{22}Z_t|^2 + \eta_a |\Sigma^\top Z_t |^2 \big)dt + Z_t \mcdot \Sigma d\mathbb{W}_t
= \frac12 Z_t \mcdot C_a  Z_tdt +Z_t \mcdot \Sigma d\mathbb{W}_t,  
\end{align*}
where  $\mathbb{W} := (W^\perp, W^{Z_2})$ is a bi-dimensional standard Brownian motion under the probability measure $\P^{Z_2}$.
Starting from $x=(s,q)$ and assuming without loss of optimality that the participation constraint is binding, we obtain the principal value function, denoted by  $V^{\rm P}_{NH}$, when the principal  is precluded from participating in the futures market. We have,
\b*
V^{\rm P}_{NH}
=
-U_{\rm P}(-r_0)
V_{\rm d}(0,x),
&\mbox{with}&
V_{\rm d}(0,x)
:=
\sup_Z \E^{\P^{Z_2}}\Big[ U_{\rm P}\Big( \pi_T - Y_T^Z +r_0\Big)\Big]=\sup_Z \E^{\P^{Z_2}}J_{\rm d}(x,Z) ,
\e*
The function $V_{\rm d}$ represents the principal value resulting from the delegation, regardless of the agent's participation constraint $r_0$. Using the Girsanov Theorem, we observe that the performance function $J_{\rm d}$ may be rewritten as
\begin{align*}
J_{\rm d}(x,Z) =&
\E^\Q\Big[ U_{\rm P}\Big( (S_T-\eps Q_T) Q_T - \frac12\int_0^T Z_t\!\cdot\! \bar C Z_t dt \Big)\Big],
\end{align*}
with $
 \bar{C} := \gamma E_{22} + (\eta_a+\eta_p) A$
and where $\Q$ is the equivalent probability measure to $\P^{Z_2}$ with density
\begin{align*}
\frac{d\Q}{d\P^{Z_2}}\vert \mathcal{F}_T = \Ec\Big(\int_0^.\eta_p\Sigma^\tp  Z_t\!\cdot\!d\mathbb{W}_t\Big)_T.
\end{align*}  
Thus, the process
$\mathbb{B}_t := \mathbb{W}_t-\int_0^t\eta_p\Sigma^\tp  Z_sds$ 
defines a $\Q-$ Brownian Motion and we have the following dynamics for the process $X_t:=(S_t,Q_t)^\tp $:
\begin{align*}
dX_t &= C_p  Z_t dt + \Sigma d\mathbb{B}_t,
\end{align*}
The principal is therefore confronted with the linear quadratic stochastic control problem
$$
V^{\rm d}(0,x)=\sup_{Z} \E^\Q\Big[ U_{\rm P}\Big( X_T \cdot MX_T - \frac12\int_0^T Z_t\!\cdot\! \bar C Z_t dt \Big)\Big], \quad
M := \begin{pmatrix} 0 & \frac12 \\ \frac12 & - \eps \end{pmatrix}.
$$
In line with the analysis in the setting without delegation, we define the following functions and parameters:
\begin{align}\label{eq:k-ell-delegation}
\Delta_{\rm d}(t) & = 1 + 2\big[\eps\gamma_d - \bar\eta(\eps\sigma^2 - \rho\sigma\nu) \big](T-t)
+ \bar\eta\nu^2\big[\gamma_d - \bar\eta(1-\rho^2)\sigma^2 \big](T-t)^2, \\
k_{\rm d}(t) & =\gamma \frac{\gamma+\eta_p(1-\rho^2)\sigma^2}{\gamma + (\eta_a+\eta_p)(1-\rho^2)\sigma^2} \frac{2\eps + \bar\eta\nu^2(T-t)}{\Delta_{\rm d}(t)},\quad
\ell_{\rm d}(t) = \frac{1 +  \bar \eta\rho\sigma\nu(T-t)}{1+\frac{\bar \eta \nu^2}{2\eps} (T-t)}, \\
\bar \eta&=\frac{\eta_a\eta_p}{\eta_a+\eta_p},\quad
\gamma_d  =  \gamma\Big( 1 - \frac{(1-\rho^2)\sigma^2 \eta_a^2}{ (\eta_a+\eta_p)[\gamma +(\eta_a+\eta_p)(1-\rho^2)\sigma^2]}
\Big).
\end{align}

We provide the explicit solution to this problem in the following proposition.
\begin{Proposition}\label{Delegation-no-hedge}
   Assume $\Delta_{\rm d}(t) >0$. The dynamics of the optimal production is
    \begin{align*}
dQ^{*}_t & = k_{\rm d}(t)
\big[  \ell_{\rm d}(t) q_m(S_t)  - Q_t^*  \big] dt  + \sigma \big(\rho dW_t^\perp + \sqrt{1-\rho^2} dW_t^*\big).
\end{align*}
The optimal level of effort for the delegated producer is
\begin{equation}\label{effort_delegated}
    \alpha_{\rm d}^*(t,s,q)=k_{\rm d}(t)(\ell_d(t)\frac{s}{2\eps}-q).
\end{equation}
\end{Proposition}
\proof See the Appendix \ref{proofdelegationnohedge}.\ep\\

The two propositions \ref{OptiProd-no-hedge} and \ref{Delegation-no-hedge} establish that the optimal production strategies in both the in-house and delegated scenarios exhibit an identical structural form, differing solely in their underlying parameter values. This remarkable symmetry is not merely an incidental observation; rather, it offers a significant methodological advantage. By revealing a consistent mathematical framework across these distinct organizational structures, this congruence facilitates direct and meaningful comparisons between the outcomes of in-house production and contractual delegation. Specifically, it enables a more precise and rigorous assessment of the moral hazard risk inherent in delegation. Because the fundamental relationships governing production remain stable, variations in performance or efficiency can be more directly attributed to the presence and management of agency problems, rather than to disparate underlying models. This structural parallelism thus provides a clear lens through which to evaluate the costs and benefits of each organizational form, particularly concerning the mitigation or exacerbation of moral hazard.\\
In particular, the impact of delegation can be measured by considering the agent's risk aversion. The sufficient condition $\gamma_d > \bar\eta \sigma^2$ to have $\Delta_{\rm d}(t)>0$ for all $t \in [0,T] $ is always satisfied for low values of the agent risk aversion parameter $\eta_a$, regardless of the value of the principal risk aversion parameter $\eta_p$. Comparing effort levels between a vertically in-house producer and a delegated situation is not straightforward and is postponed to Section \ref{sec:illustration}. Nonetheless, the following observation is interesting: the optimal effort level of a risk-neutral agent is equal to the optimal effort level of a vertically in-house producer in the absence of market risk, namely 
$$
\alpha_d^*(t,s,q)=\frac{\gamma2\eps}{1+\gamma2\eps(T-t)}\Big(\frac{s}{2\eps} -q\Big).
$$
This implies that a producer who has fully hedged their demand risk will produce the same quantity as an owner selling their asset to a risk-neutral agent. Consequently, in the absence of a futures market to hedge demand risk, a principal with production capability may gain from delegating to an agent with low risk aversion. This suggests that delegation can be viewed as a substitute for the absence of a futures market.

\subsection{Delegation vs. in-house production} 
In this section, we will leverage the structural similarities observed between the equations defining the principal's value under the two different operational modes. This commonality in their mathematical forms provides a powerful analytical tool. By exploiting these resemblances, we can systematically investigate and delineate the precise conditions under which delegation becomes a desirable strategic choice. Our approach will focus on identifying the specific parameter constellations or environmental factors that render the principal's value, when derived through a delegated structure, superior to alternative arrangements, such as in-house production. This comparative analysis, facilitated by the analogous nature of the value functions, will thus illuminate the fundamental drivers behind the decision to delegate.
We reiterate that the certainty equivalent value functions $v^{\rm i}$ 
  (in-house producer) and $v^{\rm d}$
  (delegation) correspond to the solutions of the subsequent Partial Differential Equations (PDEs):
\begin{align*}
&\Lc( \Lambda^j,v^j) =0, \quad v^{j}(T,x) = x \mcdot M x, \quad
j \in \{ {\rm i}, {\rm d}\}, 
\end{align*}
 where $\Lc(\Lambda,v) := \partial_t v   +\frac12 Tr(AD^2v)- \frac12 Dv \mcdot \Lambda Dv$ with 
$\Lambda^{\rm d}  = \bar \eta A-\gamma_d E_{22}$ and  $\Lambda^{\rm i}  = \Lambda^{\rm d}  + L$ where 
\begin{align*}
L &:= \frac{\eta_p^2}{\eta_a+\eta_p} \begin{pmatrix}
\nu^2  &  \rho\sigma\nu\\
&\\
 \rho\sigma\nu  & \sigma^2(\rho^2 + \ell)
 \end{pmatrix}, \quad
 \ell := \frac{\eta_a+\eta_p}{\eta_p^2} \frac{(1-\rho^2)\big[\gamma(\eta_p-\eta_a)+ (1-\rho^2)\eta_p^2\sigma^2\big]}{\gamma + (\eta_a+\eta_p)(1-\rho^2)\sigma^2}.
\end{align*}
The relationship directly reveals that $\Lambda^{\rm d} \leq \Lambda^{\rm i}$ if $\ell \ge 0$. This observation underpins the following proposition.

\begin{Proposition}
Delegation improves the producer’s payoff when $v^{\rm i}(s,q) \le v^{\rm d}(s,q)$, which holds true  if  
\begin{align}\label{eq:cond-del-int}
\eta_a \leq \eta^* \equiv\eta_p \Big(1 + \frac{(1-\rho^2)\eta_p\sigma^2}{\gamma}\Big). 
\end{align}
\end{Proposition}

The proof relies on a standard comparison result on the PDEs satisfied by the two functions $v^{\rm i}$ and $v^{\rm d}$, where we have $v^{\rm i} \leq v^{\rm d}$ if condition~\eqref{eq:cond-del-int} applies.
A key question explored in this paper pertains to identifying the circumstances under which the principal will favor delegation over in-house production. More precisely, given an agent's participation constraint of $r_0$, the principal, when in state $(s,q)$, will opt for delegation if
\begin{equation}
    \label{eq:PCforDeleg}
v^{\rm i}(s,q)+r_0 \le v^{\rm d}(s,q).
\end{equation}
Thus, delegation is optimal if both conditions \eqref{eq:cond-del-int} and \eqref{eq:PCforDeleg} apply. We observe that the firm can be better off delegating production risk to an agent who is even more risk averse than she is, at least if the participation constraint is not too tight given the market conditions prevailing at the time the contract is signed. All else being equal, profitable delegation is more likely when production technology is not efficient, ($\gamma$ low), or when the market risk is high, ($\sigma$ high). 

Let $v^{\rm d,fb}$ denote the certainty equivalent of the producer under first-best delegation and $v^{\rm i}$ the certainty equivalent under in-house production. The net gain or loss from delegation, which we denote as $\mathcal{B}$, is defined as $\mathcal{B} := v^{\rm d,sb} - r_0 - v^{\rm i}_0$. This gain can be decomposed into two components:
\begin{align*}
\mathcal{B} = \underbrace{v^{\rm d,sb} - v^{\rm d,fb}}_{\text{information rent}} + \underbrace{v^{\rm d,fb} - v^{\rm i}_0}_{\text{risk sharing premium}} - \,\,\, r_0
\end{align*}
It is well-established that the information rent is non-positive and its absolute value is a non-decreasing function of the agent's risk aversion, $\eta_a$. Conversely, the risk-sharing premium is always non-negative and decreases with $\eta_a$. Consequently, as the agent's risk aversion increases, the benefits from risk sharing diminish and are eventually outweighed by the cost of the information rent. Numerical illustration of this comparative static is given in Section~\ref{sec:illustration}.

When  $\rho^2=1$ , the sufficient condition for delegation, as articulated in equation \eqref{eq:cond-del-int}  holds true for agents who are no more risk-averse than the principal. In this specific case of perfect correlation between production and demand shocks, the principal can perfectly infer the agent's effort level by observing the demand shift, $S$. This allows the principal to recover the first-best value, a result confirmed by the solution to the first-best problem detailed in  Appendix \ref{app:first-best}.

As a particular case of the condition \eqref{eq:cond-del-int}, a risk-averse Principal's welfare is enhanced by engaging a risk-neutral Agent for delegated tasks, as opposed to operating under an in-house framework. Consequently, the optimal contractual form is easily derivable, which, as anticipated, provides the Principal with full assurance.
To see this, let us consider the contract defined as
\begin{align}\label{eq:Yetaa0}
Y^*_T = r_0 - \bar v + (S_T - \eps Q_T) Q_T,  \quad
\text{where} \quad
\bar v := \sup_{\alpha} \E\Big[(S_T - \eps Q_T) Q_T - \int_0^T c(\alpha_t) dt\Big].
\end{align}
Clearly, it satisfies the participation constraint of the Agent because $V^{\rm A}(Y^*_T)=\sup_{\alpha} \E\Big[Y^*_T - \int_0^T c(\alpha_t) dt\Big]=r_0$.
On the other hand, take any contract $Y_T^Z$ with $Z=(Z_1,Z_2)$ satisfying the incentive-compatibility and participation constraint conditions. The optimal effort associated to $Y_T^Z$ is $\alpha^*(Z)=Z_2$ and the participation constraint implies
$$
\E\big[ Y_T^Z \big]\ge r_0 +\E\Big[\int_0^T c(\alpha^*_t) dt\Big].
$$
Applying Jensen's inequality, we obtain
\begin{align*}
    V^{\rm P}(Y_T) &= \E\big[U_{\rm P}\big((S_T - \eps Q_T) Q_T - Y_T^Z\big)\big]\\
    &\le U_{\rm P}\Big( \E\big[(S_T - \eps Q_T) Q_T - Y_T^Z\big] \Big)\\
    &\le U_{\rm P}\Big( \E\big[(S_T - \eps Q_T) Q_T - \int_0^T c(\alpha^*_t) dt -r_0\big] \Big)\\
    &\le U_{\rm P}(\bar v - r_0),
\end{align*}
which proves the optimality of the full insurance contract.
It is noteworthy that this full insurance is completely unaffected by the quantity of risk sources and their correlation. Preferences alone determine this perfect hedging result, consequently eliminating the necessity of a forward market. Hereafter, we will assume that $\eta_a>0$ to explore the interplay between the agency problem and the forward market.

\section{Static hedging with forward contracts}\label{sec:stat-hedge}

This section focuses on the optimal production and hedging decisions of a monopolist CARA Principal. The Principal possesses the capability for direct production but also considers the alternative of engaging a CARA Agent. Additionally, the Principal has the opportunity to hedge using a forward contract at time $0$. Following the methodology of the preceding section, we commence by analyzing the case of an in-house producer.

\subsection{In-house producer}\label{sec:stat-hedge-0}

This study considers a monopolistic producer with an initial opportunity (at $t=0$) to trade a quantity $\theta$ of future production at a specified forward price $F$. Engagement in a forward market introduces a strategic incentive for producers to control output, aiming to enhance their forward market stance. For instance, a producer with a short forward position may be motivated to increase production, thereby potentially depressing spot prices. The market-maker, acting as the forward market liquidity provider, will internalize this strategic behavior when establishing the forward price $F$. Consequently, this framework necessitates an examination of market equilibrium.\\

Let us begin by examining the producer's decision problem, assuming the forward price $F$ is given. For a given position $\theta$, the in-house producer faces the following production problem starting from $S_0=s$ and $Q_0=q$,
\begin{align}\label{pb:hedgedemand}
&\sup_{\alpha} \E\Big[ U_P\Big( P(S^s_T,Q^q_T) Q^q_T  + \theta\big[P(S^s_T,Q^q_T) - F\big] - \int_0^T c(\alpha_t)dt  \Big)\Big]\\
&=-U_P(-\theta F)  \sup_{\alpha} \E\Big[ U_P\Big( P(S^s_T,Q^q_T) (Q^q_T + \theta) - \int_0^T c(\alpha_t)dt \Big)\Big].
\end{align}
 We also make the key observation from the inverse demand specification of the model that $$P(S^s_T,Q^q_T) (Q^q_T + \theta) = P(S^{s+\varepsilon \theta}_T,Q^{q+\theta}_T) Q^{q+\theta}_T.$$
Then, we deduce the in-house producer value function has the form
$$
-U_P(-\theta F)V_i(0,s+\eps\theta,q+\theta)
$$
where $V_i$ is the in-house producer value function \eqref{IntegratedValue}.
Consistent with our previous analysis, the certainty equivalent will be utilized, $V_i=U_P(v_i)$.
The in-house producer's demand for optimal hedging is derived by solving the following static problem at time $t=0$, where we denote $\bar\varepsilon=(\eps, 1)^\top$ and $x=(s,q)^\top$:
$$
\sup_\theta \big\{v_i(0,x+\theta\bar \eps)-\theta F\big\}.
$$
From the previous analysis, we know that $v_i$ is a quadratic function. More precisely, $v^{\rm i}(0,x)=K^{\rm i}_0(0)+x\cdot K^{\rm i}_2(0)x$ with
$$K^{\rm i}_2(t)=(I+2(T-t)M\Lambda_{\rm i})^{-1}M.$$ Consider the real-valued differentiable function 
$$
g(\theta)= v_i(0,x+\theta\bar \eps)-\theta F.
$$
The first-order condition gives
$$
2\bar\eps\cdot K^{\rm i}_2(0)(x+\theta\bar\eps)=F
$$
and the second-order condition gives the existence of a maximum for $g$ if $2\bar\eps\cdot K^{\rm i}_2\bar\eps<0$.

Thus, assuming $2\bar\eps\cdot K^{\rm i}_2(0)\bar\eps < 0$, the optimal forward position at time $0$ for the monopolist is
\begin{align}\label{eq:int-hedge}
\theta_i(F) & := \frac{F-2 \bar\varepsilon\cdot K^{\rm i}_2(0) X_0}{2\bar\varepsilon\cdot K^{\rm i}_2(0)\bar\varepsilon}.
\end{align}

Hypothesis $2\bar\eps\cdot K^{\rm i}_2(0)\bar\eps<0$ is indispensable for the producer to enter the forward market.  A straightforward computation reveals that
$
\text{det } K^{\rm i}_2(0)=-\frac{1}{4\Delta_i(0)},
$ 
and thus the function $v_i$ exhibits neither global concavity nor convexity when assuming $\Delta_i(0)>0$. 
This assumption can be more explicitly stated as $\eps^2\gamma - \eta_p\sigma_p^2 <0$, where we recall that $\sigma_p^2:=\epsilon^2\sigma^2+\nu^2-2\epsilon\rho\sigma\nu$ denotes the volatility of the terminal price, $P_T$. The failure of this condition implies that the firm can strategically leverage its market power by selling forward at a high price and subsequently buying back at a depressed spot price. This arbitrage allows the firm to generate profits from its forward positions, even while incurring losses on its production operations. Anticipating this behavior, the market maker becomes unwilling to provide liquidity in forward contracts, thereby precluding the existence of an equilibrium. This finding is consistent with the finance literature on strategic behavior in commodity futures markets, pioneered by the seminal works of \cite{Allaz92} and \cite{Allaz93} in a deterministic setting and later extended to mean-variance preferences of agents in \cite{Hughes97}. Furthermore, our condition for the existence of an equilibrium closely parallels the condition derived in \cite{Hughes97}.
Furthermore, observe that $\theta_i$ is a decreasing function of the forward price. The higher the forward price, the more the producer wants to sell forward.
Consequently, financial profit will materialize under conditions of low spot prices or, equivalently, high production levels. This outcome suggests that producers taking a short position in the forward market will exert greater effort. These insights are formally summarized in the following Proposition.

\begin{Proposition}\label{Integrated-hedge}
Assume $\Delta_i(t)>0$ and $2\bar\eps\cdot K^{\rm i}_2(t)\bar\eps<0$ for $t \in [0,T]$. The dynamics of the optimal production of the in-house monopolist producer when there is a forward market is
\begin{align*}
dQ^*_t & = k_{\rm i}(t)\Big[ \ell_{\rm i}(t) \Big(\frac{S_t + \eps \theta_i(F)}{2\eps}\Big) - Q^*_t - \theta_i(F)  \Big]\,dt+ \sigma \big(\rho dW_t^\perp + \sqrt{1-\rho^2} dW_t\big), \\
 &\text{ with as in Proposition \ref{OptiProd-no-hedge}, }
k_{\rm i}(t)=\gamma\frac{2\eps + \eta_p\nu^2 (T-t)}{\Delta_i(t)},
\ell_{\rm i}(t) := \frac{1 +  \eta_p\rho\sigma\nu(T-t)}{1+\frac{\eta_p\nu^2}{2\eps} (T-t)}.
\end{align*}
The optimal level of effort is
\begin{align*} 
\alpha_{i,h}^*(t)&=\alpha_{i}^*(t,s,q)+ k_{\rm i}(t)\Big(\frac{\ell_{\rm i}(t)}{2}-1\Big)\theta_i(F),
\end{align*}
where $\alpha_{i}^*(t,s,q)$ is given by \eqref{effort_integrated}.
\end{Proposition}

Under the sufficient condition $\gamma > \eta_p\sigma^2$, it is interesting to note that the condition $2\bar\eps\cdot K^{\rm i}_2(0)\bar\eps<0$ implies that $\eps$ must be smaller than a certain threshold $\hat \eps$. Thus, when technology is effective (high $\gamma$), only firms with significant market power find it advantageous to utilize the forward market.
By virtue of the explicit formulas obtained in Propositions \ref{OptiProd-no-hedge} and \ref{Integrated-hedge}, we are able to compare optimal effort levels between producers with ($\alpha_{i,h}^*$) and without ($\alpha_{i}^*$) access to a forward market. The various effects of parameters on optimal effort are illustrated in Section \ref{sec:illustration}.\\

We now direct our attention to the agents operating as counterparties in the forward market. Specifically, we assume the presence of a rational CARA market-maker responsible for liquidity provision, who is modeled as maximizing the expected value of their CARA utility, with a risk aversion parameter $m>0$.
Given a forward price $F$, the rational market maker anticipates the producer demand $\theta_i(F)$ which determines the probability distribution of the pair $(S_T,Q_T)$ through the effort level process $(\alpha_{i,h}^{\ast,F}(t))_t$. To highlight the influence of the forward price $F$ on the probability distribution, we will use the notation $\P^{F}$ for $\P^{\alpha_{i,h}^{\ast,F}(t)}$. Hence, the market maker faces the following optimization problem:

\begin{align*}
&\sup_{\theta \in \R} \E^{F}\Big[ U_m\Big( \theta(P(S_T,Q_T) - F) \Big)\Big], 
\end{align*}
where
\begin{align*}
    dQ^*_t & = k_{\rm i}(t)\Big[ \ell_{\rm i}(t) \Big(\frac{S_t + \eps \theta_i(F)}{2\eps}\Big) - Q^*_t - \theta_i(F)  \Big]\,dt+ \sigma \big(\rho dW_t^\perp + \sqrt{1-\rho^2} dW_t\big), \\
    dS_t&= \nu dW_t^\perp. 
\end{align*}
Since the pair $((S_t,Q^*_t))_{t \le T}$ is a Gaussian process, it follows that $P_T =S_T-\eps Q^*_T$ is a Gaussian random variable, and thus 
\begin{align*}
  \E^F\Big[ U_m\Big( \theta (P_T - F) \Big)\Big] 
  = U_m\Big( \theta (\E^F[P_T] - F) - \frac12 \eta_m \Var^F[P_T] \theta^2 \Big)
\end{align*}
The optimal market maker's position is then:
\begin{align*}
\theta^{(\rm m)}(F) := \frac{\E^F[P_T] - F}{\eta_m \Var[P_T]}.
\end{align*}
where it should be remembered that the distribution of $P_T= S_T-\eps Q_T^*$ depends on the forward position $\theta_i$ of the producer. First, we examine the case of a risk-neutral market-maker, ($\eta_m=0$). In that case, the equilibrium forward price $F$ must satisfy
$$
F=\E^F(P_T)=s-\eps\E^F(Q^*_T).
$$
Given a forward price $F$, the producer selects the optimal level of effort $\alpha^*_{i,h}$ which can be decomposed as
$$
\alpha_{i,h}^*(t,s,q)=k_{\rm i}(t)(\ell_{\rm i}(t)(\frac{s}{2\eps}-q)+\left(\frac{\ell_{\rm i}(t)}{2}-1 \right)\theta_i(F).
$$
Thus, we have
\begin{align*}
    \frac{d}{dt}\E^F(Q_t^*)&=k_{\rm i}(t)(\ell_{\rm i}(t)(\frac{s}{2\eps}-\E^F(Q_t^*))+\left(\frac{\ell_{\rm i}(t)}{2}-1 \right)\theta_i(F)\\
    &=-k_{\rm i}(t)\E^F(Q_t^*)+b_{\rm i}(t),
\end{align*}
where $b_{\rm i}(t) = k_{\rm i}(t)[\ell_{\rm i}(t)\frac{s}{2\eps}+\left(\frac{\ell_{\rm i}(t)}{2}-1 \right)\theta_i(F)]=B_{\rm i}(t,s,q)+k_{\rm i}(t)\left(\frac{\ell_{\rm i}(t)}{2}-1 \right)\frac{F}{2\bar\eps\cdot K^{\rm i}_2(t)\bar\eps}$ where under the assumption $\bar\varepsilon\cdot K^{\rm i}_2(t)\bar\varepsilon < 0$, we have
\begin{align*}
B_{\rm i}(t,s,q) := k_{\rm i}(t)\Big(\ell_{\rm i}(t)\frac{s}{2\eps}-\left(\frac{\ell_{\rm i}(t)}{2}-1 \right) \frac{2 \bar\varepsilon\cdot K^{\rm i}_2(t) X_0}
       {2\bar\varepsilon\cdot K^{\rm i}_2(t)\bar\varepsilon}\Big)
\end{align*}

 Solving the ODE, we obtain
$$
\E^F(Q_T^*)=\Big(q+\int_0^T B_{\rm i}(t,s,q)e^{\int_0^t k_{\rm i}(u)\,du}\,dt+\int_0^T e^{\int_0^t k_{\rm i}(u)\,du}k_{\rm i}(t)\left(\frac{\ell_{\rm i}(t)}{2}-1 \right)\frac{F}{2\bar\eps \cdot K^{\rm i}_2(0)\bar\eps} \,dt\Big)e^{-\int_0^T k_{\rm i}(t)\,dt}
$$
We thus derive the following result whose proof is immediate because the function 
$$g(F)=\E^F(P_T)-F=s-\eps \E^F(Q_T^*)-F$$ is affine.
\begin{Proposition}\label{prop:existence-F-integrated}
    If $\eta_m=0$ and  $\eps\int_0^T  \left(\frac{\ell_{\rm i}(t)}{2}-1 \right)k_{\rm i}(t)e^{-\int_t^T k_{\rm i}(u)\,du}\,dt \neq -2\bar\eps\cdot K^{\rm i}_2(0)\bar \eps$ then there exists an equilibrium forward price $F^*(s,q)$.
\end{Proposition}

\paragraph{Risk-averse market maker}

We now shift our focus back to the more general scenario involving a risk-averse market-maker. The existence of an equilibrium price in the forward market hinges on verifying that the aggregate demand for the forward contract is zero. Mathematically, this means we must find a solution to the equation $\theta_i(F)+\theta^{(\rm m)}(F)=0$, where $\theta_i(F)$ represents the demand from the owner and $\theta^{(\rm m)}(F)$ represents the demand from the market-maker, both as functions of the forward price $F$.

A crucial step in establishing this equilibrium is to demonstrate that this aggregate demand equation is affine in~F. If this linearity holds, it immediately implies not only the existence of an equilibrium forward price but also its uniqueness. This is a powerful result, as it guarantees that a single, well-defined market-clearing price will emerge.

The key insight underpinning this affine relationship, and therefore the existence and uniqueness of the equilibrium, is the observation that the variance of the spot price $P_T$ is independent of the forward price $F$. This independence simplifies the utility maximization problems for both the producer and the market-maker, allowing their individual demand functions, $\theta_i(F)$ and $\theta^{(\rm m)}(F)$, to be expressed linearly in $F$. When these linear functions are summed to form the market-clearing equation, the resulting equation remains affine, thus securing the existence and the uniqueness of the equilibrium.\\
Due to our specification of the inverse demand function, we have $$\Var(P_T)=\Var(S_T)-2\eps \Cov  (S_T,Q_T)+\eps^2 \Var(Q_T).$$
From Proposition \ref{Integrated-hedge}, we have
$$
Q_T^*=q_0+\int_0^T \alpha_{i,h}^*(t,S_t,Q_t)\,dt+\sigma(\rho W_T^\perp+\sqrt{1-\rho^2}W_T).
$$
Taking expectations, we obtain
$$
\E^F(Q_T^*)=q_0+\int_0^T \E^F(\alpha_{i,h}^*(t,S_t,Q_t))\,dt.
$$
Therefore,
\begin{align*}
Q_T^*-\E^F(Q_T^*)&=\int_0^T \Big(\alpha_{i,h}^*(t,S_t,Q^*_t)-\E^F(\alpha_{i,h}^*(t,S_t,Q^*_t))\Big)\,dt +\sigma(\rho W_T^\perp+\sqrt{1-\rho^2}W_T)\\
&=\int_0^T k_{\rm i}(t)\ell_{\rm i}(t)\Big(\frac{S_t-s}{2\eps}-(Q^*_t-\E^F(Q^*_t))\Big)+\sigma(\rho W_T^\perp+\sqrt{1-\rho^2}W_T)
\end{align*}
We then observe that the difference $Q_T^*-\E^F(Q_T^*)$ is independent of the forward price $F$ which implies that the variance of $P_T$ is independent of the forward price $F$ as expected. With these remarks, we are in a position to state an equilibrium  existence and uniqueness result.

Define
\[
\Gamma_i := \int_0^T \Big(\frac{\ell_{\rm i}(t)}{2}-1\Big)k_{\rm i}(t) \exp\Big(-\int_t^T k_{\rm i}(u)\,du\Big)\,dt.
\]
\begin{Proposition}\label{prop:in-house-eq-averse}
Assume $\Delta_i(t)>0$ for all $t\in[0,T]$ and $2\bar\eps\cdot K_2^i(0)\bar\eps<0.$ 
 If $\eta_m\Var(P_T)\neq 2\bar\eps\cdot K_2^i(0)\bar\eps + \eps\,\Gamma_i,$ then there exists a unique equilibrium forward price $F_i^*(s,q)$.
\end{Proposition}

\hs

This equilibrium arises from the interplay of two fundamental opposing forces. On one side, the producer, having taken a short position in the forward market, is incentivized to exert effort. Their objective is to strategically influence and reduce the eventual spot price $P_T$, thereby maximizing their financial gain from the forward contract. Essentially, a lower spot price translates to a larger profit from the prior forward sale.

Conversely, the market-maker, as the counterparty providing liquidity, is acutely aware of this incentive. Anticipating the monopolistic producer's potential for price manipulation, the market-maker will resist setting the forward price $F$ too low. Their risk aversion, coupled with their foresight into the producer's strategic behavior, prevents them from accepting forward contracts at prices that do not adequately compensate for the risk of a manipulated, depressed spot price.

It is precisely this inherent tension—the producer's drive to depress spot prices for forward gains versus the market-maker's defensive pricing against such manipulation—that ultimately establishes and sustains the unique equilibrium in the forward market. This dynamic balance of interests ensures that a specific forward price $F$ emerges where neither party has an incentive to unilaterally deviate from their chosen strategy.

\subsection{Second-best delegation}\label{ssec:sb-delegation}
 
This subsection examines the influence of the forward market on the design of the optimal delegation contract. By exploiting the inherent symmetry between the in-house producer's problem and that of the delegating producer, we deduce the delegated producer value function in the presence of a forward market has the form
$$
-U_P(-\theta F)V_d(0,s+\eps\theta,q+\theta)
$$
where $V_d$ is the delegated producer value function. Again, the certainty equivalent will be utilized, $V_d=U_P(v_d)$.
The delegated producer's demand for optimal hedging is derived by solving the following static problem at time $t=0$, where we denote $\bar\varepsilon=(\eps, 1)^\top$ and $x=(s,q)^\top$:
$$
\sup_\theta \big\{v_d(0,x+\theta\bar \eps)-\theta F\big\}.
$$
From the previous analysis, we know that $v_d$ is also a quadratic function. Namely,
\begin{align*}
v^{\rm d}(0,x)=K^{\rm d}_0(0)+x\cdot K^{\rm d}_2(0)x, \quad K^{\rm d}_2(t)=(I+2(T-t)M\Lambda_{\rm d})^{-1}M.
\end{align*}
The first-order condition gives
$$
2\bar\eps\cdot K^{\rm d}_2(0)(x+\theta\bar\eps)=F
$$
and the second-order condition gives the existence of a maximum for $g$ if $2\bar\eps\cdot K^{\rm d}_2\bar\eps<0$. Thus, assuming $2\bar\eps\cdot K^{\rm d}_2(0)\bar\eps<0$, the optimal forward position at time $0$ for the delegated producer is
\begin{align}\label{eq:del-hedge}
\theta_{\rm d}(F) & :=  \frac{F-2 \bar\varepsilon\cdot K^{\rm d}_2(0) X_0}{2\bar\varepsilon\cdot K^{\rm d}_2(0)\bar\varepsilon}.
\end{align}

Relying once again on the symmetry argument between the two organizational structures, we can state the following proposition which provides the dynamics of optimal production under delegation (see Appendix~\ref{proofdelegationdhedge} for an argument).

\begin{Proposition}\label{Delegation-hedge}
Assume $\Delta_d(t)>0$ and $2\bar\eps\cdot K^{\rm d}_2(t)\bar\eps<0$ for $t \in [0,T]$. The dynamics of the optimal production of the delegation monopolist producer when there is a forward market is
\begin{align*}
dQ^*_t & = k_{\rm d}(t)\Big[ \ell_{\rm d}(t) \Big(\frac{S_t + \eps \theta_d(F)}{2\eps}\Big) - Q^*_t - \theta_d(F)  \Big]\,dt+ \sigma \big(\rho dW_t^\perp + \sqrt{1-\rho^2} dW_t\big), 
\end{align*}
where $k_{\rm d}$ and $\ell_{\rm d}$ are defined in~\eqref{eq:k-ell-delegation}. The optimal level of effort is
\begin{align*} 
\alpha_{d,h}^*(t)&=\alpha_{d}^*(t,s,q)+ k_{\rm d}(t)\Big(\frac{\ell_{\rm d}(t)}{2}-1\Big)\theta_d(F),
\end{align*}
where $\alpha_{d}^*(t,s,q)$ is given by \eqref{effort_delegated}.
\end{Proposition}

\hs

We now take the point of view of the counterparty. We posit that the market-maker possesses a degree of rationality sufficient to fully anticipate the structure and parameters of the optimal delegation contract between the principal and the agent. This crucial assumption implies that the market-maker is not operating under informational asymmetry regarding the internal arrangements of the producer. Consequently, based on their understanding of this optimal contract, they can accurately foresee the stochastic properties of the production process. This foresight directly extends to the probability distribution of the spot price $P_T$ at the terminal date.

In the context of contract theory, particularly when considering the presence of a forward market, this anticipation is vital. The market-maker, being sophisticated, understands that the optimal contract itself will influence the agent's effort and, by extension, the firm's output and the resulting spot price. Therefore, by correctly anticipating the contract, they can form precise expectations about the distribution of future prices. This comprehensive understanding of the underlying uncertainty – the anticipated spot price distribution – is a prerequisite for a market-maker to rationally determine and adjust their own hedging position or liquidity provision in the forward market. Without this ability to foresee the outcome influenced by the optimal contract, their actions in the forward market would be purely speculative rather than based on informed risk assessment and equilibrium seeking. As a consequence, the optimal market maker's position is
\begin{align*}
\theta^{(\rm m)}(F) := \frac{\E^F[P_T] - F}{\eta_m \Var[P_T]},
\end{align*}
where $\P^{F}$ stands for $\P^{\alpha_{d,h}^{\ast,F}(t)}$. A key observation is that the existence and uniqueness of an equilibrium forward price critically depends on the linearity of the relevant functions $\theta_{\rm d}$ and $\theta^{(\rm m)}$ with respect to the forward price $F$. This is an immediate consequence of the symmetry arguments between the two cases elaborated upon previously which allows us to state the counterpart of the Proposition~\ref{prop:existence-F-integrated}, namely:

\begin{Proposition}\label{prop:existence-F-delegation}
    If $\eta_m=0$ and  $\eps\int_0^T  \left(\frac{\ell_{\rm d}(t)}{2}-1 \right)k_{\rm d}(t)e^{-\int_t^T k_{\rm d}(u)\,du}\,dt \neq -2\bar\eps\cdot K^{\rm d}_2(0)\bar \eps$ then there exists an equilibrium forward price $F^*(s,q)$.
\end{Proposition}

The equilibrium existence condition  $2\bar\eps\cdot K^{\rm d}_2(0)\bar\eps<0$ of equation~\eqref{eq:del-hedge} is more intricate than its counterpart in equation~\eqref{eq:int-hedge}, which models the case without delegation. Nonetheless, it can be simplified and restated as follows:
\begin{align}\label{eq:cond-eq-delegation}
& 0 < \eta_p (1-\rho^2)\sigma^2 \sigma_p^2 \eta_a^2 + \big[\eta_p\sigma_p^2\gamma_p - \eps^2\gamma (\gamma_p + \eta_p(1-\rho^2)\sigma^2) \big]\eta_a - \eta_p \eps^2\gamma \gamma_p, \\
& \text{with}\, \gamma_p := \gamma + \eta_p(1-\rho^2)\sigma^2.
\end{align}
The condition above can be analyzed as a second-order polynomial in the agent's risk aversion parameter, $\eta_a$.  Note that when the channel from production to price vanishes, i.e. $\epsilon=0$, the equilibrium condition above is always satisfied. On the contrary, as soon as $\epsilon>0$, the condition is demonstrably not satisfied for a risk-neutral agent ($\eta_a=0$). Therefore, the positive root of this polynomial defines a critical threshold for the agent's risk aversion above which an equilibrium exists. Thus, an equilibrium on the forward market can only exist if the principal delegates production to a sufficiently risk-averse agent. When the agent exhibits low risk aversion, the principal's incentive to engage in external hedging via the forward market is attenuated. In this case, the delegation contract itself serves as a sufficiently robust internal mechanism for risk-sharing between the principal and the agent. Consequently, the forward market's instrumental value to the principal-agent relationship emerges only when their collective risk aversion is high, providing a complementary channel for risk transfer beyond what can be optimally achieved through the incentive contract alone.

\paragraph{Risk-averse market maker} We conclude this section by stating the analog of the Proposition~\ref{prop:in-house-eq-averse} in the present case of delegation with risk-averse market maker.
Define
\[
\Gamma_d := \int_0^T \Big(\frac{\ell_{\rm d}(t)}{2}-1\Big)k_{\rm d}(t) \exp\Big(-\int_t^T k_{\rm d}(u)\,du\Big)\,dt.
\]
\begin{Proposition}\label{prop:delegation-eq-averse}
Assume $\Delta_d(t)>0$ for all $t\in[0,T]$ and $2\bar\eps\cdot K_2^d(0)\bar\eps<0.$  If $\eta_m\Var(P_T)\neq 2\bar\eps\cdot K_2^d(0)\bar\eps + \eps\,\Gamma_d,$ then there exists a unique equilibrium forward price $F_d^*(s,q)$.
\end{Proposition}

\hs

\subsection{The case of constant demand}

The comparison of the contracts with and witout hedging can be made transparent when considering the absence of demand shift. When \(\nu=0\), the demand process is constant, \(S_t\equiv S_0\), so the state variable reduces to \(q\) and \(\rho\) drops out of the formulas. Let
\begin{align*}
& 
\Phi(t):=\exp\!\left(-\int_0^t k_d(u)\,du\right), \quad k_d(t):=\gamma\lambda_d\frac{2\eps}{\Delta_d(t)}, \quad
\lambda_d :=\frac{\gamma+\eta_p\sigma^2}{\gamma+(\eta_a+\eta_p)\sigma^2}.
\end{align*}
When $\nu = 0$ the function $\Delta_d(t)$ defined in \eqref{eq:k-ell-delegation} reduces to
\[
\Delta_d(t):=1-2\eps\beta_d(T-t), \quad
\beta_d:=\bar\eta\sigma^2-\gamma_d,
\qquad
\gamma_d = \gamma\left(1-\frac{\eta_a^2\sigma^2}{(\eta_a+\eta_p)\bigl(\gamma+(\eta_a+\eta_p)\sigma^2\bigr)}\right).
\]
 The condition for a forward equilibrium~\eqref{eq:cond-eq-delegation} reduces to $\eta_p\eta_a\sigma^2>\gamma(\eta_a+\eta_p),$
or equivalently,
\begin{align}\label{eq:cond-eq-nu=0}
\eta_p\sigma^2>\gamma,
\quad \text{and} \quad
\eta_a>\frac{\gamma\eta_p}{\eta_p\sigma^2-\gamma}.
\end{align}

\begin{Proposition}\label{prop:nu=0}
Assume \(\nu=0\), \(\Delta_d(t)>0\) for all \(t\in[0,T]\), and that the conditions for a forward equilibrium to exist~\eqref{eq:cond-eq-nu=0} are satisfied.
Then\begin{enumerate}[\upshape (i)]

\item The equilibrium position and forward price are given by
\begin{align}
 \theta^\ast & =\frac{(q_m-q_0)\big(1-\Delta_d(0)\Phi(T)\big)}{\eps\beta_d T + \Delta_d(0)\Big[ \frac{1}{2}\big(1-\Phi(T)\big) + \eta_m\eps\Var(Q_T\Big]}, \\
F^\ast & = \frac{S_0}{2} +\eps \Phi(T)\big(q_m -  q_0\big) +\Big[\frac{\eps}{2}\big(1-\Phi(T)\big) + \eta_m\eps^2\Var(Q_T)\Big]\theta^\ast,
\end{align}
where
\[
\Var(Q_T)=\sigma^2\Phi(T)^2\int_0^T \Phi(t)^{-2}\,dt.
\]

\item  The incentive sensitivities $\hat z(t,q)$ without forward position and $\hat z^{\rm h}(t,q)$ under a static forward position \(\theta\) are respectively given by
\begin{align}
\hat z(t,q)=\lambda_d\frac{2\eps}{\Delta_d(t)}(q_m-q), \quad
\hat z^{\rm h}(t,q)=\lambda_d\frac{2\eps}{\Delta_d(t)}\left(q_m-q-\frac{\theta}{2}\right).
\end{align}

\item The expected difference in compensation between the hedged case and the unhedged case is given by
\begin{align}
\E[Y_T^{\rm h}] - \E[Y_T] = \frac12(\gamma+\eta_a\sigma^2)\left( \frac14 (\theta^\ast)^2 - \theta^\ast(q_m-q_0) \right) \int_0^T \Big(\lambda_d\frac{2\eps}{\Delta_d(t)} \Phi(t) \Big)^2\,dt < 0. 
\end{align}
\end{enumerate}
\end{Proposition}

The proof is given in Appendix~\ref{app:nu=0}. In the absence of demand shift uncertainty, the capacity to hedge reduces the agent's expected compensation.  This result has a simple intuition. In the case \(\nu=0\), the principal uses two instruments to deal with the same distortion: the incentive contract and the forward position. Without hedging, the contract must provide strong incentives $\hat z(t,q)$ when production is far from the target \(q_m\). With hedging, the principal can already reduce part of this mismatch through the forward market position $\theta$. As a result, the agent's action becomes less valuable at the margin, and the principal does not need to rely as much on incentive pay. In this sense, static hedging and incentive provision are partial substitutes. At the optimum, the reduction in incentive intensity dominates, so the agent is paid less in expectation under hedging than without hedging.

When demand is uncertain ($\nu>0$), the extension of this comparative static is not amenable to analytical solutions. For this reason, we provide in Section~\ref{sec:illustration} a numerical anlysis of this result in the presence of demand shift. It shows that our result is robust in the sense that it remains valid in the vicinity of the optimal monopoly production level $q_m$ and for a relevent range of risk-aversion level of the agent relative to the principal.

\section{Illustrations}\label{sec:illustration}

To begin, we recall the baseline parameter values used in the illustrations. The exogenous parameters of the model are the volatility parameters $\sigma,\nu$, the correlation parameter $\rho$, the preference parameters $\eta_p,\eta_a,\eta_m$, market power $\eps$, production efficiency $\gamma$, maturity $T$, and the initial conditions $q_0,s_0$. Values are gathered in Table~\ref{tab:params}. Observe that we consider the case of an initial condition $q_0 = q_m = s_0/(2\eps),$ which corresponds to the optimal target at time \(0\).

We numerically illustrate the principal results of our paper. For these illustrations, we have chosen to plot the agent's preference parameter $\eta_a$ on the x-axis. 

\begin{table}[!hbt]
\begin{center}
\begin{tabular}{l  cccccccccc}
& & & & & & & & & &  \\ \hline
Parameter & $s_0$ & $q_0$ & $\sigma$ & $\nu$ & $\rho$ & $\eps$ & $\gamma$ & $\eta_p$ & $\eta_m$ & $T$ \\
& & & & & & & & & &  \\
Value & $100$ & $5000$ & $500$ & $5$ & $0$ & $10^{-2}$ & $20$ & $10^{-4}$ & $\eta_p$ & $1$ \\ \hline
\hline
\end{tabular}
\caption{Parameters values of the model.}\label{tab:params}
\end{center}
\end{table}

\paragraph{Principal Values}

\begin{figure}[!hbt]
\begin{center}
\begin{tabular}{ccc}
(a) & (b) & (c) \\
\includegraphics[width=0.3\textwidth]{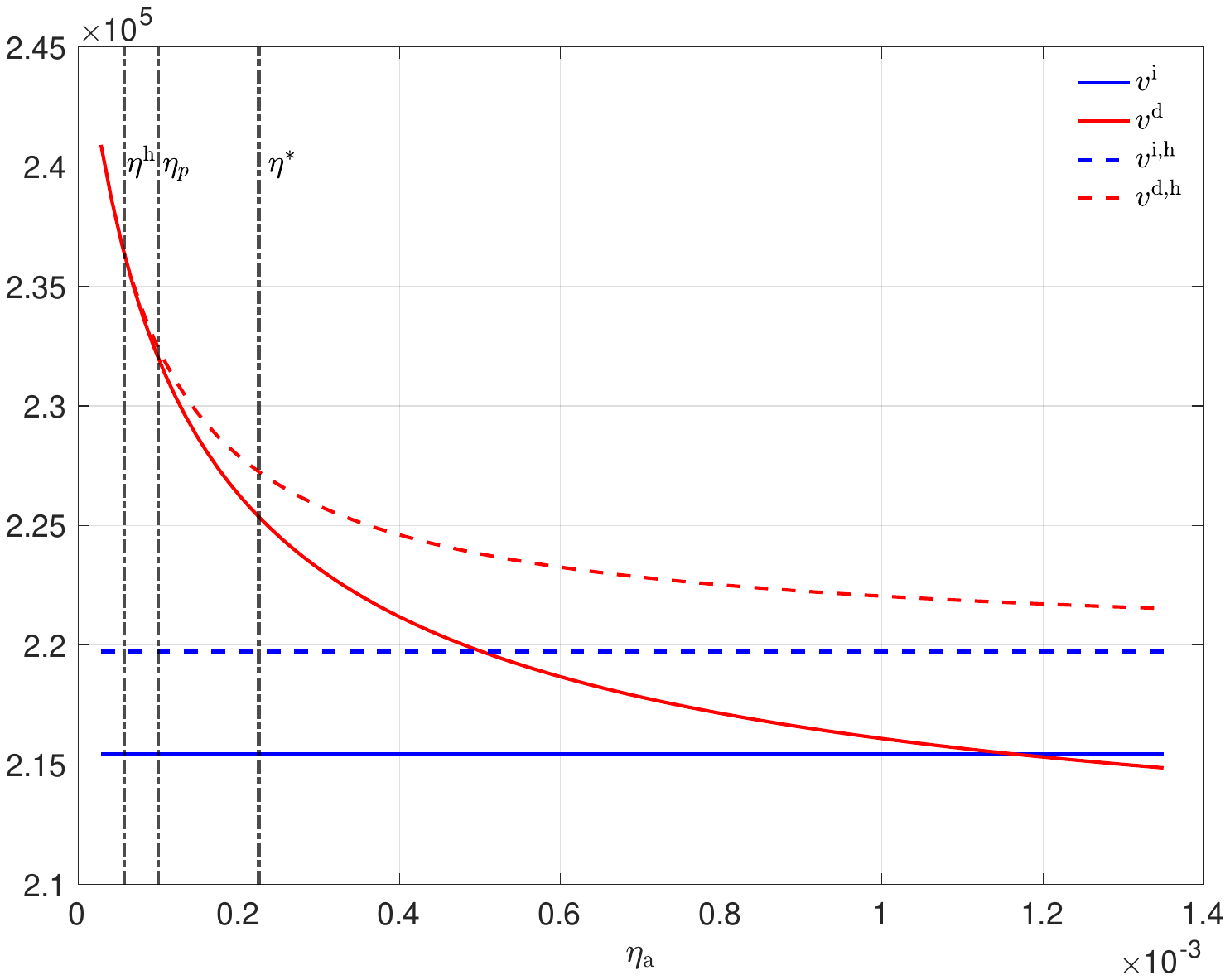} &
\includegraphics[width=0.3\textwidth]{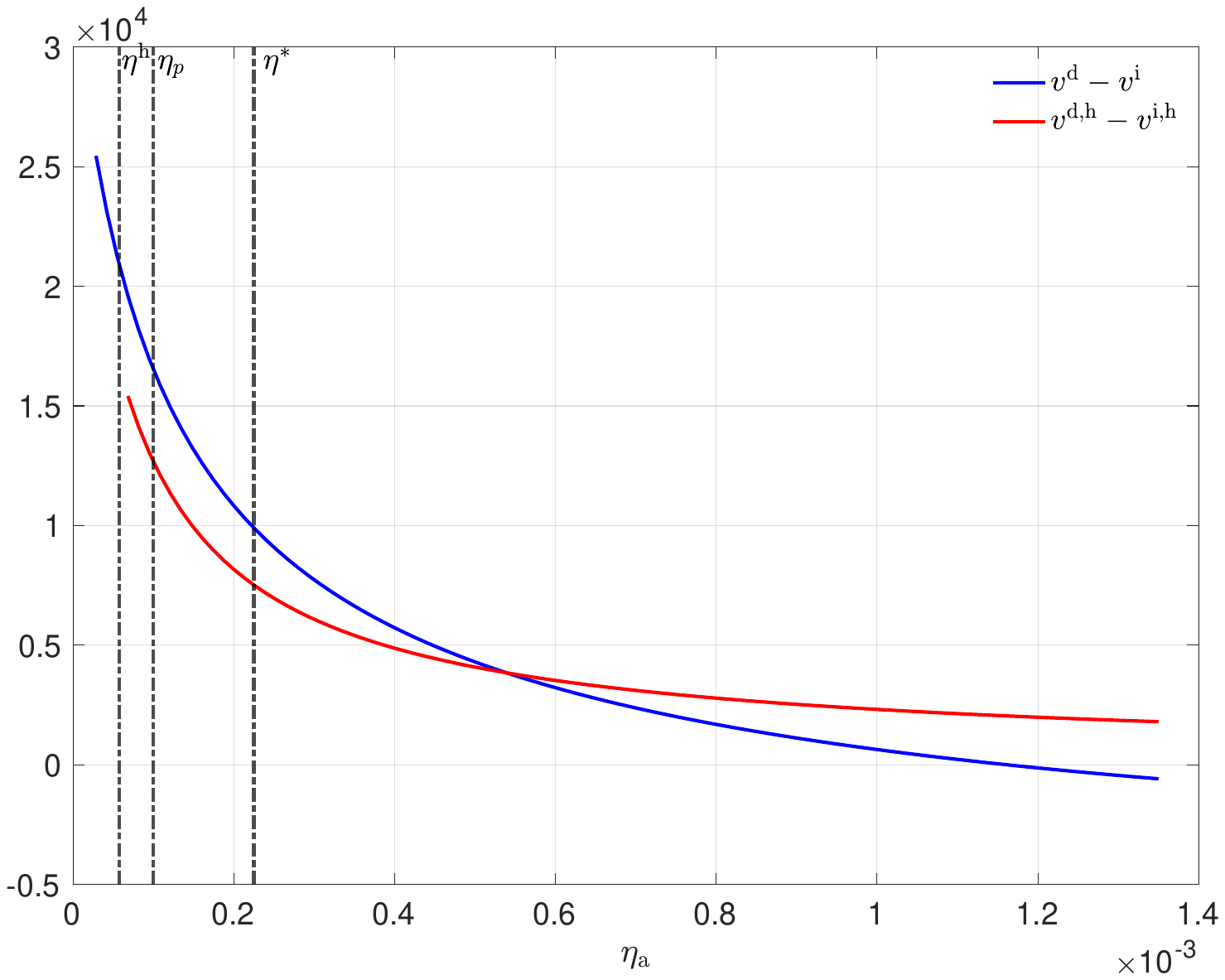}&
 \includegraphics[width=0.3\textwidth]{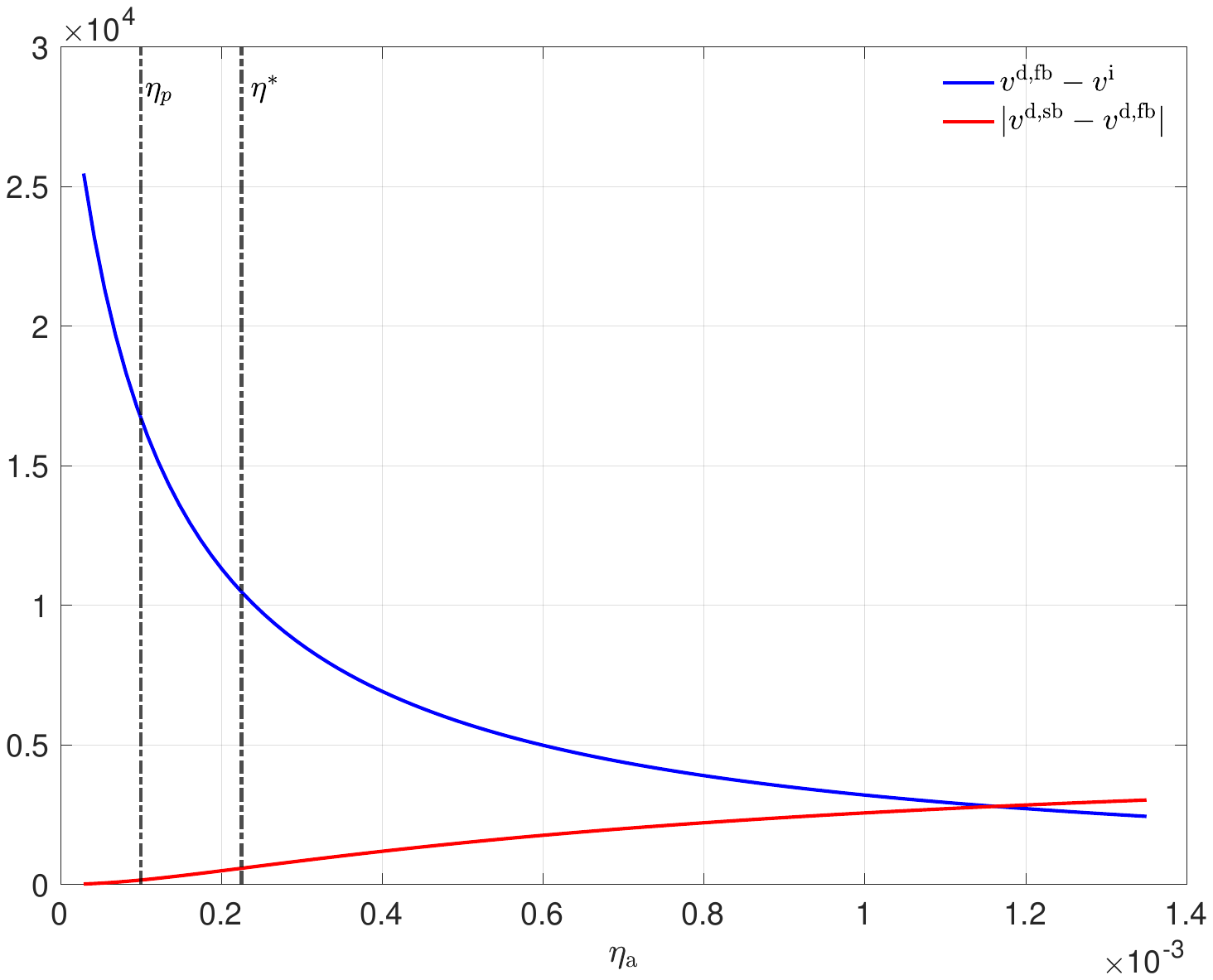}
 \end{tabular}
\end{center}
\caption{{\small In-house production versus delegation}}
\label{fig:etaA}
\end{figure}

Panel (a) illustrates the producer's certainty equivalent across four distinct scenarios, encompassing combinations of delegation versus integration, and hedging versus no-hedging. Our analysis confirms that hedging consistently improves the producer's welfare. Furthermore, the producer benefits from delegating production to an agent with a risk aversion up to $\eta^*=\eta_p\Big(1+\frac{(1-\rho^2)\eta_p\sigma^2}{\gamma} \Big)$. In this specific parameterization, $\eta^*$ is twice the producer's risk aversion. We find that the condition $\eta_a\le \eta^*$ is sufficient for delegation to yield a higher certainty equivalent compared to integration. However, for certain initial conditions, we observe that delegation can also be beneficial for agents with risk aversion greater than $\eta^*$. Crucially, we also note that the gain from hedging $v^{d,h}-v^d$ increases with the agent's risk aversion. This phenomenon is attributable to the rising cost of the information rent captured by a more risk-averse agent, a point we elaborate on further. Panel (a) further illustrates that this underlying mechanism persists even when the producer has access to a hedging market, although an explicit threshold equivalent to $\eta^*$ for optimal delegation is not accessible in the presence of a forward market.

The benefits of delegation, as illustrated in Panel (b), are evaluated under two distinct scenarios: with and without a forward market. The relative size of these gains is contingent on the agent's level of risk aversion. Due to the sensitivity of this ranking to various other parameters, it is not possible to derive a clear, general comparative static result.

Panel (c) elucidates why delegating to an agent, even with identical technological capabilities, can benefit the producer. We decompose the gain from delegation $v^d-v^i$ into two key components. The first is the risk-sharing premium $v^{d,fb}-v^i$, representing the difference between the first-best delegation value and the integration value. This term is consistently positive. The second component is the information rent $v^d-v^{d,fb}$, which is, by definition, always negative. Panel (c) demonstrates that as the agent's risk aversion increases, the risk-sharing premium diminishes, while the absolute value of the information rent concurrently rises. Eventually, the cost associated with this information rent outweighs the benefits derived from risk-sharing, rendering delegation no longer advantageous.

\paragraph{Optimal Production} The optimal control for the in-house producer is given by
$\alpha^\ast_t = 2\gamma (K^{\rm i}_2(t) X_t)_2$, i.e. the second component of the vector $(K^{\rm i}_2(t) X_t)$. Denoting $K^{\rm i}_{ij}$ the $i,j$ coefficient of the matrix $K^{\rm i}_2$ and dropping the deterministic dependence on time for these matrix coefficients, it writes 
$\alpha^\ast_t = 2 \gamma  (K^{\rm i}_{21} S_t + K^{\rm i}_{22} Q_t)$. Therefore, the computation of the first two moments of the optimal production as well as its covariance with the spot price is explicit and given by the solutions of the following three odes.

\begin{align*}
d\E[Q_t]       & = 2 \gamma \big( K^{\rm i}_{21} S_0 + K^{\rm i}_{22} \E[Q_t] \big) dt, \\
d\E[Q_t^2]   & = 4\gamma \big( K^{\rm i}_{21} \E[S_t Q_t] + K^{\rm i}_{22} \E[Q_t^2] \big) dt + \sigma^2 dt\\
d\E[S_t Q_t] & = 2\gamma \big(  K^{\rm i}_{21} (S_0^2 + \nu^2 t) + K^{\rm i}_{22} \E[S_t Q_t]  \big) dt + \rho\sigma\nu dt.
\end{align*}

Figure \ref{fig:optProd} below shows the evolution of the optimal expected production as a function of the agent's risk-aversion parameter $\eta_a$.

\begin{figure}[!thb]
\begin{center}
\begin{tabular}{c}
 (d)\\
 \includegraphics[width=0.45\textwidth]{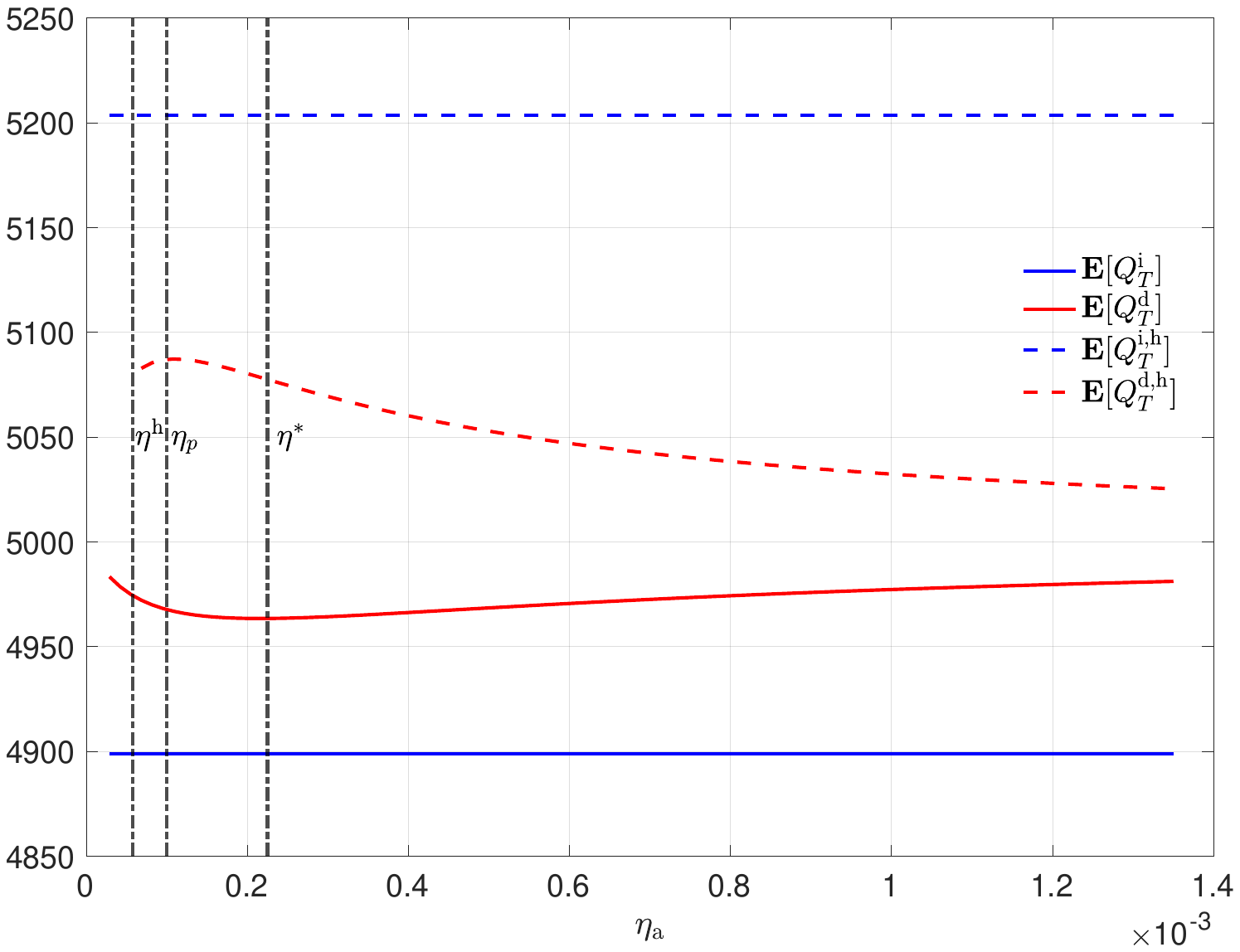}
\end{tabular}
\caption{{\small Optimal expected production as a function of the agent's risk-aversion parameter $\eta_a$. }}
\label{fig:optProd}
\end{center}
\end{figure}

The introduction of a forward market consistently leads to higher production levels (the dotted curves are always above the solid curves), a finding that is unambiguous. However, the magnitude of this positive effect is diminished when the principal delegates production. The principal's access to a forward market results in a lower optimal production level under delegation compared to in-house production. Thus, delegation negatively impacts the expected optimal output. The agent's risk aversion parameter has a distinct and opposite effect depending on the availability of hedging. When the principal cannot hedge, production rises with the agent's risk aversion; conversely, when the principal can hedge in the forward market, production falls as the agent's risk aversion increases.

\paragraph{Spot and Equilibrium forward prices}

Figure \ref{fig:prices} below illustrates the evolution of both spot and forward prices under two distinct scenarios: in-house and delegated production.

\begin{figure}[!thb]
\begin{center}
\begin{tabular}{c}
 (e) \\
 \includegraphics[width=0.45\textwidth]{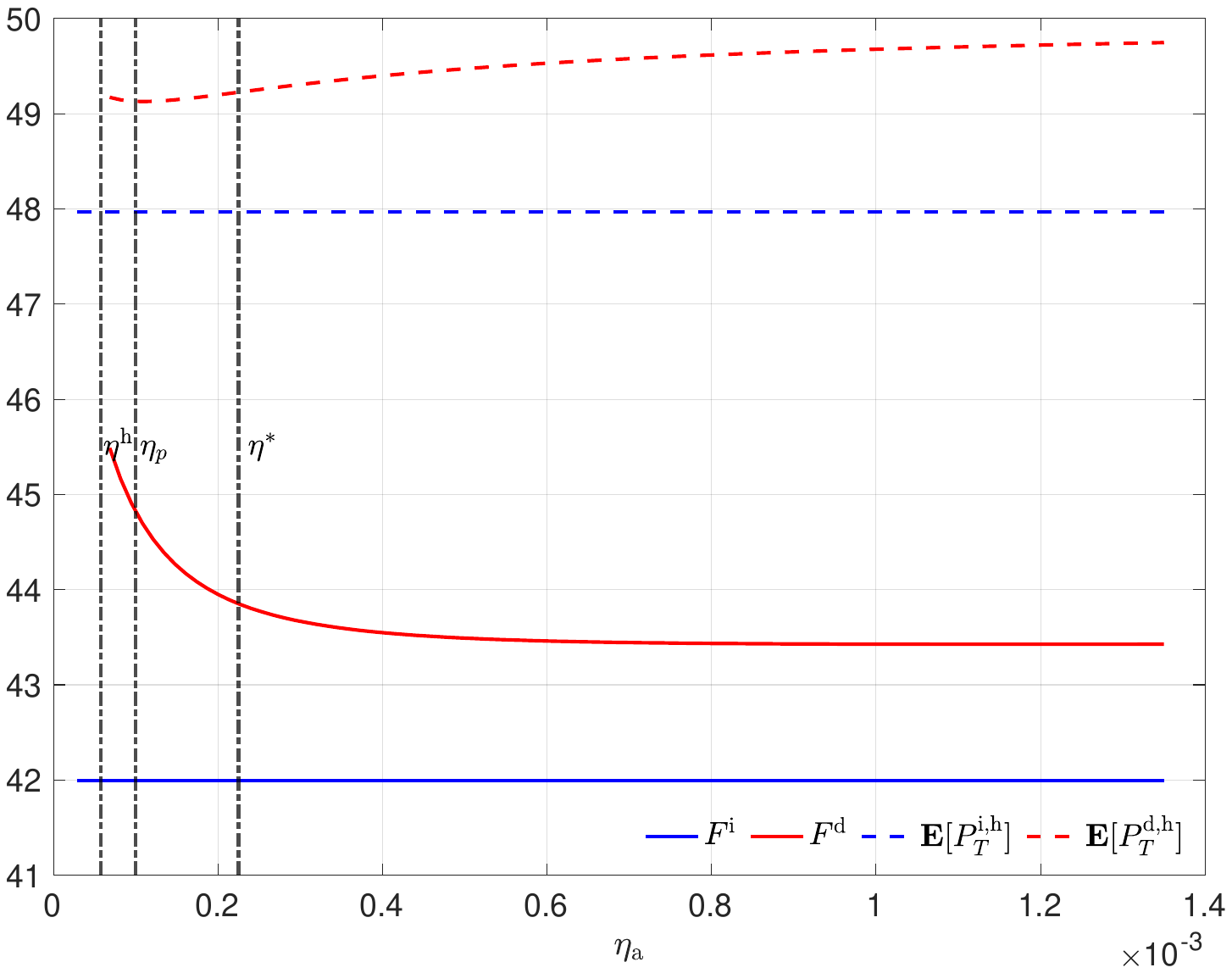}
\end{tabular}
\caption{{\small Spot and Forward prices as a function of the agent's risk-aversion parameter $\eta_a$}}
\label{fig:prices}
\end{center}
\end{figure}

We find that the impact of forward market access on spot prices is, as expected, the inverse of its effect on production. Specifically, delegation raises expected spot prices, with the magnitude of this increase being positively correlated with the agent's risk aversion parameter. Conversely, the equilibrium forward price is negatively correlated with this same parameter. This leads to an increase in the risk premium demanded by the market maker as the agent's risk aversion rises. The market maker's rational expectations regarding the optimal contract mean that the forward price remains highly sensitive to the agent's risk aversion, even though the contract is signed with the principal.

\paragraph{Forward Position}

Figure \ref{fig:position} below illustrates the evolution of the forward positions under two distinct scenarios: in-house and delegated production.

\begin{figure}[!thb]
\begin{center}
\begin{tabular}{c}
 (f) \\
 \includegraphics[width=0.45\textwidth]{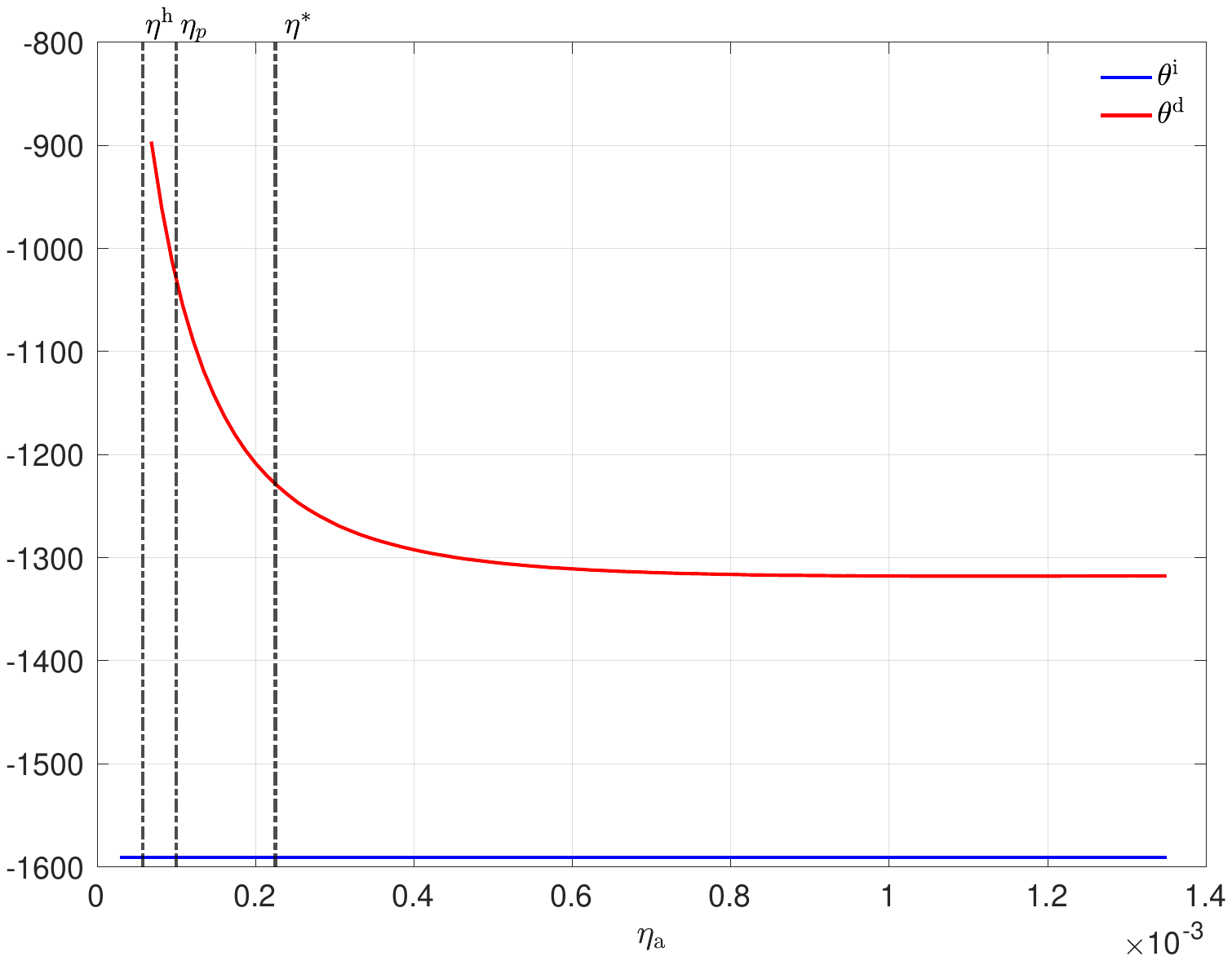}
\end{tabular}
\caption{{\small Forward positions as a function of the agent's risk-aversion parameter $\eta_a$}}
\label{fig:position}
\end{center}
\end{figure}

\hs

An in-house owner's forward market sales exceed those of an owner who delegates. When the principal delegates, the quantity sold on the forward market is positively correlated with the agent's risk aversion parameter. This holds true even as the risk premium increases (as seen in Figure \ref{fig:prices}), because the benefits of risk sharing with the agent become less pronounced with higher agent risk aversion.

\hs

\paragraph{Contract}

In this paragraph, we investigate whether the theoretical result established in Section~4.3---namely, that in the absence of demand uncertainty a hedged principal pays the agent less in expectation---remains robust when demand risk is reintroduced (as $\nu = 5>0$), at least locally around the monopoly production level $q_m$ and for values of $\eta_a$ such that $\eta_p<\eta_a<\eta^\star$. Figure~\ref{fig:compensation} provides clear support for this conclusion.

\begin{figure}[!thb]
\begin{center}
\begin{tabular}{ccc}
(a) & (b) & (c) \\
\includegraphics[width=0.3\textwidth]{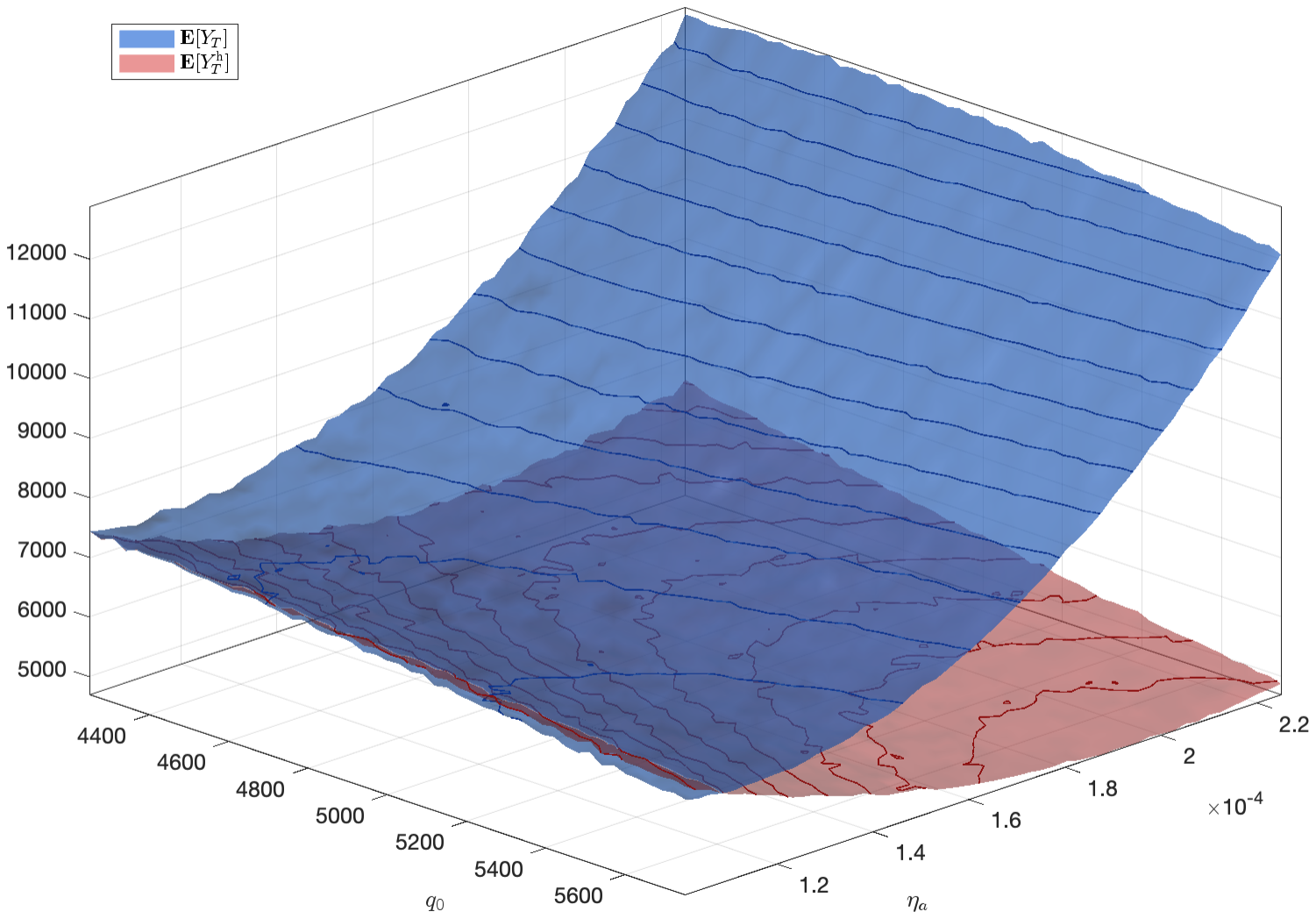} & \includegraphics[width=0.3\textwidth]{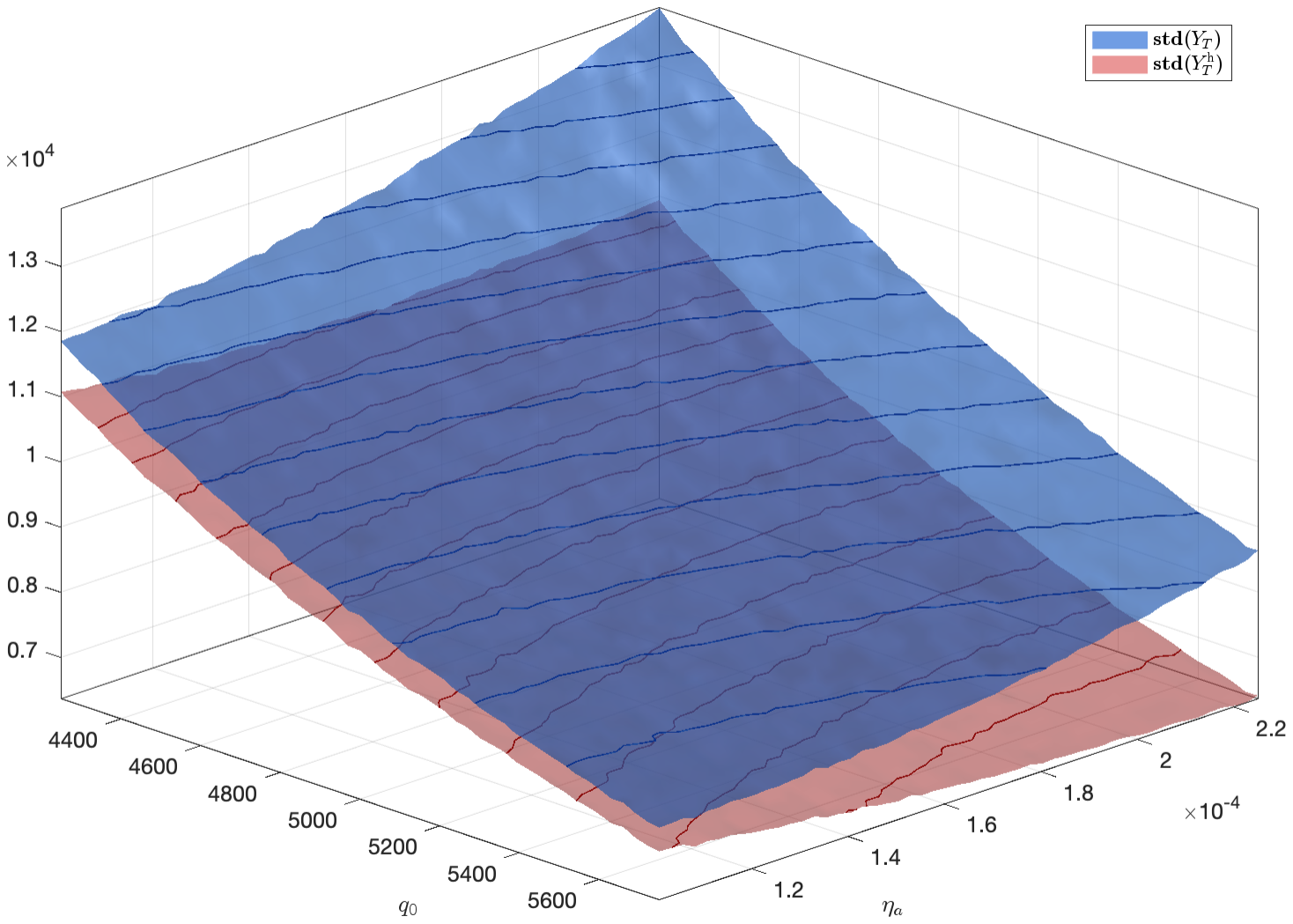} &
\includegraphics[width=0.3\textwidth]{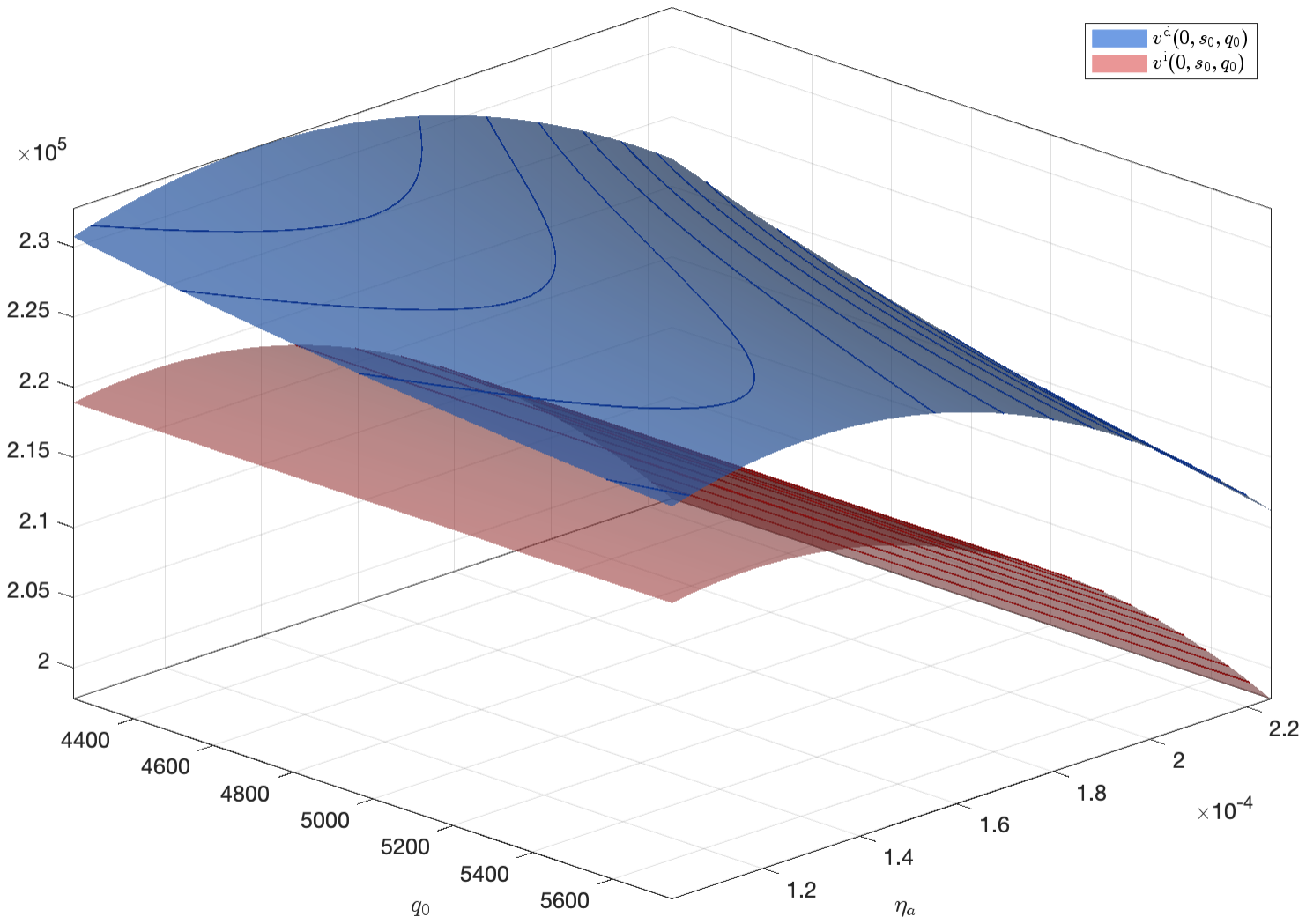} 
\end{tabular}
\caption{{\small As function of the agent's risk-aversion parameter $\eta_a$ and initial condition on $q_0$} (a) the expectations of agent's compensation with and without hedge (b) their standard deviations (c) the principal certainty equivalent values $v^{\rm i}(0,x_0)$ and $v^{\rm d}(0,x_0)$.}
\label{fig:compensation}
\end{center}
\end{figure}

Panel~(a) shows that, in this region of parameters, the hedged principal continues to pay less in expectation than the unhedged principal. This difference becomes small for sufficiently low values of $\eta_a$, but the ordering remains the same. The availability of a forward market for the principal leads to a reduction in the agent's expected compensation. This is due to the principal's ability to hedge a fraction of the output, which diminishes the degree of risk-sharing with the agent.  

As shown on Panel~(b) this hedging activity also results in a lower variance of the agent's compensation, an outcome that is welfare-enhancing for a risk-averse agent. This panel shows that the reduction in expected compensation does not come at the expense of higher payment risk for the agent. Finally, panel~(c) confirms that, for the parameter values under consideration, delegation is indeed valuable for the principal, as the delegated certainty equivalent dominates the in-house benchmark, i.e. $v^{\rm i}(0,x_0)<v^{\rm d}(0,x_0)$.

Taken together, these numerical results indicate that the compensation effect highlighted in Section~4.3 is not confined to the constant-demand benchmark, but remains robust in the economically relevant region where initial production is close to $q_m$ and for a pertinent range of risk-aversion parameter of the agent relative to the principal.

\section{Conclusion}

This paper studies how a principal’s access to a futures market interacts with delegation, risk sharing, and equilibrium pricing in a monopolistic production problem under demand uncertainty. Within a tractable continuous-time CARA framework, we compare in-house production with second-best delegation and characterize the producer’s optimal production and hedging decisions in equilibrium.

Our analysis delivers three main messages. First, delegation may remain optimal even when the agent is more risk-averse than the principal, provided that the participation constraint is not too tight. The key mechanism is the trade-off between a positive risk-sharing benefit and a negative information-rent component, which together determine the range of parameters for which delegation is valuable. Second, access to the forward market modifies both the firm’s internal allocation of risk and the equilibrium forward price. Because delegation already provides an internal risk-sharing channel, the principal hedges less under delegation than under in-house production, and this lower hedging demand translates into higher forward prices in equilibrium. Third, the availability of hedging weakens the risk-sharing role of the compensation scheme. In the constant-demand case, we show that this is associated with lower expected compensation for the agent under delegation. Our numerical results further indicate that this mechanism remains robust in the presence of demand uncertainty.

More broadly, our results show that external risk transfer through financial markets feeds back into internal organizational design. In our framework, the possibility of hedging part of the firm’s exposure on the futures market reduces the amount of price risk that must be borne within the principal-agent relationship, thereby changing both compensation design and market outcomes. This highlights a tight interaction between corporate risk management, incentive provision, and firm boundaries.

These results also contribute to the literature on corporate hedging and agency costs (e.g., \cite{Nance}) by showing how access to futures markets reshapes internal risk allocation and compensation design in a delegated production environment.

We finally examine a dynamic extension in which the principal can hedge continuously on the futures market. Under a linear-equilibrium ansatz, the numerical analysis indicates that the term structure of futures-price volatility depends on the correlation between demand and supply shocks. In this sense, the model is consistent with the view that the Samuelson effect is a conditional phenomenon driven by the resolution of uncertainty over time rather than a universal feature of commodity markets.

Several extensions appear natural. On the theoretical side, it would be useful to establish general existence and uniqueness results for the dynamic linear equilibrium. On the empirical side, the model suggests testable implications for the joint determination of firms’ hedging activity, compensation policies, and futures risk premia.

\bibliographystyle{plain}

\section{Appendix}
In this section, we denote 
$$
A := \begin{pmatrix} \nu^2 & \rho\sigma\nu \\ \rho\sigma\nu & \sigma^2 \end{pmatrix}, \,  M:=  \begin{pmatrix}  0 & \frac12 \\ \frac12 & -\varepsilon \end{pmatrix}.
$$

\subsection{Unhedgeable Risk-in-house Producer} \label{proofoptiprodnohedge}

The proof consists in building a smooth solution to the HJB equation \eqref{HJB}. To alleviate notation, we use $V$ for $V^{\rm i}$. Let us define the certainty equivalent $v(t,s,q)$, defined by $V(t,s,q) =: U_P(v(t,s,q))$, which satisfies
\begin{align*}
-\partial_tv &= \frac12 \sigma^2 \big( v_{qq} - \eta_pv_q^2 \big) + \frac12 \nu^2 \big( v_{ss} - \eta_pv_s^2 \big) 
 + \rho\sigma\nu(v_{qs} - \eta_p v_qv_s) + \sup_{a} \big\{a v_q - c(a)\big\}, \\ 
& = \frac12 \sigma^2 \big( v_{qq} - \eta_pv_q^2 \big) + \frac12 \nu^2 \big( v_{ss} - \eta_pv_s^2 \big) +  \rho\sigma\nu(v_{qs} - \eta_p v_qv_s) 
+  \frac12   \gamma  v_q^2,   \\
& = \frac12 \sigma^2   v_{qq}  + \frac12 \nu^2 \big( v_{ss} - \eta_pv_s^2 \big) +  \rho\sigma\nu(v_{qs} - \eta_p v_qv_s) 
-  \frac12   (\eta_p \sigma^2 -\gamma) v_q^2, 
 \end{align*}
together with the terminal condition $v(T,s,q) = \pi(s,q)q$. From the first line, we furthermore deduce the optimal production effort $a^*(t,s,q) := \gamma  v_q(t,s,q).$ Plugging the optimal control in the PDE above, we obtain the more compact form
\begin{align}\label{eq:v}
&0 = \partial_t v + \frac12 \Tr{ A D^2v} 
-\frac12 Dv \mcdot \Lambda_{\rm i}Dv, \text{ with }v(T,x) = x \mcdot M x\\
& \text{where}\quad x := \begin{pmatrix}s \\ q \end{pmatrix}, 
\Lambda_{\rm i} := \eta_p A - \gamma E_{22} 
= \begin{pmatrix} \eta_p \nu^2 & \eta_p \rho\sigma\nu \\ \eta_p \rho\sigma\nu & \eta_p\sigma^2-\gamma \end{pmatrix}.
\end{align}
Facing a standard linear-quadratic problem, we guess that $v$ has the quadratic form 
\be\label{v:quadra}
v(t,x) 
&=& 
K_0(t) + K_1(t)\cdot x + x \cdot K_2(t) x
\ee 
with $K_1(t)$ a vector of dimension~2 and $K_2$ a symmetric $2\times 2$ matrix. Then, it holds that $Dv = K_1  + 2 K_2(t) x$ and $D^2v = 2 K_2$. Plugging the guess into Equation \eqref{eq:v}, we get
\begin{align*}
0& = \dot K_0 + \dot K_1 \cdot x + x\cdot \dot K_2 x  \\
& + \Tr{A K_2} 
- \frac12 K_1 \mcdot \Lambda_{\rm i} K_1 
- 2x \mcdot K_2\Lambda_{\rm i} K_1 
- 2 x  \mcdot K_2 \Lambda_{\rm i} K_2 x.
\end{align*}
Identification and a little matrix algebra lead to a system of three equations:
\begin{align}\label{eq:K} 
\dot K_2& = 2 K_2 \Lambda_{\rm i} K_2, \quad K_2(T) = M, \\
\dot K_1& = 2K_2 \Lambda_{\rm i} K_1, \quad K_1(T) =0, \text{ which implies } K_1=0,\\
K_0(t)& = \int_t^T \Tr{ A K_2(u)} du.
\end{align}
To solve the first equation, let us assume that the matrix $K_2$ is invertible for all $t \in [0,T]$. Observing that $K_2^{-1} \dot K_2 = - \dot{(K_2^{-1})} K_2$, we obtain 
$$
\dot{K_2^{-1}}=-2\Lambda_{\rm i}.
$$
By integrating, we obtain  
\b*
K_2(t) & =& \big( M^{-1} + 2(T-t)\Lambda_{\rm i}\big)^{-1} = \big( I + 2(T-t)M\Lambda_{\rm i}\big)^{-1}M. 
\e*
It remains to give a set of sufficient conditions that ensures the matrix $I+2(T-t) M\Lambda_{\rm i}$ is invertible. Setting $\tau := T-t$, we introduce 
\begin{align*}
N_i(\tau) := I + 2\tau M\Lambda_{\rm i}= 
\begin{pmatrix}
1 + \eta_p\rho\sigma\nu \tau & \gamma_i \tau \\
&\\
(\eta_p\nu^2 - 2\eps\eta_p\rho\sigma\nu)\tau & 1 + (\eta_p\rho\sigma\nu - 2\eps\gamma_i)\tau
\end{pmatrix}
\end{align*}
which admits the determinant
\begin{align*}
\Delta_i(t) := \det{N_i(\tau)} & = 1 + 2\big[ \eta_p\rho\sigma\nu - \eps\gamma_i\big]\tau
+\big[ \eta_p\rho\sigma\nu(\eta_p\rho\sigma\nu - 2\eps\gamma_i) -\gamma_i(\eta_p\nu^2 - 2\eps\eta_p\rho\sigma\nu) \big]\tau^2, \\
& = 1 + 2\big[ \eta_p\rho\sigma\nu - \eps\gamma_i\big]\tau
+ \eta_p\nu^2\big[\eta_p \rho^2\sigma^2 -  \gamma_i\big]\tau^2\\
&=1+ 2\big[ \eta_p\rho\sigma\nu + \eps(\gamma-\eta_p\sigma^2)\big]\tau
+ \eta_p\nu^2\big[\gamma-\eta_p(1- \rho^2)\sigma^2\big]\tau^2
\end{align*}
Thus, a sufficient condition to have a strictly positive determinant is $\gamma > \eta_p \sigma^2$. Under this condition, we have
\begin{align*}
N_i(\tau)^{-1} = \frac1{\Delta_i(t)}
\begin{pmatrix}
 1 + (\eta_p\rho\sigma\nu - 2\eps\gamma_i)\tau  & -\gamma_i \tau \\
 &\\
(2\eps\eta_p\rho\sigma\nu - \eta_p\nu^2)\tau &1 + \eta_p\rho\sigma\nu \tau 
\end{pmatrix}
\end{align*}
and thus
\begin{align*}
K_2^i(t) &= N_i(\tau)^{-1} 
\begin{pmatrix}
0 & \frac12 \\
\frac12 &-\eps 
\end{pmatrix}
=
 \frac1{\Delta_i(t)}
 \begin{pmatrix}
-\frac12 \gamma_i \tau & \frac12 + \frac12 (\eta_p\rho\sigma\nu - 2\eps\gamma_i)\tau  + \eps\gamma_i\tau\\
&\\
\frac12 + \frac12 \eta_p\rho\sigma\nu \tau & \frac12(2\eps\eta_p\rho\sigma\nu - \eta_p\nu^2)\tau -\eps - \eps\eta_p\rho\sigma\nu \tau 
\end{pmatrix}, \\
&\\
& = 
 \frac1{\Delta_i(t)}
 \begin{pmatrix}
-\frac12 \gamma_i \tau & \frac12 + \frac12 \eta_p\rho\sigma\nu \tau\\
&\\
\frac12 + \frac12 \eta_p\rho\sigma\nu \tau & -\frac12 \eta_p\nu^2 \tau -\eps  
\end{pmatrix},\\
&=\frac1{\Delta_i(t)} (M+\frac{\eta_p}{2}\tau A^{-1}).
\end{align*}
Finally, we have an explicit formulation for the optimal level of effort given by $\gamma v_q(t,s,q)$ where
\begin{align*}
v_q(t,s,q) & = 
\frac1{\Delta_i(t)} \big[ (1 +  \eta_p\rho\sigma\nu(T-t)) s - (\eta_p\nu^2 (T-t) + 2\eps) q\big], \\
& = \frac{2\eps + \eta_p\nu^2 (T-t)}{\Delta_i(t)}\Big[ \frac{1 +  \eta_p\rho\sigma\nu(T-t)}{1+\frac{\eta_p\nu^2}{2\eps} (T-t)} \frac{s}{2\eps} - q  \Big]
\end{align*}

\subsection{Unhedgeable Risk -- Delegation}\label{proofdelegationnohedge}

We follow the same method as in Proposition~\ref{OptiProd-no-hedge}. Let $C_p:=\eta_p A+\gamma E_{22},$
$\bar C := (\eta_a+\eta_p)A+\gamma E_{22}.$ The principal's value function solves
\begin{align*}
-\partial_t V = \frac12 Tr(AD^2V) + \sup_{z}\Big\{ C_p z \cdot DV + \frac12 \eta_p\, z\cdot \bar C z\, V \Big\},
\qquad
V(T,x)=U_P(x\cdot Mx).
\end{align*}
Writing \(V=U_P(v)\), the certainty equivalent \(v\) satisfies
\begin{align*}
-\partial_t v = \frac12 Tr \Big(A(D^2v-\eta_p Dv Dv^\top)\Big) + \sup_{z}\Big\{ C_p z\cdot Dv-\frac12 z\cdot \bar C z \Big\},
\qquad
v(T,x)=x\cdot Mx.
\end{align*}
The maximizer is
\[
\hat z=\bar C^{-1}C_p Dv.
\]
A direct computation gives
\begin{align*}
&\bar C^{-1} = \frac1{\tilde\gamma} \begin{pmatrix}
\frac{\bar\gamma}{(\eta_a+\eta_p)\nu^2} & -\frac{\rho\sigma}{\nu}\\[0.3em]
-\frac{\rho\sigma}{\nu} & 1
\end{pmatrix},
\qquad
\bar C^{-1}C_p = \frac1{\tilde\gamma}
\begin{pmatrix}
\frac{\eta_p}{\eta_a+\eta_p}\tilde\gamma
&
-\frac{\rho\sigma}{\nu}\frac{\eta_a}{\eta_a+\eta_p}\gamma\\[0.3em]
0
&
\gamma+\eta_p(1-\rho^2)\sigma^2
\end{pmatrix}, \\
\text{where} & \quad \tilde\gamma := \gamma + (\eta_a+\eta_p)(1-\rho^2)\sigma^2,
\end{align*}
and therefore
\[
C_p\bar C^{-1}C_p = \frac1{\tilde\gamma}
\begin{pmatrix}
\frac{\eta_p^2\nu^2}{\eta_a+\eta_p}\tilde\gamma
&
\frac{\eta_p^2}{\eta_a+\eta_p}\rho\sigma\nu\,\tilde\gamma\\[0.3em]
\frac{\eta_p^2}{\eta_a+\eta_p}\rho\sigma\nu\,\tilde\gamma
&
(\gamma+\eta_p\sigma^2)(\gamma+\eta_p(1-\rho^2)\sigma^2)
-\frac{\eta_a\eta_p}{\eta_a+\eta_p}\rho^2\sigma^2\gamma
\end{pmatrix}.
\]
Introduce
\[
\Lambda_d:=\eta_p A-C_p\bar C^{-1}C_p.
\]
Then
\[
\Lambda_d=\bar\eta A-\gamma_d E_{22},
\qquad
\bar\eta:=\frac{\eta_a\eta_p}{\eta_a+\eta_p},
\qquad
\gamma_d
:=
\gamma\Bigg(
1-\frac{(1-\rho^2)\sigma^2\eta_a^2}{\tilde\gamma(\eta_a+\eta_p)}
\Bigg).
\]
Hence \(v\) solves
\begin{align*}
-\partial_t v = \frac12 Tr(AD^2v)-\frac12 Dv\cdot \Lambda_d Dv, \qquad v(T,s,q)=q(s-\eps q).
\end{align*}
This is the same equation as \eqref{eq:v}, with \(\Lambda_d\) in place of \(\Lambda_i\). We therefore look for a quadratic solution of the form
\[
v(t,x)=K_0(t)+x\cdot K_2(t)x.
\]
Set \(\tau:=T-t\). Then
\[
K_2(t)=\big(I+2\tau M\Lambda_d\big)^{-1}M.
\]
Writing
\begin{align*}
N(\tau):=I+2\tau M\Lambda_d =
\begin{pmatrix}
1+\bar\eta\rho\sigma\nu\tau & \beta\tau\\
(\bar\eta\nu^2-2\eps\bar\eta\rho\sigma\nu)\tau &
1+(\bar\eta\rho\sigma\nu-2\eps\beta)\tau
\end{pmatrix},
\quad \text{with} \quad \beta:=\bar\eta\sigma^2-\gamma_d,
\end{align*}
its determinant is
\[
\Delta_d(t) := \det N(\tau) =
1+2(\bar\eta\rho\sigma\nu-\eps\beta)\tau
+\bar\eta\nu^2(\bar\eta\rho^2\sigma^2-\beta)\tau^2.
\]
Under the sufficient condition \(\gamma_d>\bar\eta\sigma^2\), one has \(\Delta_d(t)>0\) on \([0,T]\), and
\[
N(\tau)^{-1} = \frac1{\Delta_d(t)}
\begin{pmatrix}
1+(\bar\eta\rho\sigma\nu-2\eps\beta)\tau & -\beta\tau\\
(2\eps\bar\eta\rho\sigma\nu-\bar\eta\nu^2)\tau & 1+\bar\eta\rho\sigma\nu\tau
\end{pmatrix}.
\]
Therefore
\[
K_2(t) = \frac1{\Delta_d(t)}
\begin{pmatrix}
-\frac12\beta\tau & \frac12+\frac12\bar\eta\rho\sigma\nu\tau\\[0.3em]
\frac12+\frac12\bar\eta\rho\sigma\nu\tau & -\eps-\frac12\bar\eta\nu^2\tau
\end{pmatrix}.
\]
We deduce
\[
v_q(t,s,q) = \frac{2\eps+\bar\eta\nu^2(T-t)}{\Delta_d(t)} \Big[\ell_d(t)\,q_m(s)-q\Big],
\qquad
\ell_d(t) := \frac{1+\bar\eta\rho\sigma\nu(T-t)}{1+\frac{\bar\eta\nu^2}{2\eps}(T-t)},
\qquad
q_m(s):=\frac{s}{2\eps}.
\]
Since
\[
\hat z=\bar C^{-1}C_p Dv = 2\bar C^{-1}C_pK_2(t)X,
\]
the second component is
\[
\hat z_q(t,s,q) = \frac1{\Delta_d(t)} \frac{\gamma+\eta_p(1-\rho^2)\sigma^2}{\gamma+(\eta_a+\eta_p)(1-\rho^2)\sigma^2}
\Big[ (1+\bar\eta\rho\sigma\nu(T-t))\,s - (2\eps+\bar\eta\nu^2(T-t))\,q \Big].
\]
Using \(\hat\alpha(z)=\gamma z_q\), the optimal production dynamics are
\[
dQ_t = k_d(t)\Big[\ell_d(t)\,q_m(S_t)-Q_t\Big]dt + \sigma\,dW_t^\perp,
\]
where
\[
k_d(t) = \gamma \frac{\gamma+\eta_p(1-\rho^2)\sigma^2}{\gamma+(\eta_a+\eta_p)(1-\rho^2)\sigma^2} \frac{2\eps+\bar\eta\nu^2(T-t)}{\Delta_d(t)}.
\]
This proves Proposition~\ref{Delegation-no-hedge}.

\subsection{Proof of Proposition \ref{Integrated-hedge}} \label{proofintegratedhedge}

From Appendix~\ref{proofoptiprodnohedge}, we have
\begin{align*}
K_2^i(t)
=
\frac{1}{\Delta_i(t)}
\begin{pmatrix}
-\frac12(\eta_p\sigma^2-\gamma)(T-t) &
\frac12\big(1+\eta_p\rho\sigma\nu(T-t)\big)\\
\frac12\big(1+\eta_p\rho\sigma\nu(T-t)\big) &
-\epsilon-\frac12\eta_p\nu^2(T-t)
\end{pmatrix}.
\end{align*}
Let \(\bar\epsilon:=(\epsilon,1)^\top\). Then
\begin{align*}
2\bar\epsilon\cdot K_2^i(0)\bar\epsilon
&=
\frac{1}{\Delta_i(0)}
\Big(
-\epsilon^2(\eta_p\sigma^2-\gamma)T
+2\epsilon\big(1+\eta_p\rho\sigma\nu T\big)
-2\epsilon-\eta_p\nu^2T
\Big)\\
&=
\frac{T}{\Delta_i(0)}
\Big(
\epsilon^2\gamma-\eta_p(\epsilon^2\sigma^2+\nu^2-2\epsilon\rho\sigma\nu)
\Big)
=
\frac{T}{\Delta_i(0)}\big(\epsilon^2\gamma-\eta_p\sigma_p^2\big),
\end{align*}
where $\sigma_p^2:=\epsilon^2\sigma^2+\nu^2-2\epsilon\rho\sigma\nu$
is the variance coefficient of \(P_t=S_t-\epsilon Q_t\). Similarly,
\begin{align*}
2\bar\epsilon\cdot K_2^i(0)X_0
&=
\frac{1}{\Delta_i(0)}
\Big(
-\epsilon(\eta_p\sigma^2-\gamma)TS_0
+\epsilon\big(1+\eta_p\rho\sigma\nu T\big)Q_0\\
&\hspace{5em}
+\big(1+\eta_p\rho\sigma\nu T\big)S_0
-(2\epsilon+\eta_p\nu^2T)Q_0
\Big)\\
&=
\frac{1}{\Delta_i(0)}
\Big(
\big(1+(\eta_p\rho\sigma\nu-\epsilon(\eta_p\sigma^2-\gamma))T\big)P_0
+\big(\epsilon^2\gamma-\eta_p\sigma_p^2\big)Q_0
\Big),
\end{align*}
with \(P_0:=S_0-\epsilon Q_0\).

By equation~\eqref{eq:int-hedge},
\begin{align*}
\theta_i(F)
=
\frac{F-2\bar\epsilon\cdot K_2^i(0)X_0}{2\bar\epsilon\cdot K_2^i(0)\bar\epsilon}
=
\frac{\Delta_i(0)F-\big(1+(\eta_p\rho\sigma\nu-\epsilon(\eta_p\sigma^2-\gamma))T\big)P_0}
{T(\epsilon^2\gamma-\eta_p\sigma_p^2)}
-Q_0.
\end{align*}

The optimal effort in the hedged problem is obtained by shifting the state variable:
\begin{align*}
\hat\alpha(t,s,q)
=
\gamma\,\partial_q v_i\big(t,s+\epsilon\theta_i(F),q+\theta_i(F)\big).
\end{align*}
Using the expression of \(\partial_q v_i\) from Appendix~\ref{proofoptiprodnohedge},
\begin{align*}
\hat\alpha(t,s,q)
&=
\frac{\gamma}{\Delta_i(t)}
\Big(
\big(1+\eta_p\rho\sigma\nu(T-t)\big)\big(s+\epsilon\theta_i(F)\big)
-\big(2\epsilon+\eta_p\nu^2(T-t)\big)\big(q+\theta_i(F)\big)
\Big)\\
&=
\frac{\gamma\big(2\epsilon+\eta_p\nu^2(T-t)\big)}{\Delta_i(t)}
\left[
\frac{1+\eta_p\rho\sigma\nu(T-t)}{2\epsilon+\eta_p\nu^2(T-t)}
\big(s+\epsilon\theta_i(F)\big)-q-\theta_i(F)
\right].
\end{align*}
Therefore, with
\begin{align*}
k_i(t):=\gamma\frac{2\epsilon+\eta_p\nu^2(T-t)}{\Delta_i(t)},
\qquad
\ell_i(t):=\frac{1+\eta_p\rho\sigma\nu(T-t)}
{1+\frac{\eta_p\nu^2}{2\epsilon}(T-t)},
\qquad
q_m(s,\theta):=\frac{s+\epsilon\theta}{2\epsilon},
\end{align*}
we obtain
\begin{align*}
\hat\alpha(t,s,q)
=
k_i(t)\Big(\ell_i(t)\,q_m(s,\theta_i(F))-q-\theta_i(F)\Big).
\end{align*}
Since \(dQ_t=\hat\alpha(t,S_t,Q_t)\,dt+\sigma\big(\rho\,dW_t^\perp+\sqrt{1-\rho^2}\,dW_t\big)\), it follows that
\begin{align*}
dQ_t^*
=
k_i(t)\Big(\ell_i(t)\frac{S_t+\epsilon\theta_i(F)}{2\epsilon}-Q_t^*-\theta_i(F)\Big)\,dt
+\sigma\big(\rho\,dW_t^\perp+\sqrt{1-\rho^2}\,dW_t\big),
\end{align*}
or equivalently,
\begin{align*}
dQ_t^*
=
k_i(t)\Big(\ell_i(t)\Big(\frac{S_t}{2\epsilon}+\frac{\theta_i(F)}{2}\Big)-Q_t^*-\theta_i(F)\Big)\,dt
+\sigma\big(\rho\,dW_t^\perp+\sqrt{1-\rho^2}\,dW_t\big).
\end{align*}
Hence
\begin{align*}
\alpha_{i,h}^*(t,s,q)
&=
k_i(t)\Big(\ell_i(t)\Big(\frac{s}{2\epsilon}+\frac{\theta_i(F)}{2}\Big)-q-\theta_i(F)\Big)\\
&=
k_i(t)\Big(\ell_i(t)\frac{s}{2\epsilon}-q\Big)+k_i(t)\Big(\frac{\ell_i(t)}{2}-1\Big)\theta_i(F)\\
&=
\alpha_i^*(t,s,q)+k_i(t)\Big(\frac{\ell_i(t)}{2}-1\Big)\theta_i(F).
\end{align*}

\subsection{Proof of Proposition \ref{Delegation-hedge}} \label{proofdelegationdhedge}

Relying on the symmetry argument of the two organizational structure, the proof follows the same lines as Proposition~\ref{Integrated-hedge}.

\hs

\subsection{Case of constant demand}\label{app:nu=0}

In this subsection, we prove the Proposition~\ref{prop:nu=0}.

\noindent{\rm (i)}  The delegated producer's certainty equivalent reduces to
\[
v_d(0,s,q)
=
\frac{s^2}{4\eps}
-\frac{\eps}{\Delta_d(0)}
\left(q-\frac{s}{2\eps}\right)^2
+\frac{\sigma^2}{2\beta_d}\log \Delta_d(0),
\]
with $\Delta_d(0)=1-2\eps\beta_d T.$ Hence the producer solves
\[
\sup_{\theta\in\R}\Big\{v_d(0,S_0+\eps\theta,q_0+\theta)-\theta F\Big\}.
\]
Differentiating with respect to \(\theta\) yields
\[
F
=
\frac{(1-\eps\beta_d T)S_0-\eps q_0}{\Delta_d(0)}
-
\frac{\eps^2\beta_d T}{\Delta_d(0)}\,\theta.
\]
On the other hand, under a static position \(\theta\), the optimal production satisfies
\[
dQ_t^{\rm h}
=
k_d(t)\left(q_m-\frac{\theta}{2}-Q_t^{\rm h}\right)\,dt+\sigma\,dW_t,
\qquad
k_d(t):=\gamma\lambda_d\frac{2\eps}{\Delta_d(t)},
\]
so that
\[
Q_T^{\rm h}
=
\Phi(T)q_0+\bigl(1-\Phi(T)\bigr)\left(q_m-\frac{\theta}{2}\right)
+\sigma\Phi(T)\int_0^T \Phi(t)^{-1}\,dW_t,
\]
where
\[
\Phi(t):=\exp\!\left(-\int_0^t k_d(u)\,du\right).
\]
Therefore,
\[
\E[P_T]
=
S_0-\eps\E[Q_T^{\rm h}]
=
\frac{S_0}{2}
+\Phi(T)\left(\frac{S_0}{2}-\eps q_0\right)
+\frac{\eps}{2}\bigl(1-\Phi(T)\bigr)\theta,
\]
and
\[
\Var(P_T)=\eps^2\Var(Q_T),
\qquad
\Var(Q_T)=\sigma^2\Phi(T)^2\int_0^T \Phi(t)^{-2}\,dt.
\]
The market maker's inverse demand is thus
\[
F
=
\E[P_T]+\eta_m\Var(P_T)\theta
=
\frac{S_0}{2}
+\Phi(T)\left(\frac{S_0}{2}-\eps q_0\right)
+
\left[
\frac{\eps}{2}\bigl(1-\Phi(T)\bigr)+\eta_m\eps^2\Var(Q_T)
\right]\theta.
\]
Equating the producer's and market maker's inverse demands gives
\[
\theta^\ast
=
\frac{
\eps(q_m-q_0)\bigl(1-\Delta_d(0)\Phi(T)\bigr)
}{
\eps^2\beta_d T
+
\Delta_d(0)\left[
\frac{\eps}{2}\bigl(1-\Phi(T)\bigr)+\eta_m\eps^2\Var(Q_T)
\right]
},
\]
and then
\[
F^\ast
=
\frac{S_0}{2}
+\Phi(T)\left(\frac{S_0}{2}-\eps q_0\right)
+
\left[
\frac{\eps}{2}\bigl(1-\Phi(T)\bigr)+\eta_m\eps^2\Var(Q_T)
\right]\theta^\ast.
\]

\noindent{\rm (ii)-(iii)} In the no-hedging case,
\[
\hat z(t,q)=\lambda_d\frac{2\eps}{\Delta_d(t)}(q_m-q),
\qquad
dQ_t=k_d(t)(q_m-Q_t)\,dt+\sigma\,dW_t,
\]
whereas under a static forward position \(\theta\),
\[
\hat z^{\rm h}(t,q)=\lambda_d\frac{2\eps}{\Delta_d(t)}
\left(q_m-q-\frac{\theta}{2}\right),
\qquad
dQ_t^{\rm h}=k_d(t)\left(q_m-\frac{\theta}{2}-Q_t^{\rm h}\right)\,dt+\sigma\,dW_t.
\]
Let
\[
M_t:=q_m-Q_t,
\qquad
M_t^{\rm h}:=q_m-\frac{\theta}{2}-Q_t^{\rm h}.
\]
Then
\[
dM_t=-k_d(t)M_t\,dt-\sigma\,dW_t,
\qquad
dM_t^{\rm h}=-k_d(t)M_t^{\rm h}\,dt-\sigma\,dW_t,
\]
with
\[
M_0=q_m-q_0,
\qquad
M_0^{\rm h}=q_m-q_0-\frac{\theta}{2}.
\]
Hence
\[
M_t=\Phi(t)\left(M_0-\sigma\int_0^t \Phi(u)^{-1}\,dW_u\right),
\qquad
M_t^{\rm h}=\Phi(t)\left(M_0^{\rm h}-\sigma\int_0^t \Phi(u)^{-1}\,dW_u\right),
\]
and therefore
\[
\E\big[(M_t^{\rm h})^2\big]-\E[M_t^2]
=
\Phi(t)^2\Big((M_0^{\rm h})^2-M_0^2\Big)
=
\Phi(t)^2\left(\frac{\theta^2}{4}-\theta(q_m-q_0)\right).
\]
Since
\[
Z_t=\hat z(t,Q_t)=\lambda_d\frac{2\eps}{\Delta_d(t)}M_t,
\qquad
Z_t^{\rm h}=\hat z^{\rm h}(t,Q_t^{\rm h})=\lambda_d\frac{2\eps}{\Delta_d(t)}M_t^{\rm h},
\]
it follows that
\[
\E\big[(Z_t^{\rm h})^2\big]-\E[Z_t^2]
=
\left(\lambda_d\frac{2\eps}{\Delta_d(t)}\right)^2
\Phi(t)^2
\left(\frac{\theta^2}{4}-\theta(q_m-q_0)\right).
\]
Now,  \(Y_0^{\rm h}=Y_0 = r_0\), then
\[
dY_t=\frac12(\gamma+\eta_a\sigma^2)Z_t^2\,dt+\sigma Z_t\,dW_t,
\qquad
dY_t^{\rm h}=\frac12(\gamma+\eta_a\sigma^2)(Z_t^{\rm h})^2\,dt+\sigma Z_t^{\rm h}\,dW_t,
\]
so that
\[
\E[Y_T^{\rm h}-Y_T]
=
\frac12(\gamma+\eta_a\sigma^2)
\int_0^T\Big(\E[(Z_t^{\rm h})^2]-\E[Z_t^2]\Big)\,dt.
\]
Therefore
\[
\E[Y_T^{\rm h}-Y_T]
=
\frac12(\gamma+\eta_a\sigma^2)
\left(
\frac{\theta^2}{4}-\theta(q_m-q_0)
\right)
\int_0^T
\left(\lambda_d\frac{2\eps}{\Delta_d(t)}\right)^2
\Phi(t)^2\,dt.
\]
Evaluating this identity at \(\theta=\theta^\ast\) gives
\[
\E[Y_T^{\rm h}-Y_T]
=
\frac12(\gamma+\eta_a\sigma^2)
\left(
\frac{(\theta^\ast)^2}{4}-\theta^\ast(q_m-q_0)
\right)
\int_0^T
\left(\lambda_d\frac{2\eps}{\Delta_d(t)}\right)^2
\Phi(t)^2\,dt.
\]

Furthermore, under the condition $\eta_p\eta_a\sigma^2>\gamma(\eta_a+\eta_p),$ one has \(\beta_d>0\). Moreover, the non-degeneracy condition \(\Delta_d(0)>0\) implies \(\Delta_d(0)\in(0,1)\). As shown above, $1-\Delta_d(0)\Phi(T)>0,$
hence the coefficient
\[
C:=
\frac{
\eps\bigl(1-\Delta_d(0)\Phi(T)\bigr)
}{
\eps^2\beta_d T
+
\Delta_d(0)\left[
\frac{\eps}{2}\bigl(1-\Phi(T)\bigr)+\eta_m\eps^2\Var(Q_T)
\right]
}
\]
is strictly positive, and
\[
\theta^\ast=C(q_m-q_0).
\]
Substituting this identity into
\[
\frac{(\theta^\ast)^2}{4}-\theta^\ast(q_m-q_0)
\]
yields
\[
\frac{(\theta^\ast)^2}{4}-\theta^\ast(q_m-q_0)
=
C\left(\frac{C}{4}-1\right)(q_m-q_0)^2.
\]
It remains to show that \(C<4\). Since the denominator defining \(C\) is positive, the inequality \(C<4\) is equivalent to
\[
1-\Delta_d(0)\Phi(T)
<
4\eps\beta_d T
+
2\Delta_d(0)\bigl(1-\Phi(T)\bigr)
+
4\Delta_d(0)\eta_m\eps\Var(Q_T).
\]
Using
\[
\Delta_d(0)=1-2\eps\beta_d T,
\]
the difference between the right-hand side and the left-hand side simplifies to
\[
1-\Delta_d(0)\Phi(T)+4\Delta_d(0)\eta_m\eps\Var(Q_T),
\]
which is strictly positive. Therefore $0<C<4.$ Consequently, $ \E[Y_T^{\rm h}-Y_T]<0.$ This proves the result.

\subsection{First-best delegation} \label{app:first-best}

The first-best problem of production delegation in the absence of forward hedging is the following:
\begin{align*}
& V^{\rm d, fb} := \sup_{\alpha,\xi} \E\big[ U_{\rm P}\big( - \xi   + P(S_T,Q_T) Q_T \big)\big], \quad
\text{s.t.} \quad \E\Big[ U_{\rm A}\Big(\xi  -\int_0^T c(\alpha_t) dt  \Big)\Big] \ge U_{\rm A}(r_0) =: R_0.
\end{align*}

The problem is solved using Lagrangian relaxation of the inequality constraint. The first-best problem is then equivalent to
\begin{align*}
& V^{\rm d, fb} = \inf_{\ell \ge 0} \Big\{ - \ell R_0 + \sup_{\alpha,\xi} \E\big[ U_{\rm P}\big( - \xi   + \Gc^\alpha_T \big) + \ell U_{\rm A}\big(\xi  - \Kc_T^\alpha  \big)\big] \Big\}, \quad
\Gc^\alpha_T := P(S_T,Q_T) Q_T, \quad \Kc_T^\alpha := \int_0^T c(\alpha_t) dt.
\end{align*}
This problem is exactly the same as the first-best problem in \cite{Aid22}. Applying the same computations, it holds that
\begin{align*}
& V^{\rm d, fb}  = R_0 \Big( \frac{\bar V}{R_0} \Big)^{1+ \frac{\eta_a}{\eta_p}}, \quad 
\text{with} \quad 
\bar V := \sup_{\alpha} \E\big[ \bar U\big( \Gc^\alpha_T - \Kc^\alpha_T\big) \big], \quad
\bar U(x) = -e^{-\bar \eta x}, \quad \bar \eta := \frac{\eta_a \eta_p}{\eta_a + \eta_p}
\end{align*}
The optimal control problem with value $\bar V$ is the same as the production problem without delegation but where the risk-aversion of the producer has been reduced from $\eta_p$ to $\bar\eta$. Thus, we get immediately that $\bar V = \bar U\big(\bar v(0,s,q)\big)$ where the certainty equivalent $\bar v(t,s,q)$ satisfies the PDE
\begin{align}\label{eq:vbar}
&0 = \partial_t \bar v + b \cdot D\bar v + \frac12 \Tr{ A D^2\bar v} 
-\frac12 D\bar v \mcdot \bar \Lambda D\bar v, \quad \bar v(T,x) = x \mcdot M x, \\
& \text{with}\quad
 \bar \Lambda   = \begin{pmatrix} \bar \eta \nu^2 & \bar \eta \rho\sigma\nu \\ \bar \eta \rho\sigma\nu & \bar\eta \sigma^2 -\gamma,  \end{pmatrix}. \nonumber
\end{align}
Thus we have
\begin{align*}
 V^{\rm d, fb}  
 = - e^{-\eta_a r_0} e^{(1+\frac{\eta_p}{\eta_a}) (-\frac{\eta_a\eta_p}{\eta_a+\eta_p} \bar v + \eta_a r_0)}
= - e^{ -\eta_p \bar v + \eta_p r_0} = U_{\rm P}(\bar v - r_0).
\end{align*}

Note that when $\rho^2=1$, it holds that $\bar\Lambda = \Lambda_d$ inducing $\bar v = v^{\rm sb}$.

\section{Online Appendix -- Extension to Dynamic Hedging}

This section considers a dynamic hedging extension of the model in which the principal can trade continuously in the futures market. Our objective is to study the implications of dynamic hedging for the equilibrium price dynamics and, in particular, for the term structure of futures-price volatility. To keep the analysis tractable, we focus on a linear market equilibrium in which the futures price process is affine in the state variables and clears the market at each date. Within this linear-equilibrium framework, we derive the system that characterizes equilibrium price dynamics and solve it numerically.

To mitigate demand and market fluctuations, we assume that the firm owner can trade her production on the futures market at any time. The futures price process \(F=\{F_t,\ t\in[0,T]\}\) is determined endogenously through the interaction between the Principal and a market maker, and must satisfy the terminal condition $F_T=P(S_T,Q_T).$ The Principal's financial revenues from dynamic hedging are $H_T^p:=\int_0^T \theta_t^p\, dF_t,$ where \(\theta_t^p\) is the quantity of goods that the producer buys (if positive) or sells (if negative) forward. Similarly, the market maker's financial revenues are $H_T^m:=\int_0^T \theta_t^m\, dF_t.$

We describe below the firm's problems in the in-house and delegation cases and introduce the linear-equilibrium concept used throughout this section.

\subsection{Linear market equilibrium}

As in the static hedging analysis, we study separately the in-house and delegated cases. The In-house producer's problem is defined by:
\begin{align}\label{DynIntegratedValue}
V^{\rm i}(0,s,q) & := \sup_{\alpha, \theta^{\rm i}} \E\Big[ U_P\Big( (S_T-\eps Q_T) Q_T -  \int_0^T c(\alpha_t) dt+ \int_0^T \theta^{\rm i}_t dF_t\Big)\Big], 
\end{align}
with controlled production and demand shift process dynamics:
\begin{align*}
dQ_t  &= \alpha_t dt + \sigma \big(\rho dW_t^\perp + \sqrt{1-\rho^2} dW_t\big), \quad dS_t =  \nu dW_t^\perp,\nonumber
\end{align*}
Given the producer's optimal response $\alpha^\ast(F)$, the market maker faces the optimization problem\footnote{ Recall that we use the notation $\P^{F}$ for $\P^{\alpha^\ast(F)}$}:
\begin{align}\label{DynMMValue}
V^{\rm m}=\sup_{\theta^{\rm m}}\E^F\left[U_m\Big(  \int_0^T \theta^{\rm m}_t dF_t\Big) \right].
\end{align}
The linear equilibrium condition for the futures market is formally defined as follows.

\begin{Definition}[In-House production futures market linear equilibrium]
   A quadruplet $(\alpha^\star,\theta^{\rm i}_\star,\theta^{\rm m}_\star,F_t^\star)$ is a linear equilibrium if
   \begin{itemize}
       \item Given the dynamics of futures prices $F^\star_t$, the pair $(\alpha^\star,\theta^{\rm i}_\star)$ is optimal for the producer,
       \item Given the dynamics of futures prices $F^\star_t$ and the pair $(\alpha^\star,\theta^{\rm i}_\star)$, $\theta^{\rm m}_\star$ is optimal for the market maker,
       \item The market clears at each date $t \le T$,i.e. $\theta^{\rm i}_\star+\theta^{\rm m}_\star=0$
       \item The dynamics of futures prices is linear, i.e. $F^\star_t=F_0(t)+F_1(t)\cdot X_t$ where $F_0,F_1$ are deterministic functions and $X_t=(S_t \, , \, Q_t)^\perp$.
   \end{itemize}
\end{Definition}

\hs

Under delegation, the Principal's problem of anticipating the price dynamics in the futures market $F_t$ transforms into:
\begin{align*}
 & V^{\rm P}:= \sup_{Z,\theta} J ^{\rm P}(Z,\theta), \quad \text{where} \quad J^{\rm P}(Z,\theta) := \E\Big[ U_{\rm P}\Big( g(X_T)  - Y_T + \int_0^T \theta_t dF_t \Big)\Big],  \quad g(x):= x \mcdot M x,\\
& dY_t = \frac12 Z_t \mcdot C_a  Z_tdt + Z_t \mcdot \Sigma d\bar W_t^\perp, \quad
dX_t = \gamma E_{22} Z_t dt + \Sigma d\bar W_t^\perp,
\end{align*}
where we recall that $C_i := \gamma E_{22}+\eta_i A, \quad  i\in \{ a,p\},$ $\bar{C} := \gamma E_{22} + (\eta_a+\eta_p) A$.
The incentive mechanism, $Y_T$, retains its functional form despite the dynamic hedging opportunity, as the latter remains solely under the Principal's control. Consequently, the Agent's best-response function remains unaffected at $\hat \alpha(z) = \gamma z_2$.
The market maker ensures the provision of hedging by solving the portfolio optimization problem specified in Equation \eqref{DynMMValue}. Accordingly, we seek an equilibrium in the futures market defined by the following definition.

\begin{Definition}[Second-best delegation futures market linear equilibrium]
A quadruplet $(Z^\star,\theta^{\rm d}_\star,\theta^{\rm m}_\star,F_t^\star)$ is a linear equilibrium if
   \begin{enumerate}[\upshape (i)]
       \item Given the dynamics of futures prices $F^\star_t$, the pair $(Z^\star,\theta^{\rm d}_\star)$ is optimal for the Principal,
       \item Given the dynamics of futures prices $F^\star_t$ and the pair $(Z^\star,\theta^{\rm d}_\star)$, $\theta^{\rm m}_\star$ is optimal for the market maker,
       \item The market clears at each date $t \le T$,i.e. $\theta^{\rm d}_\star+\theta^{\rm m}_\star=0$
       \item The dynamics of futures prices is linear, i.e. $F^\star_t=F_0(t)+F_1(t)\cdot X_t$ where $F_0,F_1$ are deterministic functions and $X_t=(S_t \, , \, Q_t)^\perp$.
   \end{enumerate}
\end{Definition}

\hs

\hs

The derivation of the system of ODEs that defines the linear market equilibrium is provided in the Online Appendix. The resolution strategy, which mirrors the methodology of the static hedging framework, is outlined below. We observe that we need to impose the two terminal conditions $F_0(T)=0$ and $F_1(T)=( 1, -\eps)^\tp$ to ensure the convergence of the forward price, $P_T$, to the spot price, $S_T - \eps Q_T$, at maturity $T$ across both scenarios: in-house and delegated production.

\begin{Proposition}\label{prop:Dyn-Eq}
Assume that the futures price process is of the affine form $F_t = F_0(t)+F_1(t)\cdot X_t,$ and $X_t=(S_t,Q_t)^\top,$
with terminal conditions $F_0(T)=0,$ $F_1(T)=(1,-\eps)^\tp.$ Then, in both the in-house and delegated cases:
\begin{enumerate}[\upshape (i)]
    \item the principal's certainty equivalent is quadratic in the state variable \(X\);
    \item the producer's optimal trading rate and the market maker's optimal trading rate are affine functions of \(X\);
    \item the linear market equilibrium condition reduces to a coupled system of ordinary differential equations for the coefficient functions \(F_0\) and \(F_1\).
\end{enumerate}
The explicit form of this ODE system, together with the derivation of the optimal coefficients, is provided in the Online Appendix.
\end{Proposition}

\hs

\begin{enumerate}

\item \textbf{Linear Price Assumption.} We postulate a linear structure for the futures price process, $F_t$, dependent on the state variable $X_t$:
$$F_t = F_0(t) + F_1(t) \cdot X_t$$

\item \textbf{Solving Agent and Market Maker Problems.}
The problems of the Principal and the Market Maker are solved via stochastic control. Both agents seek to maximize their respective certainty equivalents, $v$ and $v^m$.
\begin{itemize}
    \item HJB Equation and Quadratic Solution. The Principal's certainty equivalent, $v(t,x)$, is shown to satisfy a Hamilton-Jacobi-Bellman (HJB) equation of the generic quadratic form
    \begin{align}\label{gen-pde}
&0=\partial_t v +  \big( B_0 + B_1 x\big) \mcdot Dv + \frac12 A\mddot D^2v  + \frac12  Dv \mcdot \Lambda Dv  + H_0 + H_1 \mcdot x +  x \mcdot H_2 x, \\
& v(T,x) = g_0 + g_1 \mcdot x + x \mcdot g_2 x. 
\end{align}
    Due to this structure, the solution is quadratic, $v(t,x) = K_0 + K_1 \cdot x + x \cdot K_2 x$. The matrix coefficient $K_2$ satisfies a Riccati equation, while $K_1$ and $K_0$ are governed by linear ODEs.
    \item Linear Best-Response. Crucially, the optimal trading rates for both the Producer, $\theta^{\rm i}$, and the Market Maker, $\theta^{\rm m}$, are found to be linear functions of the state variable $X$, with their coefficients depending explicitly on the assumed price functions $(F_0, F_1)$.
\end{itemize}

\item \textbf{Equilibrium Closure.}
The equilibrium is defined by the market clearing condition, which mandates that the aggregate optimal trading rate is zero at all times $t$ and states $x$:
$$\theta^{\rm i}(t,x) + \theta^{\rm m}(t,x) = 0$$
Since the optimal trading rates are linear in $X$, this condition provides a system of ODEs that uniquely characterizes the equilibrium futures price coefficients, the pair $(F_0, F_1)$.

\item \textbf{System Complexity and Existence.}
We emphasize that the final system of ODEs for the pair $(F_0, F_1)$ is non-trivial (non-linear), as its coefficients are intricately linked to the solutions of the Riccati systems, $K_j^{\rm i}$ and $K_j^{\rm m}$, which themselves depend on $(F_0, F_1)$. We do not formally address the delicate question of existence and uniqueness of this complex system. Instead, we provide numerical approximation results which suggest the system is well-posed, deferring a rigorous exploration of this technical issue to future research.

\end{enumerate}

\hs

\begin{figure}[!thb]
\begin{center}
\begin{tabular}{ccc}
(a) & (b) & (c) \\
\includegraphics[width=0.3\textwidth]{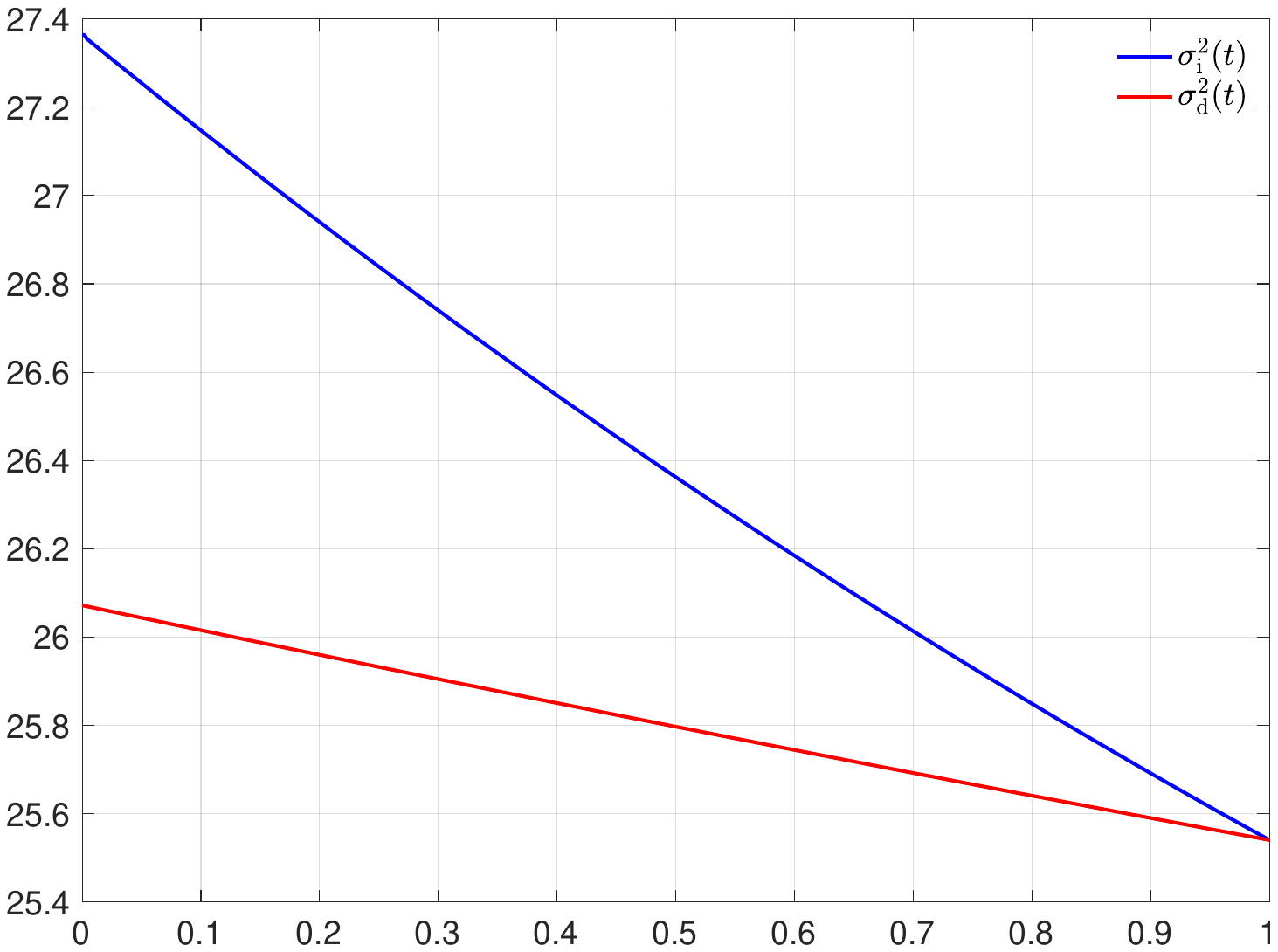} &
\includegraphics[width=0.3\textwidth]{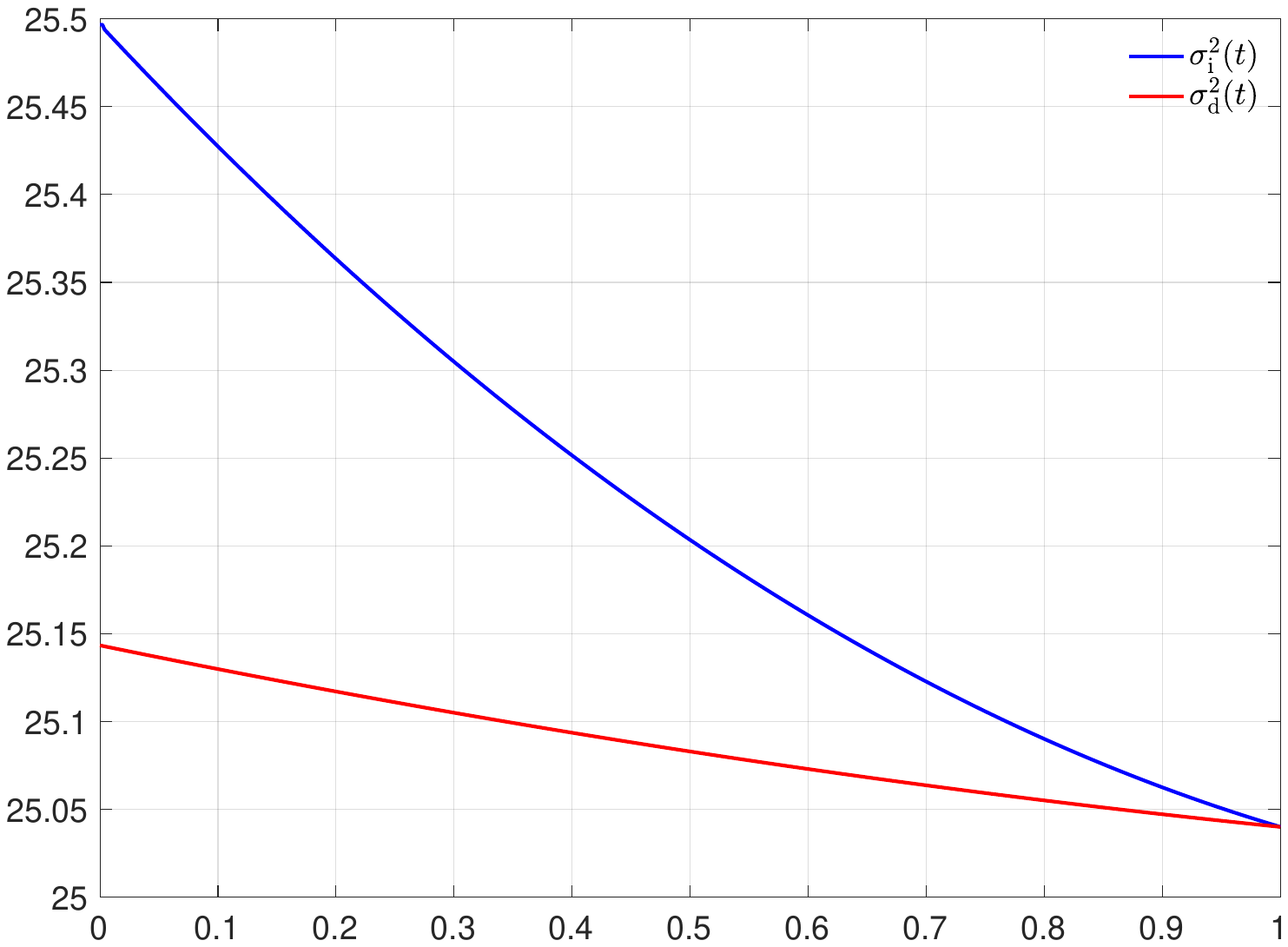} &
\includegraphics[width=0.3\textwidth]{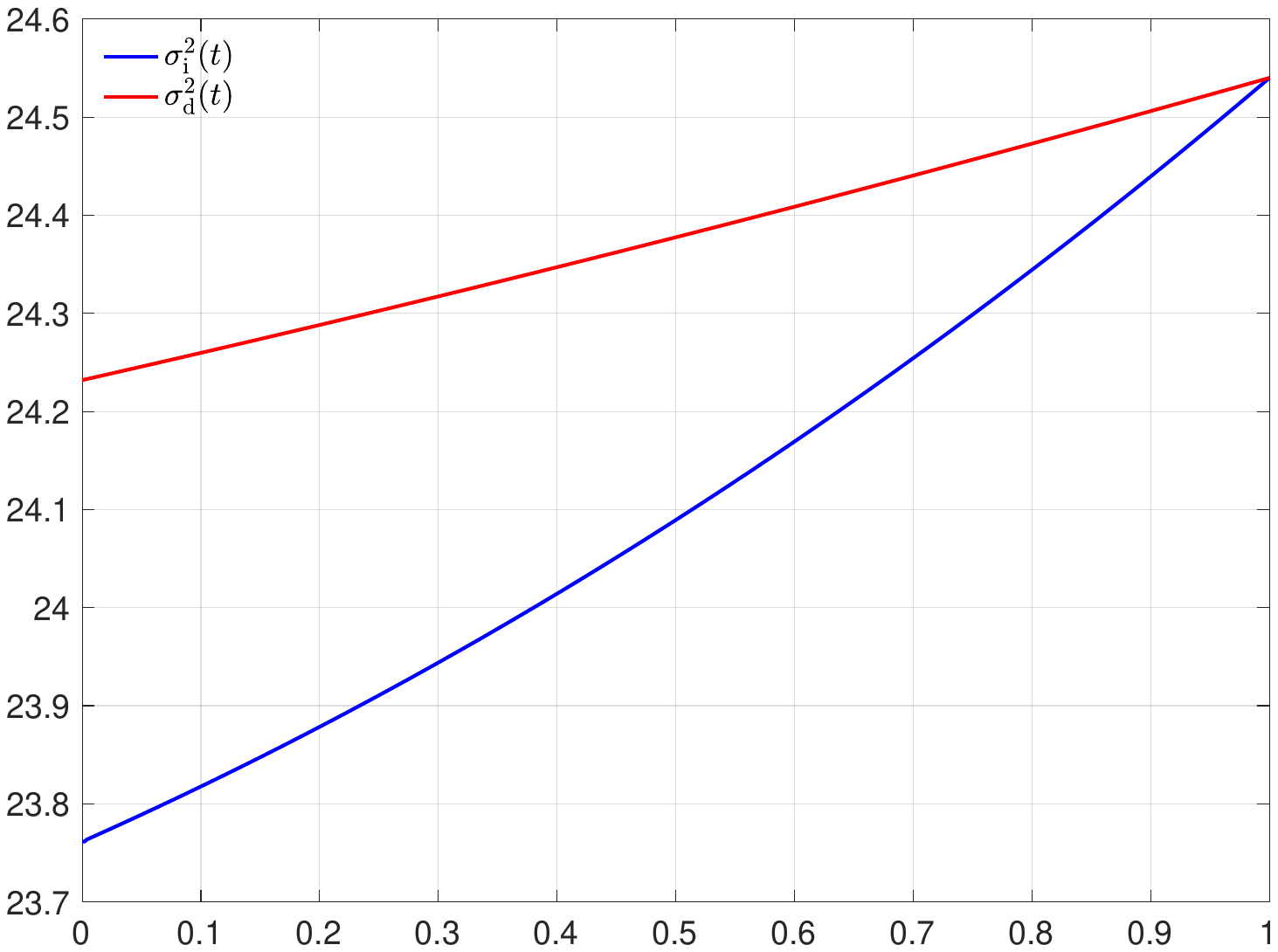} \\
\includegraphics[width=0.3\textwidth]{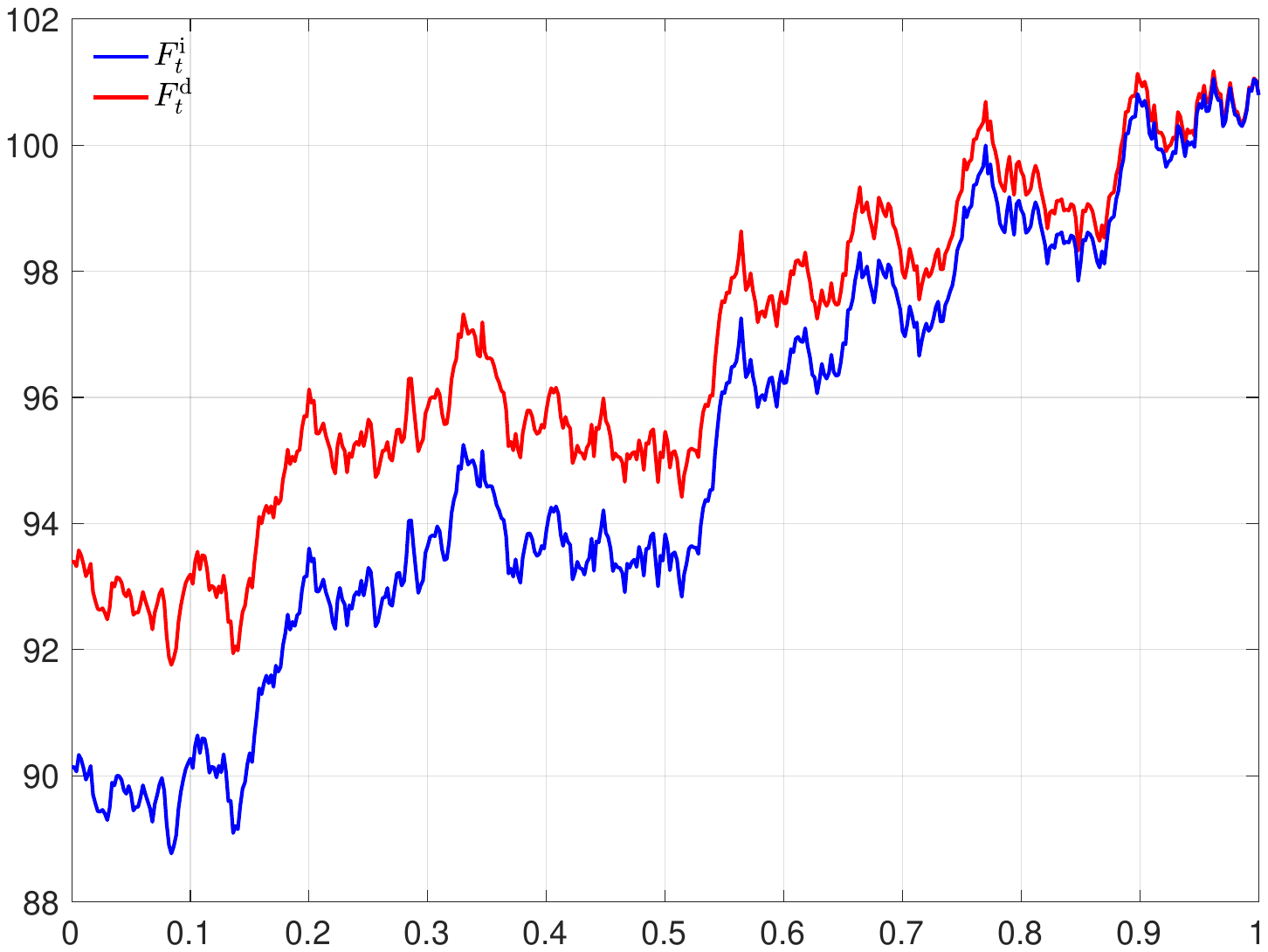} &
\includegraphics[width=0.3\textwidth]{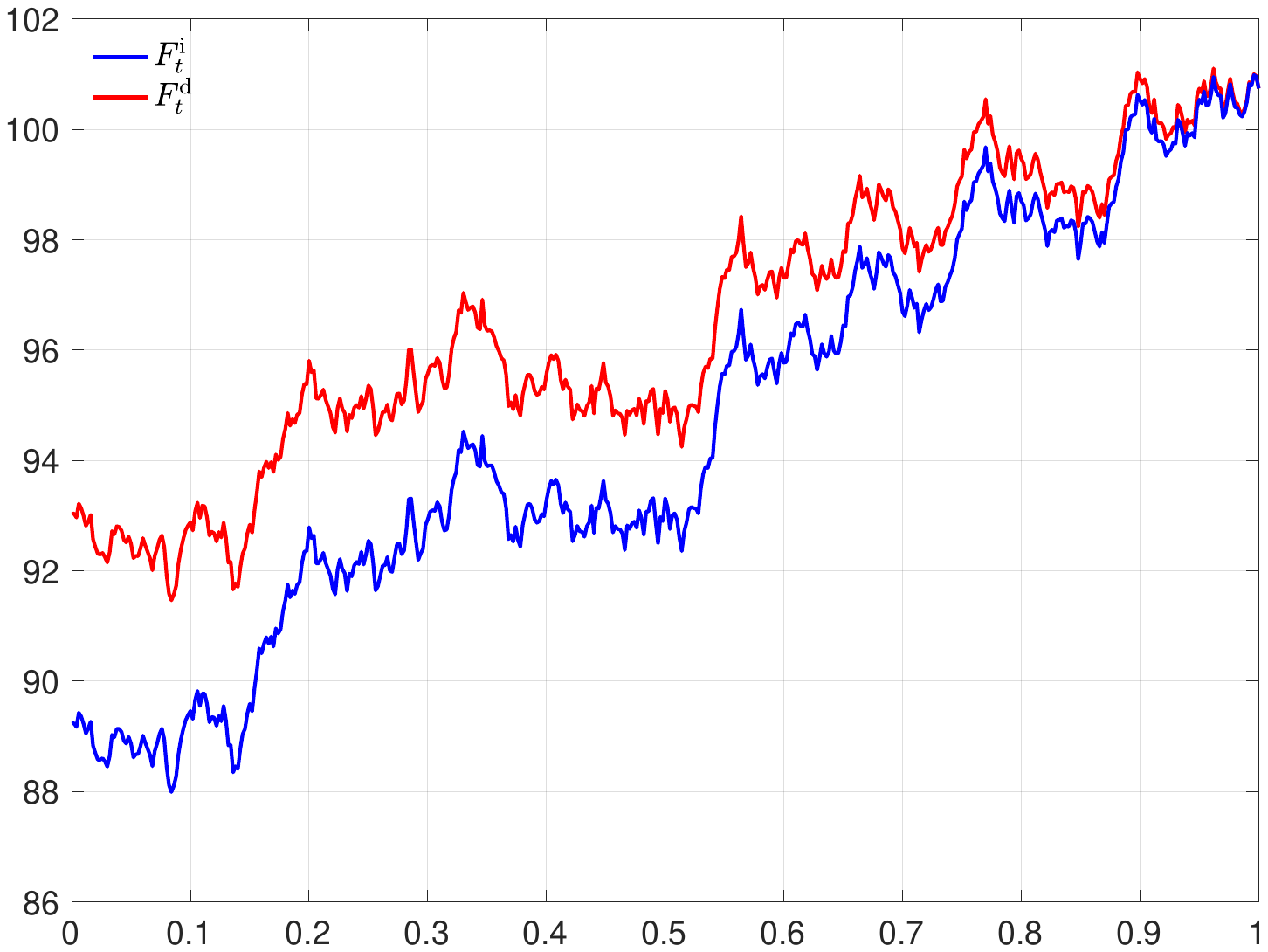} & 
\includegraphics[width=0.3\textwidth]{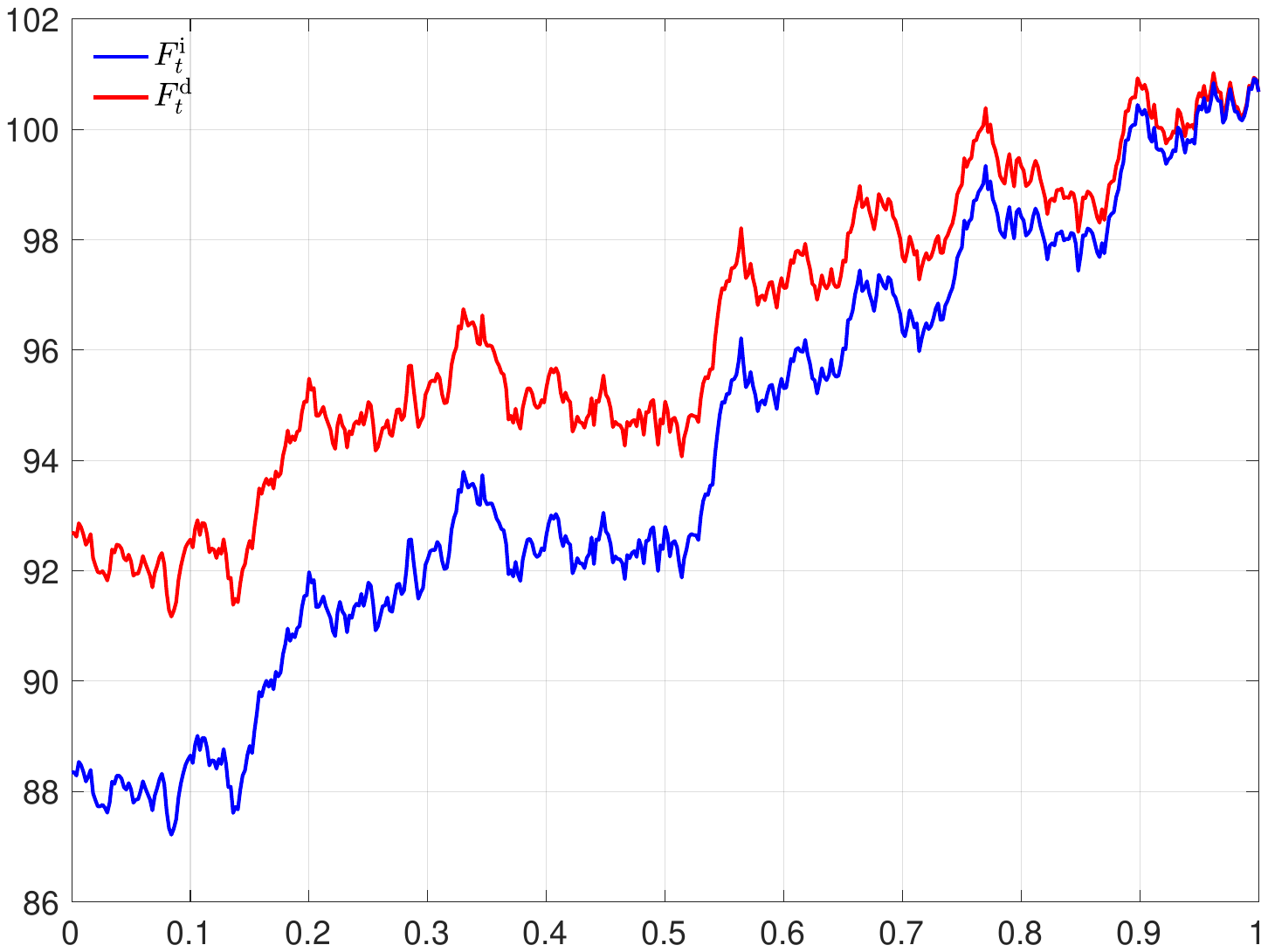} \\ 
\includegraphics[width=0.3\textwidth]{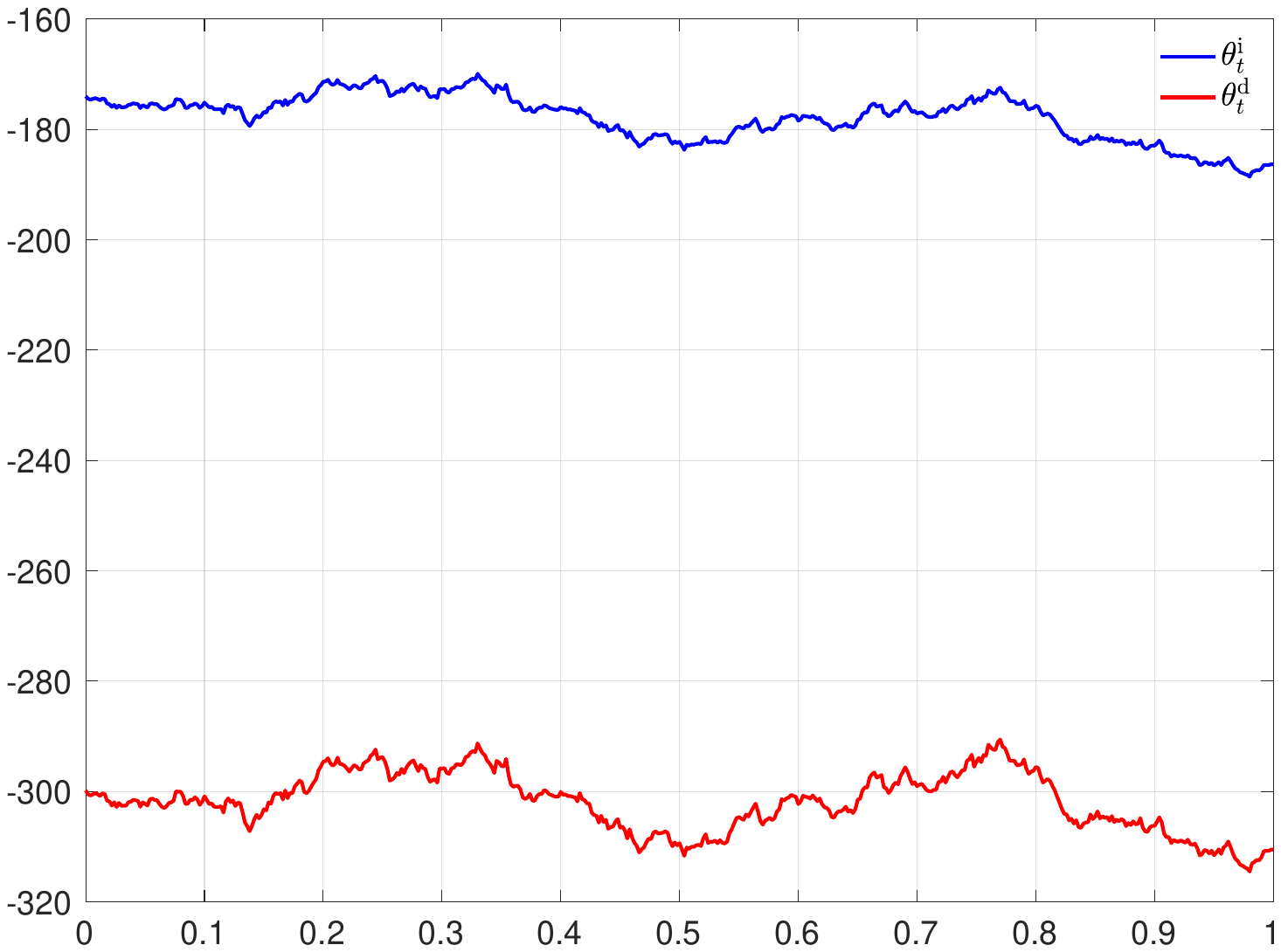} &
\includegraphics[width=0.3\textwidth]{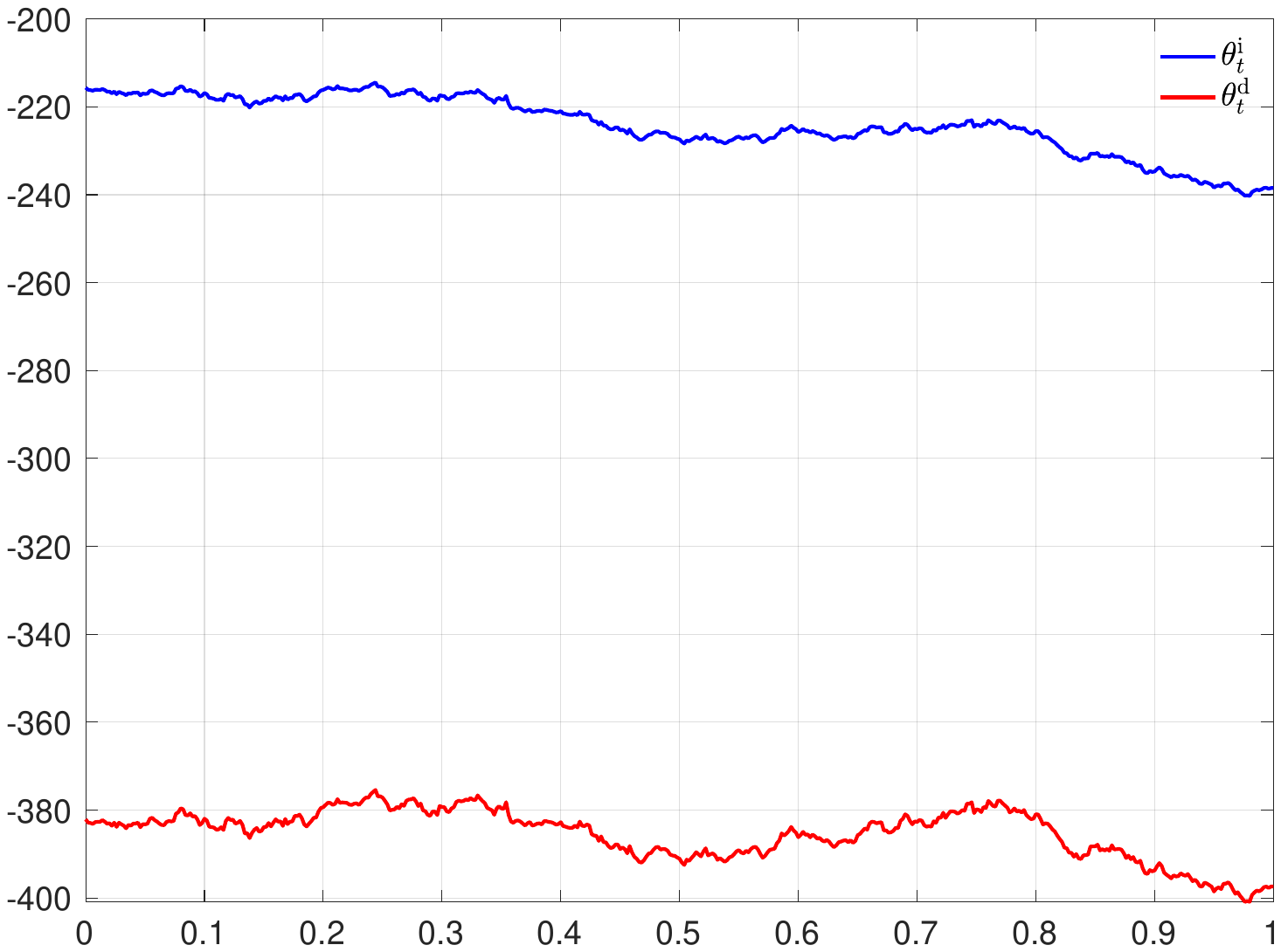} & 
\includegraphics[width=0.3\textwidth]{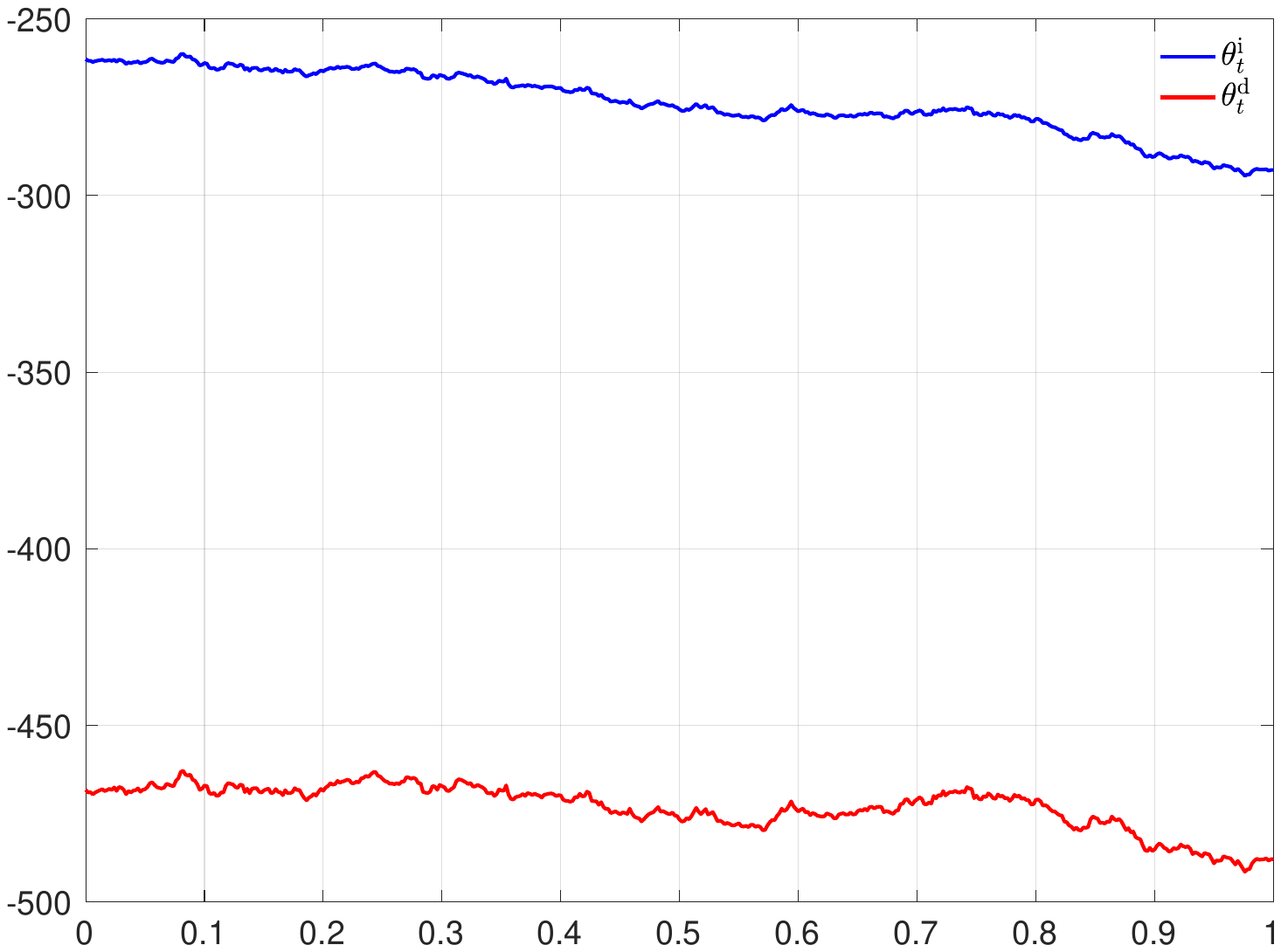} \\ 
 \end{tabular}
\end{center}
\caption{
{\small Volatilities $\sigma^2_{\rm i}(t)$  and $\sigma^2_{\rm d}(t)$ of the futures prices $F^{\rm i}_t$ and $F^{\rm d}_t$ and the respective hedging positions $\theta^{\rm i}_t$ and $\theta^{\rm d}_t$. Column (a) $\rho=-0.25$. Column (b) $\rho=0$. Column (c) $\rho=0.25$. Other parameters value: $s_0=100,\; \nu=5.0, \; \epsilon=1.00\,10^{-2},\; \sigma=20,\; \eta_p=1.00\,10^{-3},\; \eta_a=1.00\,10^{-3},\; \eta_m=1.00\,10^{-3},\; T=1,\; \gamma=4.00\,10^{-3}$. }
}
\label{fig:vol}
\end{figure}

\subsection{Futures price volatility}

Figure~\ref{fig:vol} depicts sample trajectories of the equilibrium futures prices, $F^{\rm i}_t$ and $F^{\rm d}_t$, alongside their corresponding squared volatilities, defined as $\sigma_{\rm i}^2(t) := F^{\rm i}_1(t) \cdot A F^{\rm i}_1(t)$ and $\sigma_{\rm d}^2(t) := F^{\rm d}_1(t) \cdot A F^{\rm d}_1(t)$. These dynamics, along with the associated hedging positions, are presented for three distinct values of the correlation coefficient $\rho$. \\

First, we observe that the futures price trajectories exhibit a positive drift in both the in-house and delegated scenarios. This behavior is consistent with the firm's net short position (selling forward). Furthermore, corroborating our findings from the static framework, the results indicate that delegation induces a structural increase in futures prices relative to the integrated case.\\

Second, the numerical analysis shows that the term structure of futures-price volatility is not uniform and varies with the correlation between demand and supply shocks. In our parameterization, more positive values of $\rho$ are associated with volatility profiles that are more consistent with the Samuelson effect, that is, with higher futures-price volatility as maturity approaches. Conversely, for lower values of $\rho$, the volatility term structure may flatten or even decrease with maturity.

These patterns are consistent with the state-variable view developed by \cite{Anderson83}   and \cite{Anderson85}, according to which the shape of the volatility curve depends on how uncertainty about the underlying supply-demand equilibrium is resolved over time. In this sense, our dynamic extension shows that heterogeneous volatility patterns can arise endogenously from the joint evolution of the state variables, even in the absence of storage.

Taken together, these numerical results indicate that the Samuelson effect should not be viewed as a universal prediction of the model, but rather as a correlation-dependent outcome within the linear-equilibrium framework studied in this section.

\subsection{Preliminary result on Riccati's system}

Consider the following PDE 
\begin{align*}  
&0=\partial_t v +  \big( B_0 + B_1 x\big) \mcdot Dv + \frac12 A\mddot D^2v  + \frac12  Dv \mcdot \Lambda Dv  + H_0 + H_1 \mcdot x +  x \mcdot H_2 x, \\
& v(T,x) = g_0 + g_1 \mcdot x + x \mcdot g_2 x. 
\end{align*}
with $B_0, B_1, H_0, H_1$ and $H_2$ functions of time and $A$ constant. We provide a solution of the PDE of the form 
$$v = K_0 + K_1 \mcdot x +  x \mcdot K_2 x$$
 where $K_2$ solves the matrix Riccati equation 
 \begin{align}\label{eq:K2}
0 = & \dot K_2 +  B_1^\tp K_2 +  K_2 B_1 +2 K_2 \Lambda K_2 + H_2, \quad K_2(T) = g_2, 
\end{align}
the function $K_1$ solves the linear ODE:
\begin{align*}  
0 = & \dot K_1 + 2 K_2  B_0 + B_1^\tp K_1 + 2 K_2 \Lambda K_1 + H_1, \quad K_1(T) = g_1, \\
\end{align*}
and the function $K_0$ is given by
\begin{align*}  
K_0(t) & = g_0 + \int_t^T  \big[B_0  \mcdot K_1 +  A\mddot K_2 + \frac12 K_1 \mcdot \Lambda K_1 + H_0 \big](s) ds.
\end{align*}

We first rewrite the Riccati equation in terms of  $\Gamma := K_2 - g_2$:
\begin{align*}  
0=& \dot \Gamma + 2 \Gamma \Lambda \Gamma +  (B_1^\tp + 2 g_2 \Lambda) \Gamma +  \Gamma( B_1 +2  \Lambda g_2) - L , \quad \Gamma(T) = 0, 
\quad L := - H_2 - B_1^\tp g_2 - g_2 B_1- 2 g_2 \Lambda g_2.
\end{align*}

Our existence results relies on the following Lemma. Its proof can be found in \cite{Farhat21} (section 3.7.1), which is a direct application of \cite{Reid72} (chap. 2). \\

\noindent{\bf Lemma}: Let $\Lambda$, $L$ and $A$ be continuous and bounded, $\Lambda, L \geq 0$ and $L_T>0$. Then the matrix Riccati ODE~\eqref{eq:K2}  has a unique solution on $(0,T)$.

\subsection{Integrated producer}
For a given forward price process $F_t$, the problem of the producer is
\begin{align*}
 & V^{\rm P} := \sup_{\alpha,\theta} J ^{\rm P}(\alpha,\theta), \quad \text{where} \quad J^{\rm P}(\alpha,\theta) := \E\Big[ U_{\rm P}\Big( g(X_T)  - \int_0^T c(\alpha_t) dt + \int_0^T \theta_t dF_t \Big)\Big],  \\
& dX_t = \big( b + \alpha_t e_{2}\big) dt + \Sigma d\bar W_t^\perp, \quad X_t := ( S_t \,\, Q_t)^\tp, \quad g(x) := x\mcdot M x, \\
& \Sigma := \begin{pmatrix} \nu\rho & \nu\sqrt{1-\rho^2} \\ \sigma & 0 \end{pmatrix} \quad 
M :=  \begin{pmatrix} 0 &\frac12 \\ \frac12 & -\eps \end{pmatrix}, \quad
b := ( \mu \,,\, 0)^\tp, \quad e_2 := ( 0 \,,\, 1)^\tp.
\end{align*}

The optimisation problem of the market maker who ensures the supply of hedging is now
\begin{align*}
 & V^{\rm M}:= \sup_{\theta} J^{\rm M}(\theta), \quad \text{where} \quad
 J^{\rm M}(\theta):= \E\Big[ U_{\rm M}\Big( \int_0^T \theta_t dF_t \Big)\Big].
 \end{align*}
The main objective of this section is to find an equilibrium in the following sense.

\begin{Definition}
A linear market equilibrium is defined as a quadruplet $(\alpha^P,\theta^P,\theta^M, F^\ast)$ such that
\b*
&V^{\rm P} = J^{\rm P}(\alpha^P,\theta^P), \quad V^{\rm M} = J^{\rm M}(\theta^M), \quad \theta^P+\theta^M=0, 
&
\\
&
F^\ast_t = F_0(t) + F_1(t)\mcdot X_t, \quad t\in[0,T]
&
\e*
for some deterministic functions $F_0,F_1$ satisfying $F_0(T) = 0$ and $F_1(T) = (1, \, -\epsilon)^\tp$.
\end{Definition}

The final condition on the functions $F_0$ and $F_1$ guarantees that the forward price at maturity $T$ coincides with the spot $P_T := S_T - \eps Q_T$. We also report the dynamics of $F_t$:
\begin{align*}
dF_t =\big(\bar \beta_t + \alpha_t F_1 \mcdot e_{2}\big) dt + F_1 \mcdot \Sigma d\bar W_t^\perp,
\quad
\text{where} \quad \bar \beta = \dot F_0 + F_1\mcdot b + \dot F_1 \mcdot X.
\end{align*}

\hs

\noindent{\bf Producer}: For a given forward price process $F_t := F_0(t) + F_1(t) \mcdot X_t$, the problem of the producer is 
\begin{align*}
 & V(0,x,\ell) := \sup_{\alpha,\theta} J(\alpha,\theta), \quad \text{where} \quad J(\alpha,\theta) := \E\Big[ U_{\rm P}\Big(g(X_T) - L_T\Big)\Big], \\
&dL_t := c(\alpha_t) dt - \theta_t dF_t, \quad 
dF_t = \big(\bar\beta_t + \alpha_t F_1(t)\mcdot e_2\big) dt + F_1(t) \mcdot \Sigma d\bar W_t^\perp, \\
&dL_t = \Big( c(\alpha_t) - \theta_t \big[\bar\beta_t + \alpha_t F_1(t) \mcdot e_2\big]\Big)dt -  \theta_t F_1(t) \mcdot \Sigma d\bar W_t^\perp, \quad
dX_t = (b + \alpha_t e_2) dt + \Sigma d\bar W_t^\perp.
\end{align*} 

With $A = \Sigma\Sigma^\tp$, the PDE satisfied by $V$ is
\begin{align*}
& 0 =\partial_t V + \frac12 A\mddot D^2_x V + b\mcdot D_xV \\
 & +  \sup_{a,\theta}\Big\{ 
a e_2 \mcdot D_xV + \big(c(a) - \theta(\bar\beta + a F_1\mcdot e_2) \big) V_\ell
+\frac12 \theta^2 F_1\mcdot A F_1 V_{\ell\ell}
- \theta AF_1 \mcdot D_xV_\ell
 \Big\}\\
&V(T,x,\ell) = U_{\rm P}(g(x)-\ell).
\end{align*}

We set $V(t,x,\ell) =:U_{\rm P}(v(t,x)-\ell)$, so that
\begin{align*}
D_xV  = -\eta_p V Dv, \quad D_x^2V = -\eta_pV(D^2v - \eta_p DvDv^\tp), \quad V_\ell = \eta_p V, \quad
D_xV_\ell = -\eta_p^2 V Dv, \quad V_{\ell\ell} = \eta_p^2 V, 
\end{align*}
and get the PDE satisfied by $v$
\begin{align*}
&0=\partial_t v + \frac12 A\mddot (D^2 v -\eta_p DvDv^\tp) + b\mcdot Dv 
 + \sup_{a,\theta} L(a,\theta), \quad
v(T,x) =  g(x).
\end{align*}
with 
\begin{align*}
L(a,\theta) := 
a e_2 \mcdot Dv - \big(c(a) - \theta(\bar\beta + a F_1\mcdot e_2)\big)
-\frac12 \eta_p \theta^2 F_1\mcdot A F_1
 - \theta \eta_p AF_1 \mcdot Dv.
\end{align*}
It comes that
\begin{align*}
L(a,\theta) := 
a e_2 \mcdot Dv -  c(a) + \theta\big[ \bar\beta + a F_1\mcdot e_2 - \eta_p AF_1\mcdot Dv\big]
-\frac12 \eta_p \theta^2 F_1\mcdot A F_1
\end{align*}

If $\phi_1 := F_1\mcdot A F_1>0$, $L$ admits a maximiser in $\theta^{\rm i}$ 
\begin{align*}
L(a,\theta^{\rm i}) &= 
a e_2 \mcdot Dv -  c(a) 
+ \frac12 \eta_p  \phi_1(\theta^{\rm i})^2, \quad
\theta^{\rm i} := \frac{1}{\eta_p\phi_1} \big(\bar\beta  + a  F_1 \mcdot e_2 -   \eta_p AF_1\mcdot Dv\big),
\end{align*}
from which we get
\begin{align*}
L(a,\theta^{\rm i}) &= -\frac{1}{2\gamma\eta_p \phi_1}\Big( \eta_p \phi_1 -  \gamma (F_1\mcdot e_2)^2\Big) a^2
+  \frac1{\eta_p\phi_1} \Big(  (F_1\mcdot e_2)  \bar\beta   + \eta_p [ \phi_1 e_2 - (e_2 \mcdot F_1) AF_1] \mcdot Dv  \Big) a  
+  \frac{1}{2\eta_p\phi_1}\Big(\bar\beta - \eta_p AF_1\mcdot Dv \Big)^2
\end{align*}

If $\gamma_{\rm i} := \eta_p\phi_1 -  \gamma (F_1\mcdot e_2)^2>0 $, there is a maximiser $\hat a$ given by
\begin{align*}
\hat a := \frac{\gamma}{\gamma_{\rm i}} \big(  (F_1\mcdot e_2)  \bar\beta   + \eta_p u_2 \mcdot Dv  \big), \quad
u_2 :=  \phi_1 e_2 - (e_2 \mcdot F_1) AF_1
\end{align*}
so that
\begin{align*}
L(\hat a,\theta^{\rm i}) &= 
 \frac{1}{2\eta_p \phi_1} \Big[ \frac{\gamma}{\gamma_{\rm i}} \big[(F_1\mcdot e_2)  \bar\beta   + \eta_p u_2 \mcdot Dv \big]^2
+ \big(\bar\beta - \eta_p AF_1\mcdot Dv \big)^2
\Big] \\
& =  \frac{1}{2\eta_p \phi_1}\Big[
\big(1+ \frac{\gamma}{\gamma_{\rm i}}(F_1\mcdot e_2)^2 \big)\bar\beta^2
 + 2\eta_p \bar\beta \big(  \frac{\gamma}{\gamma_{\rm i}}(F_1\mcdot e_2) u_2 -   AF_1 \big) \mcdot Dv 
 + Dv \mcdot \big(\eta_p^2 AF_1F_1^\tp A + \frac{\eta_p^2\gamma}{\gamma_{\rm i}} u_2u_2^\tp \big) Dv
\Big]\end{align*}

and the PDE reduces to
\begin{align*}
&0 = \partial_t v + \frac12 A\mddot (D^2 v -\eta_p DvDv^\tp) + b\mcdot Dv  + L(\hat a,\theta^{\rm i}), \quad
v(T,x) =  g(x).
\end{align*}

Recalling that $\bar \beta = \dot F_0 + F_1\mcdot b + \dot F_1 \mcdot X$, this PDE can be written
\begin{align*}
&0=\partial_t v +  \big( B_0^{\rm i} + B_1^{\rm i} x\big) \mcdot Dv 
+ \frac12 A\mddot D^2v  + \frac12  Dv \mcdot \Lambda_{\rm i} Dv  
+ H_0^{\rm i} + H_1^{\rm i} \mcdot x +  x \mcdot H_2^{\rm i} x, \\
& v(T,x) = g_0 + g_1 \mcdot x + x \mcdot g_2 x,
\end{align*}
with
\begin{align*}
&B_0^{\rm i} := b  + \frac{\dot F_0 + F_1\mcdot b}{ \phi_1}\big(\frac{\gamma}{\gamma_{\rm i}}(F_1\mcdot e_2) u_2 -  AF_1 \big), \quad
 B_1^{\rm i} :=  \frac{1}{\phi_1} \big(  \frac{\gamma}{\gamma_{\rm i}}(F_1\mcdot e_2) u_2 -  AF_1 \big)\dot F_1^\tp,   \\
&H_0^{\rm i} :=  \frac{1}{2\eta_p\phi_1}\big(1 + \frac{\gamma}{\gamma_{\rm i}}(F_1\mcdot e_2)^2)(\dot F_0 + F_1 \mcdot b)^2, \quad
H_1^{\rm i} := \frac{1}{\eta_p\phi_1}\big(1 + \frac{\gamma}{\gamma_{\rm i}}(F_1\mcdot e_2)^2)(\dot F_0 + F_1 \mcdot b)\dot F_1, \\ 
& H_2^{\rm i} := \frac{1}{2\eta_p\phi_1}\big(1 + \frac{\gamma}{\gamma_{\rm i}}(F_1\mcdot e_2)^2) \dot F_1 \dot F_1^\tp, \quad
\Lambda_{\rm i} := -\eta_p A  + \frac{\eta_p}{\phi_1}AF_1F_1^\tp A + \frac{\eta_p\gamma}{\gamma_{\rm i}\phi_1} u_2 u_2^\tp.
\end{align*}

\hs

\noindent{\bf Representative agent of the market}: The problem is
\begin{align*}
 & V^{\rm m}(0,x,\ell):= \sup_{\theta}  \E\big[ U_{\rm M}\big( L_T \big)\big], \quad
dX_t = \big( b + \hat a(t,X_t) e_2 \big) dt + \Sigma d\bar W_t^\perp, \\
& dF_t = \big( \bar \beta_t + \hat a(t,X_t) F_1(t) \mcdot e_2 \big) dt +  F_1(t) \mcdot \Sigma d\bar W_t^\perp, \\
& dL_t = \theta_t dF_t = \theta_t\big[ \bar\beta_t + \hat a(t,X_t) F_1 \mcdot e_2\big] dt + \theta_t F_1 \mcdot \Sigma dW_t^\perp.
 \end{align*}
The HJB for $V^m$ writes
\begin{align*}
& 0=\partial V^m_t + \frac12 A\mddot D^2_xV^m 
+ \sup_{\theta} \Big\{ 
 \big( b + \hat a e_2 \big) \mcdot D_xV^m 
+ \big( \bar \beta + \hat a  F_1 \mcdot e_2 \big)\theta V_\ell
+ \frac12 \theta^2 F_1 \mcdot AF_1 V_{\ell\ell}
+ \theta AF_1 \mcdot D_xV_\ell
\Big\}, \\
&V^m(T,x,\ell) =U_{\rm M}(\ell).
 \end{align*}
We define  $V^m(t,x,\ell) =:U_{\rm M}(v^m(t,x) + \ell)$, so that
\begin{align*}
&D_xV^m  = -\eta_m V^m Dv^m, \quad D_x^2V^m  = -\eta_m V^m(D^2v^m - \eta_p Dv^m\mcdot Dv^m), \\
& V^m_\ell = -\eta_m V^m,  \quad
D_xV^m_\ell = \eta_m^2 V^m Dv^m, \quad V^m_{\ell\ell} = \eta_m^2 V^m, 
\end{align*}
The certainty equivalent $v^m$ satisfies:
\begin{align*}
0&=\partial v^m_t + \frac12 A\mddot \big(D^2 v^m -\eta_m Dv^m (Dv^m)^\tp\big) + \big( b + \hat a e_2 \big) \mcdot Dv^m + \sup_{\theta} G(\theta), \quad
v^m(T,x) = 0, \\
\text{where}\quad 
G(\theta) & :=   \big( \bar \beta + \hat a  F_1 \mcdot e_2 - \eta_m  AF_1 \mcdot D v^m  \big)\theta 
- \frac12 \eta_m  \phi_1  \theta^2
 \end{align*}
It holds that $G$ admits a maximiser $\theta_{\rm m,i}$ also if $\phi_1 = F_1\mcdot AF_1 >0$. Then we have
\begin{align*}
G(\theta) & :=  - \frac12 \eta_m \phi_1\big( \theta - \theta_{\rm m,i}\big)^2
+  \frac12 \eta_m \phi_1 \theta_{\rm m,i}^2, \quad
\theta_{\rm m,i} := \frac1{\eta_m\phi_1} \big(\bar \beta + \hat a  F_1 \mcdot e_2 -  \eta_m AF_1 \mcdot D v^m \big)
 \end{align*}
Hence, the PDE satisfied by $v^m$ becomes
\begin{align*}
0&=\partial v^m_t + \frac12 A\mddot \big(D^2 v^m -\eta_m Dv^m (Dv^m)^\tp\big) + \big( b + \hat a e_2 \big) \mcdot Dv^m 
+   \frac1{2\eta_m \phi_1} \big(\bar \beta + \hat a  F_1 \mcdot e_2 -  \eta_m AF_1 \mcdot D v^m\big)^2
 \\
v^m(T,x) &= 0.
 \end{align*}

Recalling that 
\begin{align*}
\hat a & := \frac{\gamma}{\gamma_{\rm i}} \big(  (F_1\mcdot e_2)  \bar\beta   + \eta_p u_2 \mcdot Dv  \big) \quad \text{and} \quad Dv = K_1^{\rm i} + 2 K_2^{\rm i} x, \\
\hat a & := \frac{\gamma}{\gamma_{\rm i}} \big\{  
(F_1\mcdot e_2)(\dot F_0 + F_1 \mcdot b) 
+ \eta_p u_2 \mcdot K_1^{\rm i} 
+ (F_1\mcdot e_2) \dot F_1 \mcdot x 
+ 2 \eta_p u_2 \mcdot K_2^{\rm i} x \big\} 
 = \hat a_0 + \hat a_1 \mcdot x \\
&\text{with} \quad
\hat a_0 := \frac{\gamma}{\gamma_{\rm i}} \big\{ (F_1\mcdot e_2)(\dot F_0 + F_1 \mcdot b) 
+ \eta_p u_2 \mcdot K_1^{\rm i} \big\}, \quad
\hat a_1 := \frac{\gamma}{\gamma_{\rm i}} 
\big[(F_1\mcdot e_2) \dot F_1 
+ 2 \eta_p (K_2^{\rm i})^\tp u_2 \big].
\end{align*}

it holds that
\begin{align*}
\bar \beta + \hat a  F_1 \mcdot e_2 
& = \dot F_0 + F_1 \mcdot b + (F_1\mcdot e_2) \hat a_0 + (\dot F_1 + (F_1\mcdot e_2) \hat a_1) \mcdot x 
 =: \beta_0 + \beta_1 \mcdot x
\end{align*}

this PDE can be written
\begin{align*}
&0=\partial_t v^m +  \big( B_0^{\rm m,i} + B_1^{\rm m,i} x\big) \mcdot Dv^m + \frac12 A\mddot D^2v^m  + \frac12  Dv^m \mcdot \Lambda_{\rm m,i} Dv^m  + H_0^{\rm m,i} + H_1^{\rm m,i} \mcdot x +  x \mcdot H_2^{\rm m,i} x, \\
& v(T,x) = 0,
\end{align*}
with
\begin{align*}
&B_0^{\rm m,i} :=  b + \hat a_0 e_2 - \frac{\beta_0}{\phi_1} AF_1, \quad
 B_1^{\rm m,i} :=     e_2 \hat a_1^\tp - \frac1\phi_1 AF_1 \beta_1^\tp, \quad
H_0^{\rm m,i} :=  \frac{\beta_0^2}{2\eta_p\phi_1}, \quad
 H_1^{\rm m,i} := \frac{\beta_0}{\eta_m\phi_1} \beta_1, \\ 
& H_2^{\rm m,i} := \frac{1}{2\eta_m\phi_1} \beta_1\beta_1^\tp, \quad
\Lambda_{\rm m,i} := -\eta_m A  + \frac{\eta_m}{\phi_1} AF_1 F_1^\tp A.
\end{align*}

\noindent{\bf Equilibrium} The clearing condition $\theta^m + \theta^{\rm i} =0$ writes
\begin{align*}
\frac1{\eta_m} \big(\bar \beta + \hat a  F_1 \mcdot e_2 -  \eta_m AF_1 \mcdot D v^m  \big) 
+
 \frac{1}{\eta_p} \big(\bar\beta  + \hat a  F_1 \mcdot e_2 -   \eta_p AF_1\mcdot Dv\big)
= 0.
\end{align*}
Thus
\begin{align*}
 -  \bar\eta AF_1 \mcdot (D v^m  +  Dv)+  \bar \beta + \hat a  F_1 \mcdot e_2    = 0, \quad \bar\eta := \frac{\eta_p\eta_m}{\eta_p + \eta_m}.
\end{align*}
Now, using the ansatz that $Dv = 2 x\mcdot K_2^{\rm i} + K_1^{\rm i}$ and $Dv^m = 2 x\mcdot K_2^{\rm m,i} + K_1^{\rm m,i}$ we get
\begin{align*}
  \big[ - 2  \bar\eta (K_2^{\rm m,i} + K_2^{\rm i})  AF_1 + \beta_1 \big] \mcdot x 
  +   \beta_0 - \bar\eta AF_1 \mcdot (K_1^{\rm m,i} + K_1^{\rm i})   = 0,
\end{align*}
which yields
\begin{align*}
    - 2 \bar\eta (K_2^{\rm m,i} + K_2^{\rm i})  AF_1 + \beta_1 & = 0, \quad
    \beta_0 - \bar\eta AF_1 \mcdot (K_1^{\rm m,i} + K_1^{\rm i})     = 0,
\end{align*}
or more explicitly
\begin{align*}
     \dot F_1 + (F_1\mcdot e_2) \hat a_1 - 2 \bar\eta  (K_2^{\rm m,i} + K_2^{\rm i}) AF_1  & =0, \\ 
    \dot F_0 + F_1 \mcdot b + (F_1\mcdot e_2) \hat a_0 - \bar\eta AF_1 \mcdot (K_1^{\rm m,i} + K_1^{\rm i})   & = 0,
\end{align*}
or
\begin{align*}
 \dot F_1 + \frac{\gamma}{\gamma_{\rm i}} (F_1\mcdot e_2) \Big\{  (F_1\mcdot e_2) \dot F_1 
+ 2 \eta_p (K_2^{\rm i})^\tp u_2 \Big\}  -  2  \bar\eta (K_2^{\rm m,i} + K_2^{\rm i})  AF_1 & =0, \\ 
    \dot F_0 + F_1 \mcdot b + \frac{\gamma}{\gamma_{\rm i}} (F_1\mcdot e_2) \Big\{  (F_1\mcdot e_2)(\dot F_0 + F_1 \mcdot b) 
+ \eta_p u_2 \mcdot K_1^{\rm i}  \Big\} - \bar\eta AF_1 \mcdot (K_1^{\rm m,i} + K_1^{\rm i})   & = 0,
\end{align*}
and finally
\begin{align*}
 \Big[1 + \frac{\gamma}{\gamma_{\rm i}} (F_1\mcdot e_2)^2\Big]  \dot F_1  + 2 \eta_p \frac{\gamma}{\gamma_{\rm i}} (F_1\mcdot e_2) (K_2^{\rm i})^\tp u_2  - 2  \bar\eta (K_2^{\rm m,i} + K_2^{\rm i})  AF_1& =0, \\ 
\Big[1 + \frac{\gamma}{\gamma_{\rm i}} (F_1\mcdot e_2)^2\Big]   \dot F_0 + \Big[1 + \frac{\gamma}{\gamma_{\rm i}} (F_1\mcdot e_2)^2\Big] F_1 \mcdot b + \frac{\eta_p\gamma}{\gamma_{\rm i}} (F_1\mcdot e_2)   u_2 \mcdot K_1^{\rm i}  - \bar\eta AF_1 \mcdot (K_1^{\rm m,i} + K_1^{\rm i})   & = 0.
\end{align*}


\subsection{Second-best delegation}

We cope here with the second-best case. 
For a given forward price process $F_t$, the problem of the Principal is now
\begin{align*}
 & V^{\rm P}:= \sup_{Z,\theta} J ^{\rm P}(Z,\theta), \quad \text{where} \quad J^{\rm P}(Z,\theta) := \E\Big[ U_{\rm P}\Big( g(X_T)  - Y_T + \int_0^T \theta_t dF_t \Big)\Big],  \quad g(x):= x \mcdot M x,\\
&dY_t = \frac12 Z_t \mcdot C_a  Z_tdt + Z_t \mcdot \Sigma d\bar W_t^\perp, \quad
dX_t = \big( b + \gamma E_{22} Z_t\big) dt + \Sigma d\bar W_t^\perp,
\end{align*}
where we recall that $C_a $ and $C_p$ are defined by $C_i := \gamma E_{22}+\eta_i A, \quad  i\in \{ a,p\},$ and $\bar{C} := \gamma E_{22} + (\eta_a+\eta_p) A$ and the matrix $M$ has been defined in the integrated case, i.e.  $M := \begin{pmatrix} 0 & \frac12 \\ \frac12 & -\eps \end{pmatrix}.$ Indeed, the incentive mechanism $Y_T$ does not change its form because of the hedging possibility.  The best-response function of the Agent remains the same, namely $\hat \alpha(z) = \gamma z_2$.

The optimisation problem of the market maker who ensures the supply of hedging is still
\begin{align*}
 & V^{\rm M}:= \sup_{\theta} J^{\rm M}(\theta), \quad \text{where} \quad
 J^{\rm M}(\theta):= \E\Big[ U_{\rm M}\Big( \int_0^T \theta_t dF_t \Big)\Big].
 \end{align*}
The main objective of this section is to find an equilibrium in the following sense.

\begin{Definition}
A linear market equilibrium is defined as a quadruplet $(Z^P,\theta^P,\theta^M, F^\ast)$ such that
\b*
&V^{\rm P} = J^{\rm P}(Z^P,\theta^P), \quad V^{\rm M} = J^{\rm M}(\theta^M), \quad \theta^P+\theta^M=0, 
&
\\
&
F^\ast_t = F_0(t) + F_1(t)\mcdot X_t, \quad t\in[0,T]
&
\e*
for some deterministic functions $F_0,F_1$ satisfying $F_0(T) = 0$ and $F_1(T) = (1, \, -\epsilon)^\tp$.
\end{Definition}

The final condition on the functions $F_0$ and $F_1$ garantees that the forward price at maturity $T$ coincides with the spot $P_T = S_T - \eps Q_T$. We also report the dynamics of $F_t$:
\begin{align*}
dF_t =\big( \bar \beta_t + \gamma E_{22} F_1\mcdot Z_t \big) dt + F_1 \mcdot \Sigma d\bar W_t^\perp
=:
\mu(Z_t) dt + F_1 \mcdot \Sigma d\bar W_t^\perp
\end{align*}
where $\bar \beta = \dot F_0 + F_1\mcdot b + \dot F_1 \mcdot X.$
 and we recall that $\phi_1 = F_1\mcdot A F_1>0$ and $\phi = A F_1$.

\hs

We use the same method as in the integrated case.

\hs

\noindent{\bf Producer's problem}: In the context of a linear market equilibrium, the optimisation problem of the Principal reduces to
\b*
 V^{\rm P}
 &=&
 e^{\eta_{p}Y_0}\sup_{Z,\theta} \E\Big[ U_{\rm P}\Big( g_T(X_T)  
 + \int_0^T \Big(\theta_t\big(\bar \beta_t + \gamma E_{22} F_1 \mcdot Z_t \big)- \frac12 Z_t \mcdot C_a  Z_t\Big)dt -  Z_t \mcdot \Sigma d\bar W_t^\perp + \theta_t F_1 \mcdot \Sigma d\bar W_t^\perp \Big], 
\\
&=&
e^{\eta_{p}Y_0}
\sup_{Z,\theta} \E\Big[ U_{\rm P}\Big( g_T(X_T)  
 + \int_0^T \Big(\theta_t \big(\bar \beta_t + \gamma E_{22} F_1 \mcdot Z_t \big)- \frac12 Z_t \mcdot C_a  Z_t\Big)dt +  \Sigma^\tp\big(\theta_t F_1 - Z_t\big) \mcdot  d\bar W_t^\perp
 \Big)\Big].
\e*

\hs

 We have 
\begin{align*}
V(0,x,\ell) :=  V^{\rm P} = &
\sup_{Z,\theta} \E\Big[ U_{\rm P}\Big( g_T(X_T)  - L_T \Big)\Big], \quad
dX_t   = (b + \gamma E_{22} Z_t) dt + \Sigma \mcdot d\bar W^\perp_t \\
dL_t  & = \Big[\frac12 Z_t \mcdot C_a Z_t - \theta_t \mu(Z_t) \Big] dt - \Sigma^\tp (\theta_t F_1 - Z_t) \mcdot d\bar W^\perp_t,\end{align*}
Hence, we have the PDE with again $\phi_1 = F_1 \mcdot AF_1$
\begin{align*}
& 0 =\partial_t V + \frac12 A\mddot D^2_x V + b\mcdot D_xV \\
 & +  \sup_{a,\theta}\Big\{ 
\gamma E_{22} z  \mcdot D_xV + \big[\frac12 z \mcdot C_a z - \theta \mu(z) \big] \big) V_\ell
+\frac12\big[ \phi_1 \theta^2 - 2 (AF_1 \mcdot z) \theta + Az \mcdot z\big] V_{\ell\ell}
-  A(\theta F_1 - z) \mcdot D_xV_\ell
 \Big\}\\
&V(T,x,\ell) = U_{\rm P}(g(x)-\ell).
\end{align*}
We set again $V(t,x,\ell) =:U_{\rm P}(v(t,x)-\ell)$, so that
\begin{align*}
D_xV  = -\eta_p V Dv, \quad D_x^2V = -\eta_pV(D^2v - \eta_p DvDv^\tp), \quad V_\ell = \eta_p V, \quad
D_xV_\ell = -\eta_p^2 V Dv, \quad V_{\ell\ell} = \eta_p^2 V, 
\end{align*}
and get the PDE satisfied by $v$
\begin{align*}
& 0 =\partial_t v + \frac12 A\mddot (D^2v - \eta_p Dv \mcdot Dv) + b\mcdot Dv \\
 & +  \sup_{z,\theta}\Big\{ 
\gamma E_{22} z  \mcdot Dv -  \Big[\frac12 z \mcdot C_a z - \theta \mu(z) \Big]   
-\eta_p \frac12\big[ \phi_1 \theta^2 - 2 (AF_1 \mcdot z) \theta + Az \mcdot z\big]  
-  \eta_p A(\theta F_1 - z) \mcdot Dv
 \Big\}\\
&v(T,x) =  g(x).
\end{align*}
Thus we can rewrite it as
\begin{align*}
& 0 =\partial_t v + \frac12 A\mddot (D^2v - \eta_p Dv \mcdot Dv) + b\mcdot Dv  +  \sup_{z,\theta} L(z,\theta)
\end{align*}
with (recall $\mu(z) = \bar\beta + \gamma E_{22} F_1\mcdot z$)
\begin{align*}
L(z,\theta) & := -\frac12 \eta_p \phi_1 \theta^2 + ( \eta_pAF_1\mcdot (z - Dv) +   \mu(z)) \theta
+ C_p z  \mcdot Dv - \frac12 z \mcdot \bar C z, \\
& = -\frac{1}{2} \eta_p\phi_1 (\theta - \theta^{\rm d}(z))^2 + \frac1{2} \eta_p\phi_1 ( \theta^{\rm d}(z))^2 
+ C_p z  \mcdot Dv - \frac12 z \mcdot \bar C z \\
\text{with} & \quad \theta^{\rm d}(z) := \frac1{\eta_p\phi_1}\big(C_p F_1 \mcdot z  + \bar\beta - \eta_p AF_1 \mcdot Dv)
\end{align*}
Thus,
\begin{align*}
L(z,\theta^{\rm d}(z)) & = C_p z  \mcdot Dv - \frac12 z \mcdot \bar C z + \frac1{2\eta_p\phi_1}(C_p F_1 \mcdot z   + \bar\beta - \eta_p AF_1 \mcdot Dv)^2 \\
& = -\frac12 z \mcdot \big[ \bar C - \frac1{\eta_p\phi_1}C_pF_1 F_1^\tp C_p^\tp  \big] z
+ \big[C_p^\tp Dv + \frac{1}{\eta_p\phi_1}(\bar\beta - \eta_p AF_1 \mcdot Dv) C_p F_1 \big] \mcdot z
+ \frac1{2\eta_p\phi_1}(\bar\beta - \eta_p AF_1 \mcdot Dv)^2 \\
& = -\frac12 (z - \hat z) \mcdot \Gamma (z- \hat z) + \frac12 \hat z \mcdot \Gamma \hat z  + \frac1{2\eta_p\phi_1}(\bar\beta - \eta_p AF_1 \mcdot Dv)^2 \\
\end{align*}
 with
\begin{empheq}[box=\fbox]{align*}
 \Gamma :=  \bar C - \frac1{\eta_p\phi_1}C_pF_1 F_1^\tp C_p^\tp, \quad 
\hat z := \Gamma^{-1}\Big[ C_p^\tp Dv + \frac{1}{\eta_p\phi_1}(\bar\beta - \eta_p AF_1 \mcdot Dv) C_p F_1\Big].
\end{empheq}

We assume that the eigen values of $\Gamma$ are strictly nonnegative so that $\hat z$ is a maximiser of $L$. The optimiser $\hat z$ can be written
\begin{align*}
\hat z & =  \frac{\Gamma^{-1}}{\eta_p\phi_1}\Big[  (\dot F_0 + F_1 \mcdot b)C_p F_1  + C_p F_1 \dot F_1 \mcdot x  +  \eta_p \big[\phi_1 C_p^\tp -   C_p F_1 F_1^\tp A\big] Dv \Big] =: \hat z_0 + \hat z_1 x + \hat z_2 Dv, \\
\end{align*}
 with
\begin{empheq}[box=\fbox]{align*}
\hat z_0 := \frac{ \Gamma^{-1}}{\eta_p\phi_1}(\dot F_0 + F_1 \mcdot b)C_p F_1, \quad
 \hat z_1 := \frac{\Gamma^{-1}}{\eta_p\phi_1} C_pF_1\dot F_1^\tp, \quad
\hat z_2 := - \frac{\Gamma^{-1}}{\phi_1}\big[\phi_1 C_p^\tp -   C_p F_1 F_1^\tp A\big].
\end{empheq}

It holds that
\begin{align*}
\frac12 \hat z \mcdot \Gamma \hat z & = \frac12 (\hat z_0 + \hat z_1 x + \hat z_2 Dv) \mcdot (\Gamma \hat z_0 + \Gamma\hat z_1 x + \Gamma\hat z_2 Dv) \\
& = \frac12\Big\{ 
\hat z_0 \mcdot \Gamma \hat z_0 
+ \big[ \hat z_1^\tp (\Gamma + \Gamma^\tp) \hat z_0 \big] \mcdot x 
+ x \mcdot \hat z_1^\tp \Gamma \hat z_1 x 
 + \big[ \hat z_2^\tp (\Gamma + \Gamma^\tp) \hat z_0 + \hat z_2^\tp (\Gamma + \Gamma^\tp) \hat z_1 x  \big] \cdot Dv + Dv \mcdot \hat z_2^\tp \Gamma \hat z_2 Dv
\Big\}
\end{align*}
Moreover, we have
\begin{align*}
(\dot F_0 +  F_1 \mcdot b + \dot F_1 \mcdot x - \eta_p AF_1 \mcdot Dv)^2 &  = 
(\dot F_0 +  F_1 \mcdot b)^2 + 2 (\dot F_0 +  F_1 \mcdot b) \dot F_1 \mcdot x + x \mcdot \dot F_1\dot  F_1^\tp x \\
& - 2 \eta_p \big[ (\dot F_0 +  F_1 \mcdot b) AF_1 + AF_1 \dot F_1^\tp x \big] \mcdot Dv \\
& + \eta_p^2 Dv \mcdot AF_1 F_1^\tp A Dv
\end{align*}

The PDE of $v$ reduces then to
\begin{align*}
& 0 =\partial_t v + \frac12 A\mddot (D^2v - \eta_p Dv \mcdot Dv) + b\mcdot Dv  
+  \frac12 \hat z \mcdot \Gamma \hat z  + \frac1{2\eta_p\phi_1}(\bar\beta - \eta_p AF_1 \mcdot Dv)^2, \quad v(T,x) = g(x).
\end{align*}

It can be written in the form
\begin{align*}
&0=\partial_t v +  \big( B_0^{\rm d} + B_1^{\rm d} x\big) \mcdot Dv + \frac12 A\mddot D^2v  + \frac12  Dv \mcdot \Lambda_{\rm d} Dv  + H_0^{\rm d} + H_1^{\rm d} \mcdot x +  x \mcdot H_2^{\rm d} x, \quad
 v(T,x) = x \mcdot M x,
\end{align*}
with
\begin{empheq}[box=\fbox]{align*}
& B_0^{\rm d} = b + \frac12 \big[ \hat z_2^\tp (\Gamma + \Gamma^\tp) \hat z_0  \big] 
- \frac1\phi_1   (\dot F_0 +  F_1 \mcdot b) AF_1, \quad
 B_1^{\rm d} =  \frac12[  \hat z_2^\tp (\Gamma + \Gamma^\tp) \hat z_1 ] - \frac1{\phi_1} AF_1 \dot F_1^\tp\\
& H_0^{\rm d} =  \frac12 \hat z_0 \mcdot \Gamma \hat z_0  + \frac1{2\eta_p\phi_1} (\dot F_0 +  F_1 \mcdot b)^2, \quad
H_1^{\rm d} = \frac12 \big[ \hat z_1^\tp (\Gamma + \Gamma^\tp) \hat z_0 \big] + \frac1{\eta_p\phi_1} (\dot F_0 +  F_1 \mcdot b) \dot F_1 \\
& H_2^{\rm d} = \frac12\hat z_1^\tp \Gamma \hat z_1 + \frac{1}{2\eta_p\phi_1}  \dot F_1 \dot F_1^\tp, \quad
\Lambda_{\rm d} =  -\eta_p A + \hat z_2^\tp \Gamma \hat z_2 + \frac{\eta_p}{\phi_1} AF_1 F_1^\tp A.
\end{empheq}
\hs

Hence, using the guessed form for $v$ of $v = K_0^{\rm d} + K_1^{\rm d} \mcdot x + x \mcdot K_2^{\rm d} x$, the generic form of the Riccati's system provides the ODes for the coefficients $K_i^{\rm d}$, $i=0,1,2$. Besides, using the fact that $Dv = K_1^{\rm d} + 2 K_2^{\rm d} x$, we can rewrite $\hat z(t,x)$ as
\begin{align*}
\hat z = \hat z_0 + \hat z_1 x + \hat z_2 (K_1^{\rm d} + 2  K_2^{\rm d} x) = \hat z_0 + \hat z_2 K_1^{\rm d} + (\hat z_1 + 2 \hat z_2 K_2^{\rm d}) x. 
\end{align*}

\hs

\noindent{\bf Representative agent of the market}: The problem is
\begin{align*}
 & V^m(0,x,\ell):= \sup_{\theta}  \E\big[ U_{\rm M}\big( L_T \big)\big], \quad
dX_t = \big( b + \gamma E_{22} \hat z(t,X_t)\big) dt + \Sigma d\bar W_t^\perp, \\
dL_t  & =   \theta_t \mu(\hat z(t,X_t)) dt +  \theta_t  F_1 \mcdot \Sigma d\bar W^\perp_t, \quad
dF_t =  \mu(Z_t) dt + F_1 \mcdot \Sigma d\bar W_t^\perp
 \end{align*}
where we recall that $\bar \beta = \dot F_0 + F_1\mcdot b + \dot F_1 \mcdot X$, $\mu(z) = \bar\beta + \gamma E_{22} F_1 \mcdot z$ and  $\phi_1 = F_1\mcdot A F_1>0$. Besides, $Z_t = \hat z(t,X_t)$ with the function $z$ given in the resolution of the producer's problem.  

\hs

The PDE satisfied by $V^m$ is given by
\begin{align*}
& 0 =\partial_t V^m + \frac12 A\mddot D^2_x V^m + (b + \gamma E_{22} \hat z )\mcdot D_xV^m    
 +  \sup_{\theta}\Big\{   \theta \mu(\hat z) V_\ell
+\frac12  \phi_1 \theta^2  V_{\ell\ell}
+  \theta A F_1 \mcdot D_xV_\ell \Big\}\\
&V^m(T,x,\ell) = U_{\rm M}(0).
\end{align*}
We set again $V^m(t,x,\ell) =:U_{\rm M}(v^m(t,x) + \ell)$, so that
\begin{align*}
&D_xV^m  = -\eta_m V^m Dv^m, \quad D_x^2V^m  = -\eta_m V^m(D^2v^m - \eta_p Dv^m\mcdot Dv^m), \\
& V^m_\ell = -\eta_m V^m,  \quad
D_xV^m_\ell = \eta_m^2 V^m Dv^m, \quad V^m_{\ell\ell} = \eta_m^2 V^m, 
\end{align*}
and get the PDE satisfied by $v^m$
\begin{align*}
& 0 =\partial_t v^m + \frac12 A\mddot (D^2v^m - \eta_p Dv^m\mcdot Dv^m) + (b + \gamma E_{22} z )\mcdot Dv^m    
 +  \sup_{\theta}\Big\{    (\mu(z) -  \eta_m  A F_1 \mcdot Dv^m) \theta
- \frac12 \eta_m \phi_1 \theta^2   
  \Big\}\\
&v^m(T,x) = 0.
\end{align*}
Thus, we have
\begin{align*}
& 0 =\partial_t v^m + \frac12 A\mddot (D^2v^m - \eta_m Dv^m\mcdot Dv^m) + (b + \gamma E_{22} z )\mcdot Dv^m    
 +  \frac12 \eta_m \phi_1 \theta_{\rm m,i}^2, \\
 \text{with}  & \quad 
 \theta_{\rm m,d} := \frac1{\eta_m\phi_1} [\mu(z) -  \eta_m  A F_1 \mcdot Dv^m].
\end{align*}
Hence, using the expression of $\hat z$ as $\hat z = \hat z_0 + \hat z_2 K_1^{\rm d} + (\hat z_1 + 2 \hat z_2 K_2^{\rm d}) x$, 
we have
\begin{align*}
   \theta_{\rm m,d} & = \frac1{\eta_m\phi_1} \big[\dot F_0 + F_1 \mcdot b + \gamma E_{22} F_1 \mcdot (\hat z_0 + \hat z_2 K_1^{\rm d}) + (\dot F_1 + \gamma (\hat z_1 + 2 \hat z_2 K_2^{\rm d})^\tp E_{22} F_1 )\mcdot x   - \eta_m AF_1 \mcdot Dv^m \big],  \\
   & =: \frac1{\eta_m\phi_1} \big[ \theta^0_{\rm m} + \theta^1_{\rm m} \mcdot x - \eta_m AF_1 \mcdot Dv^m \big]
\end{align*}
Thus, the PDE can be written in the generic form
\begin{align*}
&0=\partial_t v^m +  \big( B_0^{\rm m,d} + B_1^{\rm m,d} x\big) \mcdot Dv^m + \frac12 A\mddot D^2v  + \frac12  Dv \mcdot \Lambda_{\rm m,d} Dv  + H_0^{\rm m,d} + H_1^{\rm m,d} \mcdot x +  x \mcdot H_2^{\rm m,d} x, \quad
 v^m(T,x) = 0,
\end{align*}
with
\begin{empheq}[box=\fbox]{align*}
& B_0^{\rm m,d} = b + \gamma E_{22}[\hat z_0 + \hat z_2 K_1^{\rm d}] - \frac{ \theta_{\rm m}^0}{\phi_1}  AF_1, \quad
 B_1^{\rm m,d} = \gamma E_{22}(\hat z_1 + 2 \hat z_2 K_2^{\rm d}) - \frac{1}{\phi_1}  A F_1( \theta_{\rm m}^1)^\tp, \\
 & H_0^m = \frac{(\theta_{\rm m}^0)^2}{2\eta_m\phi_1}, \quad
 H_1^m = \frac{1}{\phi_1}\theta_{\rm m}^0 \theta_{\rm m}^1, \quad
 H_2^m = \frac{1}{2\eta_m\phi_1}\theta_{\rm m}^1 (\theta_{\rm m}^1)^\tp, \\
 & \Lambda_{\rm m,d} = -\eta_m A + \frac{\eta_m}{\phi_1} AF_1 F_1^\tp A
\end{empheq}
\hs

Hence, using the guessed form for $v^m$ of $v^m = K_0^{\rm m,d} + K_1^{\rm m,d} \mcdot x + x \mcdot K_2^{\rm m,d} x$, the generic form of the Riccati's system provides the ODes for the coefficients $K_i^{\rm m,d}$, $i=0,1,2$. 

\hs

\noindent{\bf Equilibrium} We write

The clearing condition $\theta_{\rm m,i} + \theta_{\rm d} =0$ writes
\begin{align*}
 \frac1{\eta_m\phi_1} [ \bar\beta + \gamma E_{22} F_1 \mcdot z -  \eta_m  A F_1 \mcdot Dv^m]
+
\frac1{\eta_p\phi_1}\big(C_p F_1 \mcdot z  + \bar\beta - \eta_p AF_1 \mcdot Dv)
= 0
\end{align*}
thus ($C_p = \gamma E_{22} + \eta_p A$)
\begin{align*}
&\big[ (\eta_p + \eta_m) \gamma E_{22} F_1 + \eta_m\eta_p A F_1\big] \mcdot z + (\eta_p+\eta_m) \bar\beta - \eta_p\eta_m AF_1 \mcdot (Dv + Dv^m) = 0 \\
& \big[\gamma E_{22} F_1 + \bar\eta A F_1\big] \mcdot z +  \dot F_0 + F_1 \mcdot b + \dot F_1 \mcdot x - \bar\eta AF_1 \mcdot (Dv + Dv^m) = 0
\end{align*}
Substitution with the expression of $\hat z$, $Dv$ and $Dv^m$ as functions of $x$ provides
\begin{align*}
& 0 = \dot F_1 + (\hat z_1 + 2 \hat z_2 K_2^{\rm d})^\tp \big[\gamma E_{22} F_1 + \bar\eta A F_1\big] 
- 2 \bar\eta (K_2^{\rm m,d} + K_2^{\rm d})^\tp AF_1  \\
& 0 = \dot F_0 + F_1 \mcdot b + \big[\gamma E_{22} F_1 + \bar\eta A F_1\big] \mcdot \big(\hat z_0 + \hat z_2 K_1^{\rm d}\big) - \bar\eta AF_1 \mcdot (K_1^{\rm m,d} + K_1^{\rm d}).
\end{align*}
with a little simplification
\begin{align*}  
& 0 = \dot F_1 + (\hat z_1 + 2 \hat z_2 K_2^{\rm d})^\tp \bar C F_1  
- 2 \bar\eta (K_2^{\rm m,d} + K_2^{\rm d})^\tp AF_1  \\
& 0 = \dot F_0 + F_1 \mcdot b + \bar C F_1  \mcdot \big(\hat z_0 + \hat z_2 K_1^{\rm d}\big) - \bar\eta AF_1 \mcdot (K_1^{\rm m,d} + K_1^{\rm d}).
\end{align*}

But, recall that $\hat z_0$ depends on $\dot F_0$ as well as $\hat z_1$ on $\dot F_1$ in
\begin{align*}  
\hat z_0 & = \frac{ \Gamma^{-1}}{\eta_p\phi_1}(\dot F_0 + F_1 \mcdot b) C_p F_1 
=
\frac{ \dot F_0 }{\eta_p\phi_1} \Gamma^{-1} C_p F_1 
+ \frac{F_1 \mcdot b }{\eta_p\phi_1} \Gamma^{-1} C_p F_1 \\
\hat z_1 & = \frac{\Gamma^{-1}}{\eta_p\phi_1} C_pF_1\dot F_1^\tp
\end{align*}
and thus, we have
\begin{align*}  
& 0 = m_0 \dot F_0 + m_0 F_1 \mcdot b + \bar C F_1  \mcdot  \hat z_2 K_1^{\rm d} - \bar\eta AF_1 \mcdot (K_1^{\rm m,d} + K_1^{\rm d}), 
\quad m_0 :=  1 + \frac{1}{\eta_p \phi_1} \bar C F_1 \mcdot \Gamma^{-1} C_p F_1. 
\end{align*}
and
\begin{align*}  
& 0 = \dot F_1  + (\hat z_1)^\tp \bar C F_1  +   2 (\hat z_2 K_2^{\rm d})^\tp \bar C F_1  
- 2 \bar\eta (K_2^{\rm m,d} + K_2^{\rm d})^\tp AF_1, \\
& 0 =  m_1 \dot F_1 
+   2 (\hat z_2 K_2^{\rm d})^\tp \bar C F_1  
- 2 \bar\eta (K_2^{\rm m,d} + K_2^{\rm d})^\tp AF_1, \quad m_1 :=  1 + \frac1{\eta_p \phi_1} F_1^\tp C_p^\tp (\Gamma^{-1})^\tp \bar C F_1  
\end{align*}
Thus, the ODEs for $F_0$ and $F_1$ are given by
\begin{align*}  
& 0 =  \dot F_0 +  F_1 \mcdot b + \frac1{m_0} \bar C F_1  \mcdot  \hat z_2 K_1^{\rm d} - \frac{\bar\eta}{m_0} AF_1 \mcdot (K_1^{\rm m,d} + K_1^{\rm d})  \\
& 0 = \dot F_1  +   \frac{2}{m_1} (\hat z_2 K_2^{\rm d})^\tp \bar C F_1   - 2 \frac{\bar\eta}{m_1} (K_2^{\rm m,d} + K_2^{\rm d})^\tp AF_1.
\end{align*}

\end{document}